\documentclass[11.0pt,reqno]{amsart}
\usepackage[utf8]{inputenc}
\usepackage{mathtools}
\usepackage{ytableau}
\usepackage{amsmath,amssymb}
\usepackage{tikz}
\usetikzlibrary{decorations.pathreplacing}
\usepackage{wrapfig}
\usepackage{lscape}
\usepackage{rotating}
\usepackage{epstopdf}
\usepackage[colorlinks,linkcolor=blue]{hyperref}
\usepackage{lipsum}
\usepackage{emptypage}

\title[Braided finite automata ]
{Braided finite automata and representation theory}
\author[Anastasia Doikou]{Anastasia Doikou}

\address[Anastasia Doikou] {Department of Mathematics, Heriot-Watt University,
Edinburgh EH14 4AS $\&$ Maxwell Institute for Mathematical Sciences, Edinburgh EH8 9BT, UK}
\email{a.doikou@hw.ac.uk}

  \usepackage{latexsym} 
  \usepackage[all]{xy}
  \usepackage{amsfonts} 
  \usepackage{amsthm} 
  \usepackage{amsmath} 
  \usepackage{amssymb}
  \usepackage{pifont}  
  \usepackage{enumerate}
  \usepackage{dcolumn}
  \usepackage{comment}
  \usepackage{hyperref}  
\usepackage[english]{babel}
  \usepackage{amsfonts} 
\usepackage[utf8]{inputenc}
\usepackage{graphicx}
\usepackage{subcaption}
\usepackage[colorinlistoftodos]{todonotes}
\usepackage{indentfirst}
\usepackage{pifont}  
  \usepackage{enumerate}
  \usepackage{dcolumn}
  \usepackage{comment}
\usepackage[capitalise]{cleveref}
  \usepackage[capitalise]{cleveref}
\usepackage{multirow}
\usepackage{tikz}
\usetikzlibrary{automata, positioning, arrows}

\usetikzlibrary{trees,arrows}

 \newcolumntype{2}{D{.}{}{2.0}}
  \xyoption{2cell}

\DeclareMathOperator{\EEnd}{End}
\newcommand{\hiddenpower}[2] { \ifnum \numexpr#2=1 #1 \else #1^#2 \fi }
\numberwithin{equation}{section}

\def\be{\begin{equation}}
\def\ee{\end{equation}}
\def\ba{\begin{eqnarray}}
\def\ea{\end{eqnarray}}

\def\non{\nonumber}

\newcommand{\cal}{\mathcal}
\newcommand{\End}{\mathrm{End}}

\usepackage{xparse}
\ExplSyntaxOn
\newcounter{diff_order}
\newcounter{diff_power}

\newcommand{\rawdiff}[3]
{
	\setcounter{diff_order}{0}
	\clist_map_inline:nn{#3}{\stepcounter{diff_order}}
	
	\frac{\hiddenpower{#1}{\thediff_order} #2}
	{
		\def\old_var{DefaultValue}
		\setcounter{diff_power}{0}
		
		\clist_map_inline:nn{#3}
		{
			\def\new_var{##1}
			\ifnum \thediff_power=0
				\stepcounter{diff_power}
			\else
				\tl_if_eq:NNTF \new_var \old_var
				{\stepcounter{diff_power}}
				{
					#1 \hiddenpower{\old_var}{\thediff_power}
					\setcounter{diff_power}{1}
				}
			\fi

			\def\old_var{##1}
		}
		
		#1 \hiddenpower{\old_var}{\thediff_power}
	}
}
\def\Label#1{\label{#1}\ifmmode\llap{[#1] }\else 
  \marginpar{\smash{\hbox{\tiny [#1]}}}\fi} 
  \def\Label{\label} 

\ExplSyntaxOff

\newlength{\bibitemsep}
\newlength{\bibparskip}
\let\oldthebibliography\thebibliography
\renewcommand\thebibliography[1]{%
  \oldthebibliography{#1}%
  \setlength{\parskip}{\bibitemsep}%
  \setlength{\itemsep}{\bibparskip}%
}

\newtheorem{thm}{Theorem}[section]
\newtheorem{lemma}[thm]{Lemma}

\newtheorem{pro}[thm]{Proposition}
\newtheorem{defn}[thm]{Definition}

\newtheorem{rem}[thm]{Remark}
\newtheorem{exa}[thm]{Example}
\newtheorem{conj}[thm]{Conjecture}

\newcommand{\id}{\operatorname{id}}

\newcommand{\red}[1]{\textcolor{red}{#1}}
\newcommand{\blue}[1]{\textcolor{blue}{#1}}

\newenvironment{wideregion}[1][0mm]
    {\vspace{#1}\list{}{\leftmargin-2cm\rightmargin\leftmargin}\item\relax}
    {\endlist}

\newenvironment{widegather }{\wideregion[-9mm]\gather}{\endgather\endwideregion}

\begin{document}
\tikzstyle{level 1}=[level distance=22mm, sibling distance=50mm]
\tikzstyle{level 2}=[level distance=22mm, sibling distance=22mm]
\tikzstyle{level 3}=[level distance=20mm]






 



 \begin{abstract} 
We introduce classical and non-deterministic finite automata associated with representations of the braid group. After briefly reviewing basic definitions on finite automata, Coxeter's groups and the associated word problem, we turn to the Artin presentation of the braid group and its quotients. We present various representations of the braid group as deterministic or non-deterministic finite state automata and discuss connections with $q$-Dicke states, as well as Lusztig and crystal bases. We propose the study of the eigenvalue problem of the $\mathfrak{U}_q(\mathfrak{gl}_n)$ invariant spin-chain like ``Hamiltonian'' as a systematic means for constructing canonical bases for irreducible representations of $\mathfrak{U}_q(\mathfrak{gl}_n).$  This is explicitly proven for the algebra $\mathfrak{U}_q(\mathfrak{gl}_2).$ Special braid representations associated with self-distributive structures are also studied as finite automata. These finite state automata organize clusters of eigenstates of these braid representations.
\end{abstract}
\maketitle


\tableofcontents

\section{Introduction} 
\noindent The aim of this study is to utilize specific types of finite-state automata called braided finite automata in order to study irreducible representations of quantum algebras \cite{FRT, Drinfeld, Jimbo}, associated with certain representations of the braid group.  Finite automata are in general mathematical models that describe computational ``machines''. They were first studied in the 50's by Kleene who also found significant applications to computer theory \cite{Kleene} (for a concise pedagogical review on the subject, see for instance \cite{Lawson}). A finite automaton consists of a finite number of abstract states, an input alphabet and transition functions. There are two basic types of finite automata, the deterministic or classical (combinatorial) and the non-deterministic \cite{RabinScott} automata. Two fundamental special cases of non-deterministic automata are the probabilistic \cite{Rabin} and the quantum automata \cite{quantum1, quantum2} (see also \cite{review} on a detailed account on both probabilistic and quantum automata). We consider in this study both deterministic and non-deterministic automata and we apply the idea of ``linearization'' on sets, in order to map abstract finite automata to finite vector spaces $V_n$ with dimension equal to the cardinality $n$ of the set of states (in this manuscript $V_n$ is either ${\mathbb C}^n$ or ${\mathbb R}^n$). Specifically, we map the abstract states to basis vectors in $V_n$ and the transition functions to $n\times n$ matrices, called transition matrices. To describe then combinatorial or classical automata it is enough to consider the elements of a basis of the corresponding finite dimensional vector space, whereas in order 
to describe non-deterministic automata  we extend our framework to the full vector space.
We note that throughout this manuscript the characterizations deterministic, classical, set-theoretic and combinatorial are used equivalently.
Specifically, we use the name combinatorial because the matrices associated with classical automata are combinatorial (a precise definition is given later, see Definition \ref{combmat}).

As noted our main objective is to employ finite state automata to study finite irreducible representations of certain quantum groups.
Quantum algebras (or quantum groups) are special cases of Hopf algebras introduced by Jimbo and Drinfel'd \cite{Drinfeld, Drinfeld2, Jimbo} independently and may be seen as deformations of the usual Lie algebras or their infinite dimensional extensions, the Kac-Moody algebras \cite{Kac}. From the point of view of representation theory Lusztig \cite{Lusztig} introduced canonical bases of such quantum groups using both algebraic and geometrical considerations, whereas Kashiwara \cite{Kashi} showed independently that  modules of quantum groups have ``crystal'' bases with important combinatorial properties (see also a recent review on crystals \cite{crystal}). Numerous studies also explore the eigenvalue problem of periodic quantum spin chain ``Hamiltonians'' and Bethe ansatz techniques \cite{Bethe1, Bethe2}, especially in the thermodynamic limit, in connection with representation theory, combinatorics and cellular automata (see for instance \cite{Paths, KirRes1, KirRes2, Kir3, Kuniba0, Kuniba1, Kuniba2}). In this article we study the eigenvalue problem of open finite quantum spin chain Hamiltonians \cite{Sklyanin}, which are invariant under the action of the said quantum groups. The use of finite automata theory facilitates such a study providing the general structure of eigenstates. Furthermore, we prove that degenerate  eigenstates of such Hamiltonians constitute  bases for irreducible representations of the associated quantum algebra. 

The material and the key results presented in this article are outlined as follows.

The article synthesises three distinct research themes: finite automata theory, quantum algebras and associated braid group representations, and the spectral decomposition of special integrable Hamiltonians, which are specific linear combinations of all length-one words of the braid group. Consequently, it is addressed to a general mathematical physics audience, algebraically inclined researchers, and those interested in quantum computing applications. Given these diverse target areas, the introductory review material presented in Sections 2 - 5, although extended, is necessary to help readers from different fields understand the core results in Sections 6 - 7.
Specifically, in Sections 2 and 3 we review basic ideas about finite automata, alphabets, words, languages as well as the definitions of deterministic and non-deterministic automata (see also \cite{Lawson} and references therein). A brief description of probabilistic and quantum automata is also presented.  Various simple examples are presented throughout these sections.
In Section 4 we review basic definitions on Coxeter groups and recall the notion of weak order of sets and the associated word problem, which will be useful for the rest of the manuscript (see also \cite{Comb-book}).
In Section 5, and specifically in Subsection 5.1 we recall the definition of Artin braid groups and Hecke algebras and we also introduce the so called ``shuffle'' element of the Hecke algebra that yields all possible reduced words of the Hecke algebra in line with Matsumoto's solution of the word problem for Coxeter groups \cite{Matsu} (see also for instance \cite{Comb-book} and references therein). In Subsection 5.2 we recall the definition of the algebra ${\mathfrak U}_q(\mathfrak{gl}_n),$ \cite{Drinfeld, Drinfeld2, Jimbo, Jimbo2} and briefly discuss the duality between the Hecke algebra ${\cal H}_N(q)$ and ${\mathfrak U}_q(\mathfrak{gl}_n).$  A short review on Young tableaux and the Schur-Weyl duality is presented in Subsection 5.3 (see for instance \cite{FultonHarris, Fulton} for a detailed exposition on these subjects).

In Section 6 and in particular in Subsection 6.1 we introduce specific braided automata and show using tensor representations of the ${\cal H}_N(q)$ Hecke algebra how these automata act on certain tensor product states. 
We recall the ``shuffle'' operator, which is an element of $\EEnd(V_n^{\otimes N}),$ and prove that it yields all possible permutations of any state  $ \hat e_{i_1} \otimes \hat e_{i_2} \otimes \ldots \otimes \hat e_{i_N},$ where $i_1 \leq i_2 \ldots \leq i_N$ and $\{\hat e_j\},$  $j \in \{1,2 \ldots, n\}$ 
is the standard basis of $V_n$ (Theorem \ref{shuffle1}). The action of two special cases of the shuffle operator, called the $q$-symmetrizer and $q$-antisymmetrizer, on specific tensor product states yields all $q$-(anti)symmetric states. The $q$-(anti)symmetric states are elements of a $q$-deformed Fock space, which is defined in Subsection 6.1 (see relevant construction for instance in \cite{MisMiw, Lechner}, see also connection to Nichols algebras in \cite{Andru}). The $q$-symmetric states in particular are also known as qudit $q$-Dicke states in the framework of quantum computing and quantum entanglement \cite{Nepo1, Nepo2}. These are $q$-deformed, high rank generalizations of the qubit Dicke state first introduced in \cite{Dicke1}. In Subsection 6.2 we prove that the $q$-symmetric states form an orthogonal basis for an irreducible representation of $\mathfrak{U}_q(\mathfrak{gl}_n).$ These results are presented in Theorem \ref{basicpro} and Proposition \ref{basiclemma}. Finite irreducible representations of $\mathfrak{U}_q(\mathfrak{gl}_n)$ can also be easily interpreted as finite automata. The crystal limit ($q\to 0$) \cite{crystal, Kashi} is also briefly discussed.

In Section 7 we study the eigenvalue problem for the $\mathfrak{U}_q(\mathfrak{gl}_n)$ invariant quantum spin chain Hamiltonian \cite{Sklyanin, Pasquier, Kulish, MeNe, DoiNep, Doikous}. We first prove that the $q$-symmetric states are all eigenstates of the open spin chain Hamiltonian with the same eigenvalue (Proposition \ref{sypro}). We claim that sets of eigenstates of the Hamiltonian form orthogonal bases of irreducible representations of $\mathfrak{U}_q(\mathfrak{gl}_n).$ We note that the $\mathfrak{U}_q(\mathfrak{gl}_n)$ invariant Hamiltonian is nothing but the sum of all length-one words of the Hecke algebra ${\cal H}_N(q).$ 
 In general, we claim that the decomposition of the space $V_n^{\otimes N},$ on which the Hamiltonian acts, in terms of eigenspaces is given as follows:
\begin{equation}
    V_n^{\otimes N} = \bigoplus_{\lambda \vdash N} m_{\lambda}V_n^{(\Lambda_{\lambda})}, \label{deco1}
\end{equation}
where $\Lambda_{\lambda}$ are the Hamiltonian's eigenvalues that correspond to a $\lambda$-shaped Young-tableau, $V_n^{(\Lambda_{\lambda})}$ are the corresponding eigenspaces, $\mbox{dim}V_n^{(\Lambda_{\lambda})} = d_{\lambda,n}.$ Also, $m_\lambda$ is the dimension of the $\lambda$-shaped standard Young tableau and $d_{\lambda,n}$ is the dimension of the $\lambda$-shaped semi-standard Young-Tableau. As already noted a brief review of Young tableaux and related definitions are presented in Subsection 5.3. The decomposition  (\ref{deco1}) is a general claim for the algebra $\mathfrak{U}_q(\mathfrak{gl}_n),$ but we explicitly prove this statement for $\mathfrak{U}_q(\mathfrak{gl}_2),$ see  Proposition \ref{systema} and Theorem \ref{thetheorem}. Crucially, these findings establish an explicit link between pure representation theory of quantum algebras and the spectral decomposition of spin-chain Hamiltonians. Instead of using Bethe ansatz techniques, our results in Section 7 rely primarily on combinatorial and linear algebraic arguments.

In Section 8 we focus on non-involutive combinatorial or set-theoretic solutions of the braid equation. The word problem associated with braid groups is only solved when a Hecke type or involution condition also holds. Therefore, studying the eigenvalue problem of open, quantum spin-chain-like Hamiltonians for non-involutive braid solutions represents a completely new area of research. To our knowledge, no systematic techniques—such as the Bethe ansatz or highest-weight arguments—are currently available for this study. We first introduce racks and quandles \cite{Jo82, Matv, Deho}, which are algebraic structures satisfying a self-distributivity condition used to derive non-involutive, invertible combinatorial solutions. We then focus on a specific self-distributive structure, the dihedral quandle, to analyze the eigenvalue problem of its corresponding braid equation solution. To organize the eigenstates of this solution, we define rack and quandle automata, focusing specifically on the dihedral quandle automaton. Finally, we discuss the centralizers of these rack and quandle braid solutions and present preliminary results on their finite representations.

\section{Deterministic or combinatorial finite automata}
\noindent In this section, we review the foundational definitions and standard notation governing finite automata. We limit our discussion to the core concepts necessary for our subsequent analysis; for a comprehensive treatment of advanced core topics—such as the pumping lemma or the explicit proof of Kleene's Theorem (which is only stated herein)—the reader is referred to \cite{Lawson} and the references therein.

We first introduce the {\it alphabet} $\Sigma$, which is a finite non-empty set. The elements of $\Sigma$ are called \emph{letters}\index{letter}, and a finite sequence of letters is called a \emph{word}\index{word} or {\it string}.  Words are created by concatenating letters, for example, if $\Sigma=\big \{a,b,c\big \}$, then $aabaca$ is a word over $\Sigma$.  The \emph{empty sequence}\index{word!empty} is considered a word, and is denoted $\epsilon.$  The set of all words over $\Sigma$ is denoted $\Sigma^*$, and the set of all non-empty words is denoted $\Sigma^+$. The \emph{length}\index{length} of the word $w$, that is, the number of letters in $w$, is denoted $|w|$.
If $u,v\in \Sigma^*$ then we can form a new word $uv$ by concatenating the two sequences. Concatenation of words is obviously an associative and non-commutative operation on $\Sigma^*$ (i.e. order matters!), also
\begin{align*}
|uv| &= |u|+|v|,\textrm{ and}\\
u\,\epsilon &= \epsilon\,u = u.
\end{align*}
Any subset of $\Sigma^*$ is called a \emph{language}\index{language} over $\Sigma$. Also, $\Sigma^*$ is a free monoid on $\Sigma,$ whereas $\Sigma^+$ is a free semigroup.

We define a collection of basic operations on languages over $\Sigma$. 
The product operation on words can be naturally extended to languages: if $K$ and $L$ are languages over $\Sigma$, we define their \emph{concatenation product}\index{concatenation product} $KL$ to be the set of all products of a word in $K$ followed by a word in $L$: 
$KL = \big \{uv \mid u\in K\textrm{ and }v\in L\big \}.$
 The union and intersection of two languages $K, L$ over $\Sigma$ are defined as  $ K\cup  L = \big \{x| x \in K ~ \mbox{or} ~ x \in L\big \}$ and 
  $ K\cap  L = \big \{x| x \in K ~ \mbox{and} ~ x \in L\big \}$ respectively. The complement of language $L$ is defined as $L^c = \big \{x| x \notin L \big\}.$
 We define for any language $L$ the power notation: $L^0 = \big \{\epsilon\big \}$ and $L^{n+1} = L^n \cdot L.$ For $n> 0$ the language $L^n$ consists of all strings $u$ of the form $u= w_1 w_2 \ldots w_n,$ where $w_i\in L.$
We finally define the \emph{Kleene star}\index{star (Kleene)} of a language $L$ denoted $L^*$ as $L^* = \underset{n\geq 0}{\bigcup} L^n,$ we also define $L^+ =\underset{n\geq 1}{\bigcup} L^n.$ 

\begin{rem} \label{tree} {\bf (Left tree order).} It is useful to have a standard way to list strings over an alphabet. This can be achieved using the so-called tree order on $\Sigma^*,$ also known as the {\it length-plus lexicographic order} (see also, for instance, \cite{Lawson}). 
Let $\Sigma = \big \{a_1, a_2, \ldots, a_n \big \}$ be an alphabet. Choose a fixed order for the elements of the alphabet, e.g. the standard ordering: $a_1 <a_2 < \ldots < a_n,$ or any other order can be chosen (all possible permutations of the elements of the alphabet). If a non-standard ordering is chosen, it should be stated. 
We may now grow a tree on $\Sigma^*,$ whose root is $\epsilon$ and whose vertices are labelled by elements of $\Sigma^*$ according to the following rules: if $w$ is a vertex, the vertices growing from $w$ are $a_1w, a_2 w, \ldots, a_n w.$ The tree order on $\Sigma^*$ is obtained as follows: 
\begin{eqnarray}
x<y ~~\mbox{if} ~~ |x| < |y|, ~~\mbox{or} ~~|x| = |y|~~ \mbox{and the string $x$ is located to the left of the string $y$.} \nonumber   \end{eqnarray}
This ordering means that a string precedes all strictly longer strings, while all strings of the same length are listed lexicographically, that is, they are listed in a dictionary (or {\it lexicon} in Greek) based on the ordering of the corresponding alphabet.
\end{rem}

\begin{exa} \label{extree}
 We consider a simple example of an alphabet and construct the associated tree order. Let $\Sigma = \big \{a,b,c\big \}$  and consider the standard order $a<b<c.$

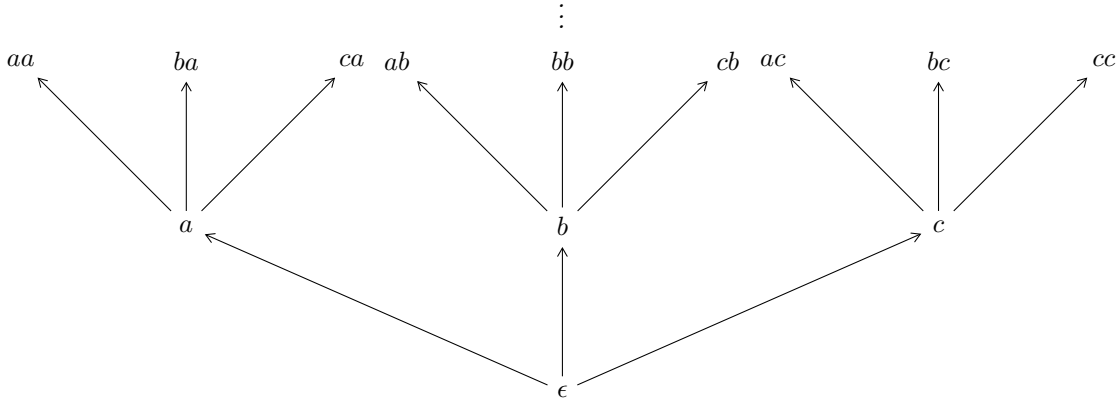
\begin{center}
$$\vdots$$
\begin{tikzpicture}[shorten >=0.8pt, node distance=2cm, grow=up,->,>=angle 60]
\begin{scope}[yshift=-4cm]
  \node {$\epsilon$}
   child {node {$c$}
    child {node {$cc$}
      }      child {node {$bc$}
      }
      child {node {$ac$}
      }
    }    child {node {$b$}
    child {node {$cb$}
      }        child {node {$bb$}
      }
      child {node {$ab$}
      }    
    }
    child {node {$a$}
    child {node {$ca$}
      }      child {node {$ba$}
      }
      child {node {$aa$}
      }
    };
\end{scope}
\end{tikzpicture}
\end{center}
We can proceed to construct strings of length three and so on. In the tree diagram above the string ordering reads as follows: $\epsilon, a, b, c, aa,ba,ca, ab,bb,cb, ac, bc, cc \ldots.$ 
\end{exa}

After the brief review on alphabets, strings (words) and languages we provide a formal definition of a deterministic finite automaton. 
\begin{defn}
A deterministic finite automaton is a tuple $(Q, \Sigma, \delta_a, q_0, F),$ 
where $a \in \Sigma$ and:
\begin{enumerate}
\item $Q$ is a finite set called the states
\item $\Sigma$ is a finite set called the alphabet
\item $\delta_a : Q \to Q$ called the transition functions
\item $q_0 \in  Q$  is the start state
\item  $F \subseteq Q$ is the set of accepting (or terminal) states.
\end{enumerate}
\end{defn}
We can also define the composition of transition maps in finite automata. If $\delta_b(q_i) = q_j$ and $\delta_a(q_j) = q_k,$ 
$q_i,q_j,q_k\in  Q,$ then $\delta_a(\delta_b(q_i)) = : \delta_{ab}(q_i) = q_k.$
Also, we say that an automaton is {\it incomplete} when some of the transitions are not defined, i.e. certain states are not mapped to new states via these transitions. In this case, an obvious choice would be to send all the ``un-mapped'' states to an extra added state denoted $\hat q$ (see more below in the text when we introduce the linearization of an automaton).

Some useful definitions of accepted words, language recognition and regular languages follow.
\begin{defn} (Word acceptance)
Let $\Sigma$ be our alphabet and let $A = (Q, \Sigma, \delta_a, q_0, F),$ be our finite
automaton. A finite sequence $w_1, w_2, \ldots, w_n,$ where each $w_i\in \Sigma$ 
is accepted by $A$ if and only if there exists a
sequence of states $r_0, r_1,\ldots, r_n \in Q$ such that:
\begin{enumerate}
\item $r_0 = q_0,$ we begin from the starting state
\item for each $i\in \big \{0,1, \ldots, n-1\big \},$ $\delta_{w_{i+1}}(r_i) = r_{i+1},$ i.e. the
computation follows exactly the word
\item $r_n \in F$, i.e. we end up in an accepting state.
\end{enumerate}
\end{defn}
{That is to say, after processing a valid input string from a recognized language, the deterministic finite automaton terminates in a designated accepting state. Any state reached during a transition that is neither an initial state nor an accepting state is defined as an intermediate state.}

\begin{defn} (Language Recognition)
We say that a deterministic finite automaton $A$ recognizes a
language $L$ if and only if $L = \big \{w | w ~\mbox{is accepted by} ~A\big \}.$
\end{defn}

\begin{defn} (Regular language, Kleene's Theorem)
A language is called regular if and only if it is recognized by some deterministic
finite automaton.
\end{defn}

In fact, we can apply certain regular operations on
languages, such that the language regularity is preserved, i.e. if we start
with a regular language, no matter how many times we
will apply these operations, we will still have a regular
language. The operations of {\it concatenation product}, {\it union}, {\it intersection}, {\it complement} 
and {\it Kleene's star} defined earlier are all regular operations.
We also note that two automata are said to be equivalent if they accept the same language.

Below, we illustrate directed graphs of finite-state automata with distinguished accepting states and labelled transitions, where each action labels exactly one outgoing arrow. The start state is indicated by a free incoming arrow on the left, whereas an accepting (final) state is represented by a double circle. {This double circle indicates that the string processed thus far belongs to the language recognized by the machine; any other state is represented by a single circle. Any state within a finite automaton can be designated as accepting, and accepting states may feature outgoing transitions; an automaton may pass through an accepting state mid-string, continuing to subsequent states as long as characters remain to be read. }
Note that in deterministic automata it is impossible for two arrows to leave the same state carrying the same label, 
i.e. the diagram in Figure \ref{fig1} is {\it forbidden}:
\begin{center}
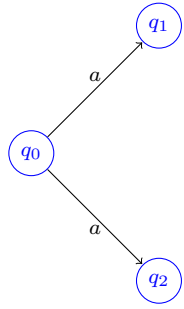
\begin{figure}[ht]
\footnotesize
\begin{tikzpicture}[shorten >=0.3pt,node distance=2.4cm,on grid,auto] 
   \node[state,inner sep=0, minimum size=2.0em, blue] (q_0)   {$q_0$}; 
   \node[state,inner sep=0, minimum size=2.0em, blue] (q_1) [above right=of q_0] {$q_1$}; 
   \node[state,inner sep=0, minimum size=2.0em, blue] (q_2) [below right=of q_0] {$q_2$}; 
    \path[->] 
    (q_0) edge[right, above]  node{$a$} (q_1)
    (q_0) edge[right, below] node{$a$} (q_2);
\end{tikzpicture}
\caption{Forbidden diagram in deterministic automata.} \label{fig1}
\end{figure}
\end{center}
\begin{exa} \label{exa01}
Our first example is a 3-state automaton $Q = \big \{q_1, q_2, q_3\big \}$ with $q_1$ 
being the start state, the alphabet is $\Sigma =\big \{a,b\big \}$ and we have chosen $q_2$ as the final state (Figure 2).
\begin{center}
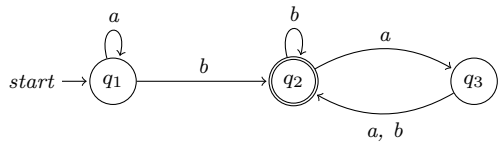
\begin{figure}[ht]
\footnotesize
\begin{tikzpicture}[shorten >=1pt,node distance=2.4cm,on grid,auto] 
\node[state,inner sep=0, minimum size=2.0em, initial] (q1) {$q_1$};
\node[state,inner sep=0, minimum size=2.0em, accepting, right of=q1] (q2) {$q_2$};
\node[state,inner sep=0, minimum size=2.0em, right of=q2] (q3) {$q_3$};
\path[->] (q1) edge[loop above] node{a} (q1)
(q1) edge[above] node{b} (q2)
(q2) edge[loop above] node{b} (q2)
(q2) edge[bend left, above] node{a} (q3)
(q3) edge[bend left, below] node{a, b} (q2);
\end{tikzpicture}
\caption{A 3-state automaton} \label{3s}
\end{figure}
\end{center}
The transition table for any automaton consists of rows and columns; 
the rows are labeled by the states and the columns are labeled by the input letters.
The transition table for the automaton is
\begin{center}
\begin{tabular}{ |c|c|c|c| } 
\hline
 & $a$ &  $b$  \\
\hline
\multirow{3}{2em}
{\ $q_1$ \\ \ $q_2$\\ \ $q_3$  } 
& $q_1$ & $q_2$  \\ 
& $q_3$ & $q_2$  \\ 
& $q_2$ & $q_2$  \\
\hline
\end{tabular}
\end{center}
\begin{center}
{Table 1}
\end{center}
\end{exa}
\begin{exa}\label{e1}
Our second example is a 4-state automaton $Q=\big \{q_1,q_2,q_3,q_4\big \},$ $\Sigma = \big \{a,b\big \},$ $q_1$ is 
the start state and $q_4$ is an accepting state.
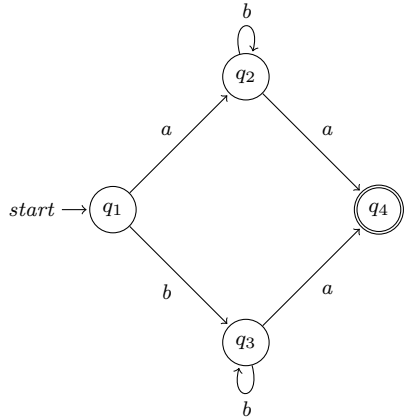
\begin{center}
\footnotesize
\begin{tikzpicture}[shorten >=1pt,node distance=2.5cm,on grid,auto] 
   \node[state,inner sep=0, minimum size=2.0em, initial] (q_0)   {$q_1$}; 
   \node[state, inner sep=0, minimum size=2.0em] (q_1) [above right=of q_0] {$q_2$}; 
   \node[state, inner sep=0, minimum size=2.0em] (q_2) [below right=of q_0] {$q_3$}; 
   \node[state,inner sep=0, minimum size=2.0em, accepting](q_3) [below right=of q_1] {$q_4$};
    \path[->] 
    (q_0) edge  node {a} (q_1)
          edge  node [swap] {b} (q_2)
    (q_1) edge  node  {a} (q_3)
          edge [loop above] node {b} ()
    (q_2) edge  node [swap] {a} (q_3) 
          edge [loop below] node {b} ();
\end{tikzpicture}
\end{center}
\end{exa}
In both examples above, given the transition tables, various other choices of starting and accepting states can be made.

\noindent{\bf Combinatorial automata.}
Throughout this manuscript, we consider maps of abstract finite automata on ${\mathbb R}^{n}$ (or ${\mathbb C}^n$ depending on the type of automaton we consider), i.e. we consider the 
{\it linearization} of the automaton.
\begin{rem} \label{rem1} {\bf (Linearization.)}
 Let $Q = \big \{q_1, q_2, \ldots, q_n\big \}$ be the set of states for some automaton. Let also $\Sigma =\big \{w_1,w_2, \ldots, w_n\big \}$ be the alphabet and $\delta_a: Q \to Q,$ $a\in \Sigma$ be the transition maps, such that $q_i \mapsto \delta_a(q_i) = q_j\in Q.$ 
 Via the linearization process, we will be able to express the states as vectors and the maps $\delta_a:  Q \to Q$ as $n \times n$ matrices. Specifically, consider the vector space $V= \mathbb{C}Q$  of dimension equal to the cardinality of $Q$. 

 Let ${\mathbb B}_n = \big  \{b_x \big  \},~x\in Q$ be a basis\footnote{We always consider orthonormal bases, given that any basis can be made  orthonormal by means of the Gram-Schmidt process.} of the $n$-dimensional vector space ${\mathbb C}^n,$ i.e. $b_x$ are in general $n$-dimensional linearly independent column vectors, such that $b_x^{\dagger} b_y = \delta_{x,y},$ $x,y\in Q,$ where $\delta_{x,y}$ is Kronecker's $\delta$ and $^{\dagger}$ denotes transposition and complex conjugation. In this article  we shall be primarily considering the standard canonical basis of ${\mathbb C}^n$ (or ${\mathbb R}^n$) given by $n$-dimensional column vectors $\hat e_{q_j},$ (or just $\hat e_j$) $q_j \in Q,$ such that they have one non-zero entry in the $j^{th}$ entry of the column vector that takes the value 1.
Let also ${\mathbb B}_n^*= \big  \{\hat e_x^T   \big  \},~x\in Q$ ( $^T$ denotes transposition) be the dual basis: $\hat e_x^T \hat e_y= \delta_{x,y},$ also $e_{x,y} := \hat e_x  \hat e_y^T,$ which form a basis of $\End({\mathbb C}^n),$ $x,y \in Q.$ 

Specifically, via linearization:  $Q \to {\mathbb C}^n,$ such that 
$q_i \mapsto  \hat e_{q_i}$ and any transition function $\delta_a: Q \to Q,$ $a \in \Sigma,$ is expressed as $n \times n$ matrix, $\delta_a \mapsto M_a \in \End({\mathbb C}^n)$:
$M_a = \underset{x,y \in Q}{\sum} (M_a)_{x,y}e_{x,y},$ such that $M_a \hat e_{q_i} = \hat e_{q_j}, ~~ q_i, q_j \in Q.$ Moreover, strings are created via matrix  multiplications: $M_{ab} =  M_a M_b,$ $a, b \in \Sigma,$ and this can be extended to elements in $\Sigma^*,$ that is $M_aM_w = M_{aw}$ for $a\in \Sigma$ and $w \in \Sigma^*,$ also $M_{\epsilon}$ is the identity matrix. In summary, for any finite state automaton of $n$ states, the transitions between states are represented by $n\times n$ matrices, called the transition matrices, whereas the states are represented by the standard basis of ${\mathbb C}^n$ as $n$-column vectors.

In the special case where a transition $\delta_a(y)$ is not defined for some $y \in Q$ and some $a \in \Sigma$, i.e. $y$ is not mapped to any state via $\delta_a,$  then in the linearized version we consider $M_a \hat e_y =0,$ i.e. $(M_a)_{x,y} =0,$ for the given $a, y$ and for all $x\in Q.$ That is to say, when there are undefined transitions, the corresponding transition matrix has zero columns. To conclude, all undefined transitions are mapped to the zero column vector.
\end{rem}
\begin{exa}

$ $

\begin{enumerate}
\item The linearization of the $3$-state automaton of Example \ref{exa01}:
\begin{equation}
    q_i \mapsto   \hat e_{q_i}, ~~i\in \big \{1,2,3\big \} ~~\mbox{and} ~~ 
    M_a = \begin{pmatrix} 1 &0 &0 \\ 0 &0 &1\\ 0 &1 & 0 \end{pmatrix}, ~~~M_b = \begin{pmatrix} 0 &0 &0 \\ 1 &1 &1\\ 0 &0 & 0 \end{pmatrix} \nonumber\end{equation}

\item The linearization of the $4$-state automaton of Example \ref{e1}:
\begin{equation}
     q_i \mapsto   \hat e_{q_i}, ~~i\in \big \{1,2,3,4\big \} ~~\mbox{and} ~~ M_a = \begin{pmatrix} 0 &0 &0 &0\\ 1 &0 &0 &0\\ 0 &0 & 0 &0\\ 0&1&1&0 \end{pmatrix}, ~~~M_b = \begin{pmatrix} 0 &0 &0 &0\\ 0 &1 &0&0\\ 1 &0 & 1& 0\\ 0&0&0&0 \end{pmatrix} \nonumber\end{equation}
\end{enumerate}
\end{exa}

In this manuscript, we distinguish three types of finite automata, depending on the type of transition matrices:
\begin{enumerate}
    \item Combinatorial automata:  the transition matrices are combinatorial.
    \item Probabilistic automata:  the transition matrices are stochastic. 
    \item Quantum automata: transition matrices are unitary.
\end{enumerate}
Precise definitions of combinatorial and probabilistic vectors and matrices are given later (see Definitions \ref{combmat} and \ref{probmat}). If the transition matrices are not combinatorial, stochastic or unitary then the automaton is simply characterized as non-deterministic. In any case, the probabilistic and quantum automata are special cases of non-deterministic automata.

We start by introducing the definitions of combinatorial vectors and matrices.
\begin{defn} \label{combmat} The column vector with all its entries being zero except one, which takes the value 1, is called a combinatorial vector.
   A matrix with columns that are combinatorial or zero vectors is called a combinatorial matrix. An $n\times n$ matrix with columns being $n$ distinct combinatorial vectors is called fully combinatorial. The $n \times n$ matrices $e_{i,j},$ $i,j \in [n]$ defined in Remark \ref{rem1} are called elementary combinatorial matrices.
   \end{defn}
We may now define the {\it combinatorial (or classical) finite automaton}, adapted to the purposes of the present analysis.
\begin{defn} Let $Q= \big \{q_1, q_2, \ldots, q_n\big \}$ be a finite set of abstract states.
A combinatorial finite automaton is a tuple $({\mathbb B}_n,\hat 0, \Sigma, M_a, q_0, F),$ 
where $a \in \Sigma$ and :
\begin{enumerate}
\item ${\mathbb B}_n = \big  \{\hat e_{q_1}, \hat e_{q_2}, \ldots, \hat e_{q_n} \big  \}$ is the standard canonical basis of ${\mathbb R}^n$. 
\item $\hat 0$ is the $n$-dimensional zero column vector
\item $\Sigma$ is a finite set called the alphabet
\item $M_a: {\mathbb B}_n \to {\mathbb B}_n \cup \hat 0$ are $n\times n$ combinatorial matrices, called transition matrices
\item $q_0 \in  {\mathbb B}_n$ is the start state
\item  $F \subseteq {\mathbb B}_n \cup \hat 0$ is the set of accepting (or terminal) states
\end{enumerate}
\end{defn}
According to Remark (\ref{rem1}) any deterministic abstract automaton can be mapped to a combinatorial automaton.

We shall now introduce the definition of isomorphic combinatorial automata.
\begin{defn} (Combinatorial isomorphisms) 
Two combinatorial automata $A:= ({\mathbb B}_n,\hat 0, \Sigma, M_a, q_0, F),$ $A':=({\mathbb B}_n,\hat 0, \Sigma, M'_a, q'_0, F'),$  
are isomorphic if there exists a combinatorial $n\times n$ matrix $S,$ such that for a bijective function $f: Q \to Q,$ $x \mapsto f(x),$ $S:= \underset{x\in Q}\sum e_{f(x), x},$ i.e. $\hat e_{f(x)} = S \hat e_x,$ $x \in Q$ and $M_a' = SM_aS^{-1},$ $a\in \Sigma.$
\end{defn}
Such combinatorial transformations basically reshuffle the elements of the basis, i.e., $\hat e_x \mapsto \hat e_{f(x)},$ for all $x\in Q,$ but the basis does not change. Henceforth, when we say deterministic automaton we refer to a combinatorial automaton. 
Note also that in principle, the set of states and the alphabet can be infinite sets; however, 
in this analysis we will be focusing on finite sets of states and finite alphabets.

Before we discuss non-deterministic automata in the next section we introduce the semi-combinatorial (or semi-deterministic)
automaton which will be used in our present analysis.
\begin{defn} Let $Q= \big \{q_1, q_2, \ldots, q_n\big \}$ be a finite set of abstract states.
A semi-combinatorial finite automaton is a tuple $({\mathbb B}_n,\hat 0, \Sigma, M_a, q_0, F),$ 
where $a \in \Sigma$ and :
\begin{enumerate}
\item ${\mathbb B}_n =\big \{\hat e_{q_1}, \hat e_{q_2}, \ldots, \hat e_{q_n}\big \}$ is the standard canonical basis of ${\mathbb C}^n$. 
\item $\hat 0$ is the $n$-dimensional zero column vector
\item $\Sigma$ is a finite set called the alphabet
\item $M_a: {\mathbb B}_n \to {\mathbb B}_n \cup \hat 0$ are $n\times n$ matrices, called transition matrices, such that $M_a  =\underset{x \in Q}{\sum} m_{x}^{(a)} e_{f(x), x},$ where $f:Q\to Q,$ $x \mapsto f(x),$ and  $m_{x}^{(a)} \in {\mathbb C}.$
\item $q_0 \in  {\mathbb B}_n$ is the start state
\item  $F \subseteq {\mathbb B}_n \cup \hat 0$ is the set of accepting (or terminal) states
\end{enumerate}
\end{defn}
Any transition from $\hat e_x$ to $\hat e_{f(x)},$ (or from $x$ to $f(x)$) $x\in Q$ in an automaton graph is represented as
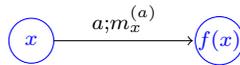
\begin{center}
\begin{figure}[ht]
\footnotesize
\begin{tikzpicture}[shorten >=0.8pt,node distance=2.5cm,on grid,auto] 
   \node[state,inner sep=0, minimum size=2.0em, blue] (q_0)   {$x$}; 
   \node[state,inner sep=0, minimum size=2.0em, blue] (q_1) [right=of q_0] {$f(x)$}; 
    \path[->] 
    (q_0) edge[]  node{$a$;$m^{(a)}_x$} (q_1);
\end{tikzpicture}
\end{figure}
\end{center}

A typical example of a semi-combinatorial automaton follows.
Note that in this manuscript, we indicate in the automaton diagram either the elements of the alphabet and the elements of the basis or the elements of the alphabet and the abstract states. In the graph of the example below for instance we just indicate the elements of the basis and the elements of the alphabet. It is also convenient to define the shorthand notation for any $n\in {\mathbb Z}^+,$ $[n]:= \big \{1,2,\ldots,n \big \}$
\begin{exa} \label{exafirst}
(The $\mathfrak{U}_q(\mathfrak{sl}_2)$ automaton.) We first recall the 
algebra $\mathfrak{U}_q(\mathfrak{sl}_2)$ \cite{Jimbo} (see also Definition \ref{defq}), which is the unital associative algebra over ${\mathbb C}$ (or ${\mathbb R}$) generated by $e$, $f$,
$q^{\pm \frac{h}{2}}$, and relations: 
\begin{equation}
q^{\frac{h}{2}}\ f=qf\ q^{\frac{h}{2}}\, \qquad q^{\frac{h}{2}}\ e
= q^{-1}e\ q^{\frac{h}{2}}, \qquad \big [f,\ e \big ] = \frac{q^h-q^{-h}} {q-q^{-1}},
\label{10} 
\end{equation}
where $\big [ ,  \big ]:  {\mathfrak U}_q(\mathfrak{sl}_2) \times{\mathfrak U}_q(\mathfrak{sl}_2) \to{\mathfrak U}_q(\mathfrak{sl}_2),$ such that $\big [a,b \big ] = ab -ba,$ $a,b \in {{\mathfrak U}_q(\mathfrak{sl}_2)}.$
In this manuscript we consider $q = e^{\mu}, \mu\in {\mathbb R}.$

We recall the standard canonical basis of ${\mathbb C}^d,$ ${\mathbb B}_d =  \big  \{\hat e_k \big  \},$ $k\in [d],$ $d\in {\mathbb Z}^+$ (Remark \ref{rem1}).  We also recall the $d$-dimensional irreducible representation ${\mathfrak U}_q(\mathfrak{sl}_2),$ 
$\rho: {{\mathfrak U}_q({\mathfrak{sl}_2})} \to \End({\mathbb C}^d),$ such that $q^{h} \mapsto q^{\mathrm{h}},$ $e \mapsto {\mathrm E}$ and $f\mapsto {\mathrm F}:$ ${\mathrm F} \hat e_1 =0,$ ${\mathrm E}\hat e_d=0$ and
\begin{eqnarray}
&&q^{\mathrm{h}}\ \hat e_k = q^{\hat a_k} \hat e_{k}, ~~k\in[d] \label{a1}\\
&&{\mathrm E}\ \hat e_{k}= \hat c_ke_{k+1}, ~~~{\mathrm F}\ \hat e_{k+1} = \hat c_k\hat e_{k},~~~k\in[d-1], \label{a2}
\end{eqnarray}
where $\hat a_k = d+1-2k,$ $\hat c_k = \sqrt{[k]_q[d-k]_q}$ and $[k]_q = \frac{q^k -q^{-k}}{q-q^{-1}}.$ 

There are other $d$-dimensional irreducible representations of ${\mathfrak U}_q(\mathfrak{sl}_2),$ up to an algebra homomorphism. 
Indeed, let ${\mathfrak d} \in \mathfrak{U}_q(\mathfrak{sl}_2)$ be an invertible element such that,
\begin{equation}
\big [ {\mathfrak d}, q^{h}\big ] = \big [{\mathfrak d},\ FE \big] =0. \nonumber
\end{equation}
And consider the map ${\mathfrak h}: \mathfrak{U}_q(\mathfrak{sl}_2) \to \mathfrak{U}_q(\mathfrak{sl}_2),$ such that
\begin{equation}
 e\mapsto e':= e {\mathfrak d}^{-1}, \quad f \mapsto f' : = {\mathfrak d} f, \quad q^{h} \mapsto q^{h}.\nonumber
\end{equation}
Then ${\mathfrak h}$ is an algebra homomorphism.

Recall the representation $\rho: {\mathfrak U}_q(\mathfrak{sl}_2) \to \End({\mathbb C}^d),$  (\ref{a1}), (\ref{a2}) and $ {\mathfrak d} \mapsto {\cal D} = \underset{0\leq k\leq d-1}{\sum} \hat c^*_{k}e_{k+1,k+1} \in \End({\mathbb C}^d),$ $\hat c^*_0 =1$ and $\hat c^*_k = \hat c_k,$ $k \in [d-1]$ (in general ${\cal D}$ can be any diagonal $d\times d$ matrix) then, 
\begin{equation}
    {\mathrm E} \mapsto {\mathrm E}' := {\mathrm E} {\cal D}^{-1}, \quad {\mathrm F} \mapsto {\mathrm F}':={\cal D}{\mathrm F},\quad q^{{\mathrm h}} \mapsto q^{\mathrm h}, \label{homo0}
\end{equation}
and 
\begin{eqnarray}
&&q^{\mathrm{h}}\ \hat e_k = q^{\hat a_k} \hat e_{k}, ~~k\in[d] \label{a1b}\\
&&{\mathrm E}'\ \hat e_{k}=\hat  e_{k+1}, ~~~{\mathrm F}'\ \hat e_{k+1} = \hat \kappa_k \hat e_{k},~~~k\in[d-1], \label{a2b}
\end{eqnarray}
where $\hat a_k = d+1-2k,$ $\hat \kappa_k = [k]_q[d-k]_q.$ 

The ${\mathfrak U}_q(\mathfrak{sl}_2)$ automaton: $Q = [d],$ $\Sigma = \big \{q^{h}, e,f\big \}$ 
and the transition matrices are given in (\ref{a1}), (\ref{a2}) (or (\ref{a1b}), (\ref{a2b})); this is obviously a semi-combinatorial automaton. 
The automaton is graphically depicted in Figure 3, if $\hat e_1$ is chosen as a start state.
\begin{center}
\begin{figure}[ht]
\tiny
\begin{tikzpicture}[shorten >=1pt,node distance=2.5cm, on grid,auto] 
\node[state, inner sep=0, minimum size=2.0em, initial] (q1) {${\hat e_1}$};
\node[state,inner sep=0, minimum size=2.0em, right of=q1] (q2) {$\hat e_2$};
\node[inner sep=0, minimum size=2.0em, right of=q2] (q3) {$\ldots$};
\node[state,inner sep=0, minimum size=2.0em, right of=q3] (q3b) {$\hat e_d$};
\node[state,inner sep=0, minimum size=2.0em] (q0) [below right= 3.5 and 3.5 of q1] {$\hat 0$};
\path[->] 
(q1) edge[bend right, left] node{$f$} (q0)
(q3b) edge[bend left, right] node{$e$} (q0)
(q1) edge[bend left,  above] node{$e;\hat c_1(1)$} (q2)
(q2) edge[bend left, below] node{$f;\hat c_1(\hat \kappa_1)$} (q1)
(q2) edge[bend left, above] node{$e; \hat c_2(1)$} (q3)
(q3) edge[bend left, below] node{$f;\hat c_{2} (\hat \kappa_2)$} (q2)
(q3) edge[bend left, above] node{$e; \hat c_{d-1}(1)$} (q3b)
(q3b) edge[bend left, below] node{$f; \hat c_{d-1} (\hat \kappa_{d-1})$} (q3)
(q2) edge[loop above] node{$q^{h}; \hat a_{2}$} (q2)
(q3b) edge[loop above] node{$q^{h}; \hat a_{d}$} (q3b)
(q1) edge[loop above] node{$q^{\mathrm h};\hat a_1$} (q1);
\end{tikzpicture}
\caption{A semi-combinatorial automaton}\label{comb}
\end{figure}
\end{center}
Henceforth, the undefined (zero) transitions are not depicted in the automaton diagrams.
\end{exa}

\section{Non-deterministic automata: probabilistic and quantum finite automata}

\noindent Non-deterministic finite automata
might include actions, labeled by a letter of the alphabet, that lead to different states
simultaneously (see e.g. Figure 1).
Every deterministic finite automaton is just a
special case of non-deterministic finite automata. 
Non-deterministic automata were introduced by Rabin and Scott \cite{RabinScott}, 
who showed their equivalence to deterministic automata. Recall, two automata are said to be equivalent if they accept the same language.

\begin{exa} \label{non} Consider the two state non-deterministic automaton, where $Q = \big \{q_1, q_2\big \}$ and $\Sigma = \big \{a,b \big \},$ 
and the associated transition matrices given as
$$M_a = \begin{pmatrix} x &0  \\ y &1  \end{pmatrix}, ~~~M_b = \begin{pmatrix} 1 &y  \\ 0 &x \end{pmatrix}, $$
and choose $q_1$ as the start state.
The corresponding graph for the non-deterministic automaton is shown in Figure 4 
(from now on we do not indicate accepting states in the automaton graph).
\begin{center}
\begin{figure}[ht]
\footnotesize
\begin{tikzpicture}[shorten >=1pt,node distance=2.8cm,on grid,auto] 
\node[state,inner sep=0, minimum size=2.0em, initial] (q1) {$q_1$};
\node[state,inner sep=0, minimum size=2.0em, right of=q1] (q2) {$q_2$};
\path[->] (q1) edge[loop above] node{a:x} (q1)
(q2) edge[loop below] node{b:x} (q2)
(q1) edge[loop below] node{b:1} (q1)
(q2) edge[loop above] node{a:1} (q2)
(q1) edge[bend left, above] node{a:y} (q2)
(q2) edge[bend left, below] node{b:y} (q1);
\end{tikzpicture}
\caption{A non-deterministic automaton} \label{nda}
\end{figure}
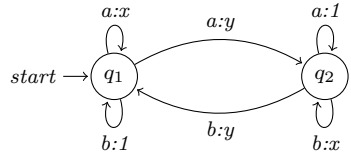
\end{center}
Notice that edges between states are labeled by the letter of the alphabet labeling the transition, and by the matrix element $(M_a)_{j,i}$ that corresponds to the transition from $q_i$ to $q_j$. In the non-deterministic automaton, contrary to the combinatorial automaton, some of the elements (or maybe all) $(M_a)_{j,i} \neq 1, 0.$ 
\end{exa}

We focus now on some special cases of non-deterministic automata, 
namely the probabilistic and quantum automata. Probabilistic automata, first introduced in \cite{Rabin}, are
finite Markov chains and can also be seen as random walks on directed graphs.
Quantum automata were more recently introduced in \cite{quantum1} and \cite{quantum2} independently 
and are prototypes of quantum computers (see also \cite{review} for more details on definitions, 
examples and historical information on the subject).

\subsection{Probabilistic automata}
\noindent Before we define the probabilistic automaton, we give a couple of necessary preliminary definitions.
For a more detailed exposition on probabilistic automata, see \cite{Rabin}.
Note that the definitions given in this section are based on the linearization of an automaton, 
given in Remark \ref{rem1}. Indeed, if $Q$ is a finite set of states of cardinality $n$, then this 
is mapped to the standard basis ${\mathbb B}_n$ of ${\mathbb R}^n.$
\begin{defn} \label{probmat}
A vector is stochastic if all its entries are non-negative real numbers and
sum to 1. A matrix is stochastic if all its column vectors are stochastic.
\end{defn}

\begin{defn}
A probabilistic automaton $A$ is a tuple $({\mathbb B}_n, \Sigma, M_a, q_0, F),$ $a\in \Sigma,$
where $\Sigma$ is some alphabet, ${\mathbb B}_n$ is the standard basis of ${\mathbb R}^n,$ $M_a$ are $n\times n$ stochastic matrices (transition matrices),
$q_0\in {\mathbb B}_n$ is the initial state and $F \subseteq {\mathbb B}_n$ is a set of accepting states.
\end{defn}

We recall that any transition matrix can be expressed as $M_a =\underset{x,y \in Q}{\sum}(M_a)_{x,y} e_{x,y}$
The value $(M_a)_{x,y}$ is the probability that the automaton moves from the
state $\hat e_y$ to the state $\hat e_x$ after reading the letter $a$. As in the deterministic case, $M_a M_b = M_{ab},$ $a,b \in \Sigma$ 
and we extend to $\Sigma^*:$ $M_{aw} = M_aM_w$ for $a\in \Sigma$ and $w \in \Sigma^*,$ recall that $M_{\epsilon}$ is the identity matrix.

We represent the initial and final states by $n$-dimensional column vector denoted $q_0, q_F \in {\mathbb B}_n$ respectively. 
Let $q_0 = \hat e_y,$ $q_F=\hat e_x,$ $x,y \in Q.$
The state distribution induced by a string $w \in \Sigma^*$ is
$P_{w} = M_w \hat e_y,$ such that $P_{w} = \underset{x\in X}{\sum} (M_w)_{x,y} \hat e_x$ is a stochastic vector and
 $(M_w)_{x,y}$ is the probability that the automaton moves to $\hat e_x$ after reading the
string $w$ with the initial state distribution $\hat e_y$.
The probability of the automaton accepting $w \in \Sigma^*$ is therefore $q_F^TP_w$. 
\begin{exa}
Consider the $2$-state probabilistic automaton $Q = \big  \{q_1, q_2 \big  \}$ and $\Sigma = \big \{a,b\big \}.$
\begin{equation} 
    M_a = \begin{pmatrix} p &1  \\ 1-p &0  \end{pmatrix}, ~~~M_b = \begin{pmatrix} 1 &1-p  \\ 0 &p \end{pmatrix},\nonumber \end{equation}
    where $0\leq p \leq  1.$ This is a special case of the non-deterministic automaton \ref{non}. If $p=0,$ or $p=1$ we obtain combinatorial (incomplete) automata. The zero vector can be added as an extra state.
\end{exa}

\begin{exa}
A 3-state probabilistic automaton: $\Sigma = \big \{a\big \},$ $Q= \big \{q_1, q_2, q_3\big \}$, let $q_1$ be the start state, and the transition matrix 
is given as
\begin{equation}
M_a =   \begin{pmatrix} 0 &\frac{1}{3} & \frac{2}{3} \\ 
\frac{1}{3} &\frac{1}{3}  & \frac{1}{3}\\ \frac{2}{3}& \frac{1}{3} &0 \end{pmatrix} \nonumber
\end{equation}
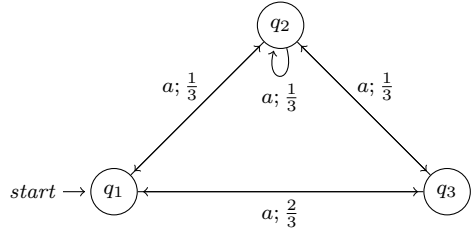
\begin{center}
\footnotesize
\begin{tikzpicture}[shorten >=2pt,node distance=2.5cm] 
   \node[state,inner sep=0, minimum size=2.0em, initial] (q_0)   {$q_1$}; 
   \node[state,inner sep=0, minimum size=2.0em] (q_1) [above right=of q_0] {$q_2$}; 
   \node[state,inner sep=0, minimum size=2.0em] (q_2) [below right=of q_1] {$q_3$}; 
    \path[->] 
    (q_0) edge[left, above]  node {$a;\frac{1}{3}~~~~$} (q_1)  
    (q_1) edge[left, below]  node {$ $} (q_0) 
    (q_1) edge[left, above]  node [swap]  {$~~~a; \frac{1}{3}$} (q_2)
    (q_2) edge[right, below]  node  [swap] {$ $} (q_1)
    (q_1) edge [loop below] node {$a;\frac{1}{3}$} (q_1)
    (q_0) edge[left, above]  node  {$ $} (q_2)
 (q_2) edge[left, below]  node  {$a;\frac{2}{3}$} (q_0);
\end{tikzpicture}
\end{center}
\end{exa}
We conclude our brief description of probabilistic automata by defining language equivalent probabilistic automata \cite{Rabin, Tzeng}.
\begin{defn} (Language equivalence). 
Two probabilistic automata $A_1$ and $A_2$ with the same alphabet
are said to be language equivalent (for short, we use only equivalent) if for all strings $w\in \Sigma^*$ 
the two automata accept $w$ with the same probability. 
\end{defn}

\subsection{Quantum automata}
Here we provide a generic definition of quantum automata, more specific definitions that describe the dynamics 
of quantum systems can be found for instance in \cite{review}.
Before we give the definition of the quantum automaton we are going to use in this manuscript, we recall that a complex valued $n\times n$ matrix $U$ is called unitary if $U^{-1} = U^{\dagger}$ (recall, $^\dagger$ denotes transposition and complex conjugation).
\begin{defn}
A quantum automaton $A$ is a tuple $({\mathbb B}_n,\Sigma, M_a, q_0, F),$ $a\in \Sigma,$
where $\Sigma$ is some alphabet, ${\mathbb B}_n$ is the standard basis of ${\mathbb C}^n,$ $M_a$ are $n\times n$ unitary matrices, the transition matrices,
$q_0\in {\mathbb B}_n$ is the initial state and $F \subseteq {\mathbb B}_n$ is a set of accepting states.
\end{defn}
In probabilistic and quantum automata, as opposed to classical (combinatorial) automata, a state obtained after the action of a transition matrix can be a superposition
of basis states. In this sense, any probabilistic or quantum automaton sends basis states to linear combinations of basis states. The
analogous combinatorial-machine sends only basis
states to basis states. Each quantum
automaton consists of basis states and 
the state of the automaton after the action of the transition matrix is a superposition over
them. However, in quantum automata the superposition over basis states
is {\it not} a probability distribution any more. A more precise interpretation of the superposition of basis states in quantum automata will be given towards the end of this subsection.
\begin{exa}
A 2-state quantum automaton: $\Sigma = \big \{1\big \},$ $Q = \big \{q_1, q_2,\big \}$, let $q_1$ be the start state, and the transition matrix is given as
\begin{equation}
M_1 =  \begin{pmatrix} a & b\\ -b^* &a^*   \end{pmatrix}, \nonumber \end{equation}
where $a,b \in {\mathbb C},$ $a^*, b^*$ the complex conjugates and $|a|^2 +|b|^2 =1,$ i.e. $M_1$ is a unitary matrix (Figure 5).
\begin{center}
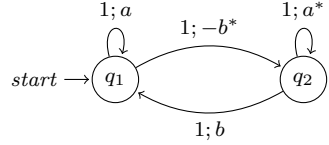
\begin{figure}[ht]
\footnotesize
\begin{tikzpicture}[shorten >=1pt,node distance=2.5cm,on grid,auto] 
\node[state,inner sep=0, minimum size=2.0em, initial] (q1) {$q_1$};
\node[state,inner sep=0, minimum size=2.0em, right of=q1] (q2) {$q_2$};
\path[->] (q1) edge[loop above] node{$1;a$} (q1)
(q1) edge[bend left, above] node{$1;-b^*$} (q2)
(q2) edge[loop above] node{$1;a^*$} (q2)
(q2) edge[bend left, below] node{$1;b$} (q1);
\end{tikzpicture}
\caption{A 2-state quantum automaton} \label{2s}
\end{figure}
\end{center}
\end{exa}

\begin{defn} (Isomorphic quantum automata) 
Two quantum automata $A:= ({\mathbb B}_n, \Sigma, M_a, q_0, F),$ $A':=({\mathbb B}_n', \Sigma, M'_a, q'_0, F'),$ 
are isomorphic if there exists a unitary $n\times n$ matrix $S$, such that $\hat e'_x = S \hat e_x,$ for all $x \in Q$ and $M_a' = SM_aS^{-1},$ for all $a\in \Sigma.$
\end{defn}
\begin{lemma} \label{lemma1}
Let $M$ be an $n \times n$ unitary matrix and ${\mathbb B}_n =\big  \{\hat e_{q_j}\big  \},$ $q_j\in Q,$ $j\in [n],$ be a canonical basis in ${\mathbb C}^n,$ i.e. $\hat e_a^{\dagger}\hat  e_b = \delta_{a,b},$ $a,b \in Q.$ If $\psi_a = M\hat e_a,$ then $ \psi_a^{\dagger}\psi_b = \delta_{a,b},$ $a,b \in Q.$ Also, if $M_1, M_2$ are both $n\times n$ unitary matrices, then ${\mathbb M}: =M_1M_2$ is also unitary.
\end{lemma}
\begin{proof}
    The proof is straightforward.
\end{proof}
As in the deterministic and probabilistic case, $M_a M_b = M_{ab},$ $a,b \in \Sigma$ 
and we extend to $\Sigma^*:$ $M_{aw} = M_aM_w$ for $a\in \Sigma$ and $w \in \Sigma^*,$ recall that $M_{\epsilon}$ is the identity matrix.
Let $({\mathbb B}_n,\Sigma, M_a, q_0, F),$ be a quantum automaton and the transition matrix for a string $w \in \Sigma^*$ is expressed as
$M_w= \underset{x,y \in Q}{\sum}(M_w)_{x,y} e_{x,y}$. Let also $\psi_{w,y}:= M_w \hat e_y,$ then by Lemma \ref{lemma1} we deduce
\begin{equation}
\psi_{w,y} = \sum_{x\in Q} (M_w)_{x,y} \hat e_x, ~~~~\sum_{x\in Q}|(M_w)_{x,y}|^2 =1.\nonumber
\end{equation}
The amplitude $|(M_w)_{x,y}|^2$ is then interpreted as the probability for the quantum automaton to move in the final state $\hat e_x,$ after reading the string $w$ with initial state $\hat e_y,$ $x,y \in Q.$ In that sense, two quantum automata are said to be equivalent if they accept any given input string with the same probability.

\section{A brief review on Coxeter groups}

\noindent We give a brief review on Coxeter groups and the associated word problem. That is the problem of finding all possible words for a given Coxeter group. We provide for this purpose the definition of the {\it weak Bruhat order} (or  tree order) and give some concrete examples. This is a very brief overview on the subject presenting the basic definitions necessary for our analysis here, however for a more detailed description, see for instance \cite{Comb-book} and references therein.

A Coxeter group is a group with a certain presentation. 
\begin{defn}
Choose a finite generating set $S = \big  \{s_1, \ldots , s_p \big  \}$ and for every $i < j,$ choose
an integer $m(i , j) \geq 2,$ or $m(i , j) = \infty$. We define the associated Coxeter groups $W$:
\begin{equation}
W = \langle S \ |\ s_i^2 = 1, \ \forall i ~~\mbox{and}~~ (s_i s_j)^{m(i ,j)} = 1,\ \forall i < j \rangle \nonumber
\end{equation}
\end{defn}
In the special case where $m(i,j) \in \big \{2,3,4,6\big \}$ the Coxeter group is called a Weyl group. We consider below two basic examples.

\begin{exa} \label{defs0} The two basic examples we consider are the dihedral group and symmetric group.
   \begin{enumerate}
   \item {\it The dihedral group:} The group is defined as ($m(1,2) =4$)
\begin{equation}
b_2 = \langle 
S=\big \{s_1, s_2\big \}\ |\ s_1^2 =
1 = s_2^2, ~~\mbox{and} ~~
s_1s_2s_1{s_2}= s_2s_1s_2s_1 \rangle. \nonumber
\end{equation}
This is the Coxeter group associated with the root system $b_2$.
Its elements are
$1, s_1,$ $s_2,$ $s_1s_2,$ $s_2s_1,$ $s_1s_2s_1,$ $s_2s_1s_2,$ $s_1s_2s_1s_2 = s_2s_1s_2s_1.$

\item  The symmetric group $S_{p+1}$ ($A_p$):
This is the group of permutations. Let
$s_i = (i, i+1),$ the symmetric group is a Coxeter group defined as
\begin{equation}
S_{p+1} = \langle 
S= \big \{s_1, s_2, \ldots, s_p\big \}\ |\ s_j^2 = 
1, ~~\forall j \in \big \{1,2,\ldots,p\big \} ~~\mbox{and} ~~
s_1s_2s_1= s_2s_1s_2 \rangle. \nonumber
\end{equation}
\end{enumerate}
\end{exa}

We note that the set $S$ is called the set of {\it simple reflections}. The set, 
$T= \langle  wsw^{-1}: \ w\in W, \ s\in S \rangle, $
is called the set of {\it reflections}.
We also introduce the notion of {\it reduced words}. The group $W$ is generated by
$S$, each element $w$ of $W$ can be written
(in various ways) as a word in the ``alphabet'' $S$.
\begin{defn}
Given a Coxeter system $(W, S ),$ an expression $w = s_{i_1} \ldots s_{i_m} \in W$ is called
reduced if $w$ cannot be written as a product of $s_i$ with fewer terms, and
$m$ is called the length $|w|$ of $w$.
\end{defn}
That is to say, a word of minimal length, among words for $w$, 
is a {\it reduced word} for $w$.
The length $|w|$ of $w$ is the length of a reduced word for $w$ (solution to the word problem for $W$ 
by Matsumoto \cite{Matsu}):
any word for $w$ can be converted to a reduced word by a
sequence of
\begin{enumerate}
\item  braid moves: $s_i s_j s_i \ldots \leftrightarrow s_j s_i s_j \ldots $ ($m(i , j)$ letters)
\item  nil moves: delete $s_i^2$.
\end{enumerate}
Any two reduced words for $w$ are related by a sequence of braid
moves.

 We also introduce the set of {\it left inversions} (or just inversions) of $w \in W,$ denoted $inv(w),$ {which consists of any} reflection $t \in T,$ such that
$|tw| < |w|$. 
Specifically, if $a_1 \ldots a_k$ is a reduced word for $w$, then write $t_i = a_1 \ldots a_i \ldots a_1,$ and
$inv(w) = \big \{t_i : 1 \leq i\leq k\big \}.$
The sequence $t_1, \ldots, t_k$ is the reflection sequence for the reduced
word $a_1 \ldots a_k.$

One of the basic problems is the notion of order for any set (recall Remark \ref{tree} about the left tree order). 
We give below the definition of a weak Bruhat order.

\begin{defn}  The (left) weak order on a Coxeter group $W$ sets $u \leq w$ if and only if a
reduced word for $u$ occurs as a suffix of some reduced word for $w$.
The covers are $w < sw$ for $w \in W$ and $s \in S$ with $|w| < |sw|.$
Equivalently, $u \leq w$ if and only if $inv(u) \subseteq inv(w).$
\end{defn}

\begin{exa} \label{dih}
The  dihedral group $b_2$ weak Bruhat order: 
\begin{center}
\begin{tikzpicture}[shorten >=1pt,node distance=2.4cm,on grid,auto] 
   \node[] (q_1)   {$1$}; 
   \node[] (q_3) [above right=of q_1] {$s_2$}; 
   \node[] (q_2) [above left=of q_1] {$s_1$}; 
   \node[](q_4) [above=of q_2] {$s_2s_1$};
 \node[](q_5) [above=of q_3] {$s_1s_2$}; 
  \node[](q_6) [above =of q_4] {$s_1s_2s_1$};
  \node[](q_7) [above =of q_5] {$s_2s_1s_2$};
  \node[](q_9) [above left=of q_7] {$s_2s_1s_2s_1 =s_1s_2s_1s_2$};  
\path[->]
  (q_1) edge[left, below]  (q_2)
   (q_1) edge[right, below]  (q_3)
   (q_2) edge[left, above]  (q_4)   
    (q_3) edge[right, above]  (q_5)   
     (q_4) edge[left, above]  (q_6)
    (q_5) edge[right, above] (q_7)
    (q_7) edge[right, below]  (q_9)
    (q_6) edge[left, below] (q_9) ;    
    \end{tikzpicture}    
    \end{center}
\end{exa}

\begin{exa} \label{sym2} 

$ $

\begin{enumerate}
    \item  The symmetric group $S_2$ weak Bruhat order: 
\begin{center}
\begin{tikzpicture}[shorten >=1pt,node distance=2.5cm,on grid,auto] 
   \node[] (q_1)   {$1$}; 
   \node[] (q_2) [right=of q_1] {$s_1$}; 
\path[->]
  (q_1) edge[left, below]  (q_2);
    \end{tikzpicture}    
    \end{center}
    
    \item The symmetric group $S_3$ weak Bruhat order: 
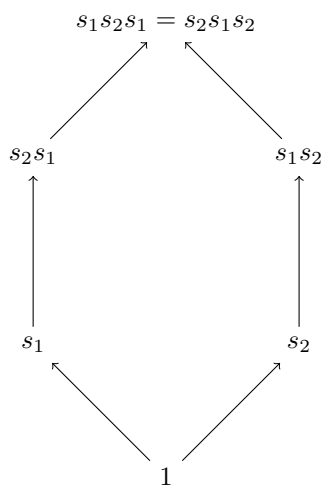
\begin{center}
\begin{tikzpicture}[shorten >=1pt,node distance=2.5cm,on grid,auto] 
   \node[] (q_1)   {$1$}; 
   \node[] (q_3) [above right=of q_1] {$s_2$}; 
   \node[] (q_2) [above left=of q_1] {$s_1$}; 
   \node[](q_4) [above=of q_2] {$s_2s_1$};
 \node[](q_5) [above=of q_3] {$s_1s_2$}; 
  \node[](q_6) [above right=of q_4] {$s_1s_2s_1=s_2s_1s_2$};
\path[->]
  (q_1) edge[left, below]  (q_2)
   (q_1) edge[right, below]  (q_3)
   (q_2) edge[left, above]  (q_4)   
    (q_3) edge[right, above]  (q_5)   
     (q_4) edge[left, above]  (q_6)
    (q_5) edge[right, above] (q_6);
    \end{tikzpicture}    
    \end{center}
\end{enumerate}
\end{exa}
We observe in the weak Bruhat order in the examples above that each horizontal level contains words of equal length, whereas as we ascend levels the length of words is increased. In general, according to the weak-order definition, words at higher horizontal levels are larger than words of lower levels. We cannot, however, compare words that occur at the same level, with the exception of level one that contains $s_1,s_2,\ldots, s_p$ (level 0 contains the unit element only), where that standard lexicographic order is chosen, i.e. $s_1 < s_2 < \ldots < s_p$. That is, these are partially ordered sets (posets).

The main difference with the left tree order in Remark \ref{tree} is that here we make use of 1) nil moves, i.e. $s_i^2$ is deleted whenever occurs in the graphs and 2) braid moves, i.e. in Example \ref{dih}, $s_1s_2s_1s_2=s_2s_1s_2s_1$ and in Example \ref{sym2}, $s_1s_2s_1 = s_2s_1s_2$, and this is the reason that comparison of strings of the same level is not possible.

\section{Hecke and quantum algebras } 
\noindent In this section we review some background material necessary for the purposes of this study. Specifically, in Subsections 5.1 and 5.2 we recall  definitions regarding the braid group, Hecke algebras and the algebra $\mathfrak{U}_q(\mathfrak{gl}_n)$ \cite{Jimbo, Jimbo2}. In Subsection 5.3 we briefly discuss definitions and examples on
Young-tableaux and elements of representation theory (see also for instance \cite{Fulton, FultonHarris} for a more detailed exposition on these subjects).

\subsection{Hecke algebra and the shuffle element}
\noindent We recall the definitions of the $A$-type braid group and the Hecke algebra ${\cal H}_N(q)$ and also introduce the so-called ``shuffle'' element that generates all reduced words of the Hecke algebra.
\begin{defn}
The braid group associated with a Coxeter system $(W, S)$ is\\
$B_W = \langle  \sigma_i, \ i \in I \ |\ \sigma_ i \sigma_j \sigma_i \ldots
= \sigma_j \sigma_i \sigma_j \ldots, 
\ m(i ,j)\ \mbox{terms},\    m(i , j) < \infty \rangle.$
\end{defn}
Note that in general $\sigma_i^2 \neq 1.$ 

We will focus in this study on the Artin presentation of the braid group, i.e. 
the standard braid group on $N$ strands and its quotient, the Hecke algebra ${\cal H}_N(q).$
\begin{defn}  \label{Artin}
The $A$-type Artin braid group $B_{N}$ is defined by generators $\sigma_1,\ \sigma_2, \ldots, \sigma_{N-1}$ and relations
\begin{equation}
\sigma_i\sigma_{i+1}\sigma_i = \sigma_{i+1} \sigma_i \sigma_{i+1}, ~~ \mbox{and }~~\sigma_i\sigma_j=\sigma_j\sigma_i~~\mbox{if}~~|i-j|>1. \nonumber
\end{equation}
\end{defn}
Every braid on $N$ strands determines a permutation on $N$ elements. This assignment becomes a map $B_N \to S_N,$ 
such that $\sigma_i \in B_N$ is mapped to the transposition $s_i = (i, i+1) \in S_N$. These transpositions generate the symmetric group (see previous section), satisfy the braid group relations and in addition  $s_i^2 =1.$ This transforms the Artin presentation of the braid group into the Coxeter presentation of the symmetric group (see the Definition of the symmetric group in Example \ref{defs0} (2)). 

\begin{defn} \label{Hecke1}
The Hecke algebra ${\cal H}_N(q),$ $q \in {\mathbb C},$ is a unital associative algebra over ${\mathbb C}$, defined by generators $t_1, t_2, \ldots, t_{N-1}$ and relations
\begin{eqnarray}
t_it_{i+1}t_i = t_{i+1} t_i t_{i+1}, ~~ (t_i-q1)(t_i+q^{-1}1)=0 ~~\mbox{and}~~t_it_j=t_jt_i ~~\mbox{if} ~~|i-j|>1. \nonumber
\end{eqnarray}
\end{defn}

\begin{rem} \label{tree2} {\bf (Left tree order for the Hecke algebra ${\cal H}_N(q)$).}
Any word for $w$ can be converted to a reduced word by a
sequence of
\begin{enumerate}
\item  braid moves: $t_j t_{j+1} t_j \leftrightarrow t_j t_{j+1}t_j$  and $t_i t_j \leftrightarrow  t_jt_i,$ $|i-j|>1,$ for all 
$i,j \in [N-1].$
\item  nil moves: delete $t_i^2$.
\end{enumerate}
Any two reduced words for $w$ are related by a sequence of braid
moves and nil moves due to $t_i^2 = (q-q^{-1}) t_i + 1,$ (same logic as in \cite{Matsu}). Specifically, regarding the nil moves, let $w$ be a word of length $l$ and $t_iw$ be a word of length $l+1.$ Then from the condition $t_i^2 = (q-q^{-1})t_i +1,$ we deduce that $t_i^2w$ reduces to either $w$ or $t_i w,$ i.e. it reduces to words that already exist at level $l$ and $l+1$ respectively, and thus $t_i^2w$ is deleted from the left order tree diagram.
The left tree diagram of the Hecke algebra coincides with the weak order diagram of the symmetric group. For example, the diagram of the order of the left tree for $n=3$ is given in Example (\ref{sym2}) (2) ($s_i \leftrightarrow t_i$ in the diagram).
\end{rem}

\begin{rem} \label{pol} Consider a free unital, associative, algebra over some field $K,$ generated by the alphabet 
$\Sigma = \big \{ a_1,a_2,\ldots, a_{N} \big \},$ that is a non-commutative polynomial algebra over $K$ (in our analysis here the field is either 
${\mathbb C}$ or ${\mathbb R}$). The {sum of all} length-$k$ words are then generated by $G_{k,N}= (a_1 +a_2+ \ldots +a_{N})^k,$ i.e. we have $N^k$ words of length $k.$ The corresponding left order tree diagram is given in Remark \ref{tree} and at each horizontal level $k$ of the tree diagram we have $N^k$ words of length $k.$   In general, the generating function of a special linear combination of all words in the free associative algebra generated by $\Sigma = \big \{a_1, a_2, \ldots, a_{N}\big \}$ is
 \[ {\cal G}_N := \frac{1}{1- t\underset{1\leq j\leq N}{\sum} a_j} = 1 +\sum_{k=1}^{\infty} t^k (a_1 +a_2 +\ldots a_{N})^k = 1 + \sum_{k=1}^{\infty}t^k G_{k,N} \].

In the case of the Hecke algebra sums of words of length $k$ are still generated by $G_k$, but subject to the conditions of Remark \ref{tree2}.  In general, for any associative, unital algebra (quotient of the free algebra) the left (right) order tree diagram can be constructed as in Remark \ref{tree}, but subject to the algebra relations.  
\end{rem} 
{One can analyze in principle  any linear combination of the free associative algebra generators, however we focus here on their sum (see relevant comments below for the Hecke algebra)}. As an example to illustrate Remark \ref{pol} we consider the relatively simple case for $N=2,$ i.e.  $\Sigma = \big \{a,b\big \}$ and we distinguish two cases:

\noindent {\bf A. The free algebra:} Consider a free associative, unital, algebra generated by $\Sigma = \big \{a_1,a_2, \ldots, a_{N}\big \}$ over some field $K.$ In this case the left (right) tree order diagram, e.g. for $N=2,$ $\Sigma= \big \{a,b\big \},$ is shown below.
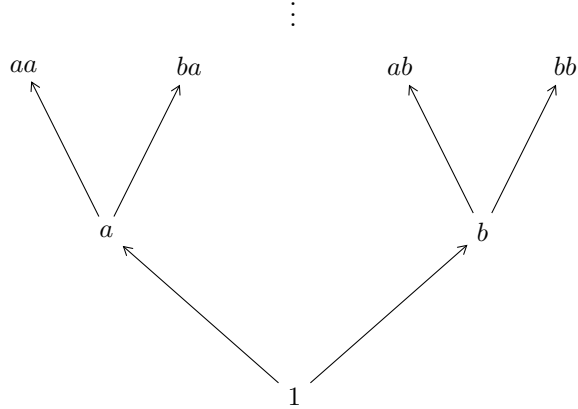
\begin{center}
$$\vdots$$
\begin{tikzpicture}[grow=up,->,>=angle 60]
\begin{scope}[yshift=-4cm]
  \node {$1$}
    child {node {$b$}
          child {node {$bb$}
      }
      child {node {$ab$}
      }    
    }
    child {node {$a$}
        child {node {$ba$}
      }
      child {node {$aa$}
      }
    };
\end{scope}
\end{tikzpicture}
\end{center}
and the sum of all length-$k$ words is indeed generated by $G_{k,2} = (a+b)^k$. It is clear, by construction the sums of words of different lengths commute.

\noindent {\bf B. The Hecke algebra:} Consider for instance the Hecke algebra ${\cal H}_3(q);$ in this case, the left tree order diagram is shown in Example \ref{sym2} (2) ($a=t_1, b =t_2$). The sum of length-$k$ words is still given by $(a+b)^k,$ but subject to the Hecke algebra conditions, as described in Remark \ref{tree2}:
\begin{enumerate}
    \item Reduced words of length one generated by $a+b,$ i.e. $a,b$
    \item Reduced words of length two generated by $(a+b)^2,$ but due to the Hecke condition $a^2 =b^2 =1,$ i.e. the only reduced words of length two are $ab, ba.$ 
    \item Reduced words of length three: one can proceed by considering $(a+b)^3,$ and use again the Hecke algebra relations. But having excluded $a^2$ and $b^2$ in the previous order, we can just consider $(a+b)(ab +ba)$ and use the Hecke algebra conditions, so the only reduced word of length three is $aba =bab.$ This process terminates, because if we keep going we produce already existing reduced words. This is in accordance to the left weak diagram of Example \ref{sym2} and as expected in this case all possible reduced words are: $1,a,b,ab,ba,aba.$  
\end{enumerate}
{While one can analyze any linear combination of Hecke algebra generators, this manuscript focuses specifically on their sum. As demonstrated later, sums of braid group generators correspond to spin-chain-like Hamiltonians with specific boundary conditions (see for instance the final two sections and \cite{Doikous}). Although the general case of arbitrary linear combinations warrants detailed study, we reserve that investigation for future publications.}

The process described in the example above, which is a consequence of the construction of the tree order diagram, is quite tedious for long alphabets, i.e. for big values of $N$. Fortunately, due to the solution of the word problem for Coxeter groups by Matsumoto \cite{Matsu}, one can define a generating function of all reduced words for the Hecke algebra ${\cal H}_N,$ which we call the ``shuffle element'', as follows (see also for instance \cite{Lechner}).
\begin{defn} \label{shuf} ({\bf The shuffle element.})
Let ${\cal H}_N(q)$ be the Hecke algebra generated by $t_1,t_2,\ldots, t_{N-1}$ (Definition \ref{Hecke1}) and let 
$${\cal S}_k(\hat z) := 1 +\hat z t_k + \hat z^2t_{k-1} t_k + \ldots + \hat z^k t_1 t_2 \ldots t_k \in {\cal H}_N(q), ~~\hat z \in {\mathbb C}.$$  
The shuffle element for ${\cal H}_N(q)$ is defined as 
\begin{equation}
{\mathfrak y}_N(\hat z) := {\cal S}_{N-1}(\hat z) {\cal S}_{N-2}(\hat z) \ldots {\cal S}_2(\hat z) {\cal S}_1(\hat z) \in {\cal H}_N(q). \label{shuffle}
\end{equation}
\end{defn}
The power of $\hat z$ in the expansion of ${\mathfrak y}_N(\hat z)$ indicates the length of the words of that order, where the maximum length of a word is $\frac{N(N-1)}{2}$ as follows from the definition of ${\mathfrak y}_N$ (see for instance the first few cases, $N=2,3,4$ below). Hence, we may write in a compact form:
$${\mathfrak y}_N(\hat z) = 1+ \sum_{l=1}^{\frac{N(N-1)}{2}} \hat z^l {\mathfrak s}_l,$$
where ${\mathfrak s}_l=\underset{w_l\in {\cal H}_N }{\sum} w_{l},$ $w_l$ denotes reduced words of length $l$ in ${\cal H}_N(q).$ 

For a certain representation of the Hecke algebra and for special values of  $\hat z,$ ${\mathfrak y}_N$ becomes the $q$-symmetrizer (or anti-symmetrizer) as this will become transparent later in the manuscript, when it is shown that ${\mathfrak y}_N,$ in a certain representation, produces all possible permutations in 
$[N]$. The shuffle element can be defined using the opposite order, but we will consistently use this definition in this manuscript. We will come back to the shuffle element when discussing the construction of $q$-(anti)symmetric states in Theorem \ref{shuffle1}.

The shuffle element generates all $N!$ reduced words including the empty word (unit element) (see also Theorem \ref{shuffle1}). One can easily check  a few examples: 
\begin{enumerate}
\item $N=2,$: ${\mathfrak y}_2(z) = 1 + \hat z t_1.$

\item $N=3:$ ${\mathfrak y}_3(z) = 1 + \hat z( t_1 +t_2)+  \hat z^2(t_1t_2+t_2t_1) + \hat z^3 t_1t_2t_1.$  

\item $N=4:$ ${\mathfrak y}_4(z) = \big (1+ \hat z t_3 + \hat z^2  t_2t_3 + \hat z^3 t_1t_2t_3 \big )\big(1 + \hat z( t_1 +t_2)+  \hat z^2(t_1t_2+t_2t_1) + \hat z^3 t_1t_2t_1\big ).$ 
\end{enumerate}
Note that the word of maximum length for any $N$ is $w_{max} : = t_{1}t_{2} \ldots t_{N-1} t_{1} t_{2} \ldots t_{N-2} \ldots t_1t_2t_1$ with length $l_{max} = \frac{N(N-1)}{2}.$

\begin{conj} \label{conj1} Let ${\mathfrak s}_{k}$ be the sum of all reduced words of length $k$ in ${\cal H}_N(q),$ then ${\mathfrak s}_k {\mathfrak s}_l = {\mathfrak s}_l {\mathfrak s}_k,$ $1\leq k,l\leq \frac{N(N-1)}{2},$ or equivalently ${\mathfrak y}_N(\hat z) {\mathfrak y}_N(\hat z') = {\mathfrak y}_N(\hat z') {\mathfrak y}_N(\hat z),$ $\hat z,\hat z'\in {\mathbb C}.$
\end{conj}
Conjecture \ref{conj1} provides a generalized statement of ``quantum integrability'', given that ${\mathfrak y}_{N}(z)$ generates $\frac{N(N-1)}{2}$ mutually commuting quantities. However, one needs to check how many of these are not just powers of previous conserved quantities. A comparison with conserved quantities coming from Sklyanin's double row transfer matrix \cite{Sklyanin} would be the starting point of such an analysis (see also \cite{DoiSmo2} for an explicit construction of conserved quantities coming from the open transfer matrix in terms of the elements of the symmetric group). Simple examples for $N=2$ or $N=3$ verify the conjecture above. 
To prove the conjecture in general one should use the Hecke algebra relations repeatedly.

In the next section, we examine finite automata associated with certain Hecke algebra representations. 
Given that for any automaton to be described a set of states is needed, the choice of a specific representation is required.
We are also going to prove some interesting propositions associated with the shuffle element based on the choice of representation. 
Before we discuss the automata associated with Hecke algebras we briefly recall in the following subsections the quantum algebra $\mathfrak{U}_q(\mathfrak{gl}_n)$ as well as basic definitions on Young tableaux.

\subsection{The quantum algebra and centralizers}
We recall the definition of the quantum algebra ${\mathfrak U}_q(\mathfrak{gl}_n),$ \cite{Drinfeld, Jimbo, Jimbo2} and the duality between the Hecke algebra ${\cal H}_N(q)$ and ${\mathfrak U}_q({\mathfrak{gl}_n}).$

Let 
\be 
a_{ij} = 2 \delta_{ij} - (\delta_{i\ j+1}+ \delta_{i\ j-1}), ~~i,\ j\in \big \{1, \ldots , n \big \}, \nonumber
\ee 
 be the Cartan matrix of the affine Lie algebra
${{\mathfrak{sl}_n}}$\footnote{For the ${\mathfrak{sl}_{2}}$ case in particular \be a_{ij} =2\delta_{ij} -2 (\delta_{i1}\ \delta_{j2} 
+\delta_{i2}\ \delta_{j1}), ~~i,\ j \in \big \{ 1, 2\big \}\ee}. Also define: 
\begin{eqnarray} && [m]_{q} =\frac{q^{m} -q^{-m}}{q -q^{-1}}, ~~~[m]_{q}!= \prod_{k=1}^{m}\ [k]_{q},~~~[0]_{q}! =1 \non\\  
&&\left [ \begin{array}{c}
m \\
n \\ \end{array} \right  ]_{q} = 
\frac{[m]_{q}!}{[n]_{q}! [m-n]_{q}!}, ~~~m\geq n \geq 0. \nonumber
 \end{eqnarray}
\begin{defn} \label{defq} The quantum enveloping
algebra ${\mathfrak U}_q({\mathfrak{sl}_n})$ is the unital associative algebra over ${\mathbb C}$ generated by
the Chevalley-Serre generators $e_{i}$, $f_{i}$,
$q^{\frac{\pm {h_{i}}}{2}}$, $i\in [n-1]$ {\it and  relations:} 
\begin{eqnarray}
&& \Big [q^{\frac{\pm h_{i}}{2}},\ q^{\frac{\pm h_{j}}{2}} \Big]=0\, \qquad q^{\frac{h_{i}}{2}}\ f_{j}=
q^{\frac{1}{2}a_{ij}}f_{j}\ q^{\frac{h_{i}}{2}}\, \qquad q^{\frac{h_{i}}{2}}\ e_{j}
= q^{-\frac{1}{2}a_{ij}}e_{j}\ q^{\frac{h_{i}}{2}}, \non\\
&& \Big [f_{i},\ e_{j}\Big ] = \delta_{ij}\frac{q^{h_{i}}-q^{-h_{i}}}{q-q^{-1}},
~~~~i,j \in [n-1],
\label{1} \nonumber
\end{eqnarray}
and the $q$--deformed Serre relations 
\begin{eqnarray} 
&& \sum_{n=0}^{1-a_{ij}} (-1)^{n}
\left [ \begin{array}{c}
  1-a_{ij} \\
   n \\ \end{array} \right  ]_{q} 
\chi_{i}^{1-a_{ij}-n}\ \chi_{j}\ \chi_{i}^{n} =0, ~~~\chi_{i} \in \big \{e_{i},\ f_{i} \big \}, ~~~ i \neq j. \label{chev} \nonumber
\end{eqnarray}
 Recall, $\big [ \ , \big ]: \mathfrak{U}_q(\mathfrak{gl}_n) \times \mathfrak{U}_q(\mathfrak{gl}_n) \to \mathfrak{U}_q(\mathfrak{gl}_n),$ such that $(a,b )\mapsto ab - ba.$  
 \end{defn}

\begin{rem} \label{rem3q} The generators $e_{i}$, $f_{i}$, $q^{\pm h_{i}}$ for $i\in [n] $ form the algebra
${\mathfrak U}_q({\mathfrak{sl}_n}).$ 
Also, $q^{\pm h_{i}}=q^{\pm (\varepsilon_{i} -\varepsilon_{i+1})}$,  $i\in [n-1],$ where the elements
$q^{\pm \varepsilon_{i}}$ belong to ${\mathfrak U}_q({{\mathfrak{gl}}_n}) $. Recall that ${\mathfrak U}_q({{\mathfrak{gl}_n}})$ is obtained by adding to 
${\mathfrak U}_q({{\mathfrak{sl}_n}})$ the elements $q^{\pm \varepsilon_{i}}$ $i\in [n],$  so that $q^{\underset{1 \leq i \leq n}{\sum}\varepsilon_{i}}$ 
belongs to the center (see \cite{Jimbo}).
\end{rem}

We recall that $\big ( {\mathfrak U}_q({\mathfrak{gl}_n}) , \Delta, \epsilon, s\big )$ 
is a Hopf algebra over ${\mathbb C}$ equipped with \cite{Drinfeld, Jimbo}:
\begin{itemize}
\item A coproduct $\Delta:\  {\mathfrak U}_q({\mathfrak{gl}_n})
\to {\mathfrak U}_q({\mathfrak{gl}_n}) \otimes {\mathfrak U}_q({\mathfrak{gl}_n}),$ such that
\begin{eqnarray}
 && \Delta(\chi_i) = q^{- \frac{h_{i}}{2}} \otimes \chi_i + \chi _i\otimes q^{\frac{h_{i}}{2}}, ~~\chi_i \in \big  \{e_{i},\ f_{i}\big  \}, ~~i \in [n-1] \label{copba} \\
&&\Delta(q^{\pm\frac{\varepsilon_{i}}{2}}) = q^{\pm\frac{\varepsilon_{i}}{2}} \otimes
q^{\pm\frac{\varepsilon_{i}}{2}}, \quad i \in [n].\label{copbb} 
\end{eqnarray} 
The $l$ co-product 
$\Delta^{(l)}:\ {\mathfrak U}_q({\mathfrak{gl}_n})
\to {\mathfrak U}_q({\mathfrak{gl}_n})^{\otimes l}$ is defined by 
$\Delta^{(l)} = (\id \otimes \Delta^{(l-1)})\Delta =(\Delta^{(l-1)} \otimes \id)
\Delta.$

\item A co-unit $\epsilon: {\mathfrak U}_q({\mathfrak{gl}_n}) \to {\mathbb C},$ such that
\begin{equation}
\epsilon(e_j) = \epsilon(f_j) =0, ~~~~\epsilon(q^{\varepsilon_j}) =1. \nonumber
\end{equation}

\item An antipode $s:  {\mathfrak U}_q({\mathfrak{gl}_n}) \to  {\mathfrak U}_q({\mathfrak{gl}_n}),$ such that
\begin{equation}
s(q^{\varepsilon_i}) = q^{-\varepsilon_i} , ~~~s(\chi_i) =- q^{\frac{h_i}{2}} \chi_i q^{-\frac{h_i}{2}}, ~~~\chi_i \in \big \{ e_i,\ f_i\big \}.  \nonumber
\end{equation}

\end{itemize}

We recall the fundamental representation of 
${\mathfrak U}_q{(\mathfrak{gl}_n})$ \cite{Jimbo}
$\pi: {\mathfrak U}_q({\mathfrak{gl}_n})\to \End({\mathbb C}^n)$:
\begin{eqnarray} 
\pi(f_{i})= e_{i, i+1}, ~~~\pi(e_{i})=e_{i+1, i}, ~~~\pi(q^{\frac{\varepsilon_{i}}{2}}) = 
q^{\frac{e_{i,i}}{2}},~~ i \in [n-1]. \label{eval} 
\end{eqnarray}

We recall also the tensor representation of the $A$-type Hecke algebra \cite{Jimbo2}, given by
$\rho: {\cal H}_N(q) \to \End(({\mathbb C}^n)^{\otimes N}),$ such that 
\begin{equation}
t_j \mapsto {\mathrm g}_j:= 1_n \otimes 1_n \otimes \ldots 1_n \otimes \underbrace{{\mathrm g}}_{j, j+1} \otimes   1_n \ldots 1_n \otimes 1_n, \label{braidq1}
\end{equation}
where $1_n$ is the $n \times n$ identity matrix and
\begin{equation}
{\mathrm g} =\sum_{x \neq y\in X} \Big ( e_{x, y} \otimes e_{y,x}  - q^{-sgn(x-y)} e_{x,x} \otimes e_{y,y} \Big ) +q 1_{n^2}. \label{braidq}
\end{equation}
Indeed, the elements ${\mathrm g}_j$ satisfy the braid relations 
as well as the Hecke constraint. Moreover, the Hecke algebra ${\cal H}_N(q)$ is central to ${\mathfrak U}_q(\mathfrak{gl}_n)$ and vice versa, i.e.
\begin{equation}
\big[{\mathrm g}_j, \pi^{\otimes N}\Delta^{(N)}(y) \big ] =0, ~~~y\in {\mathfrak U}_q(\mathfrak{gl}_n),~~~ j \in [N-1] \label{symm1}
\end{equation}
which follows from the fact that $\big [ {\mathrm g}, \pi^{\otimes 2} \Delta(y)\big ] {=0},$ for all $y \in {\mathfrak U}_q(\mathfrak{gl}_n).$

This brief review on ${\mathfrak U}_q({\mathfrak{gl}_n})$ will be particularly 
useful for the findings of the subsequent sections.

\subsection{Young tableaux, combinatorics and representation theory.}
It is useful for the analysis of the following sections to recall basic material on Young tableaux as these are essential combinatorial objects that play a key role in representation theory of the symmetric group $S_N$ and $\mathfrak{gl}(n)$ \cite{FultonHarris, Fulton}.
Note that in this subsection we only provide essential information about Young tableaux needed for our discussion here. For more information in relation to Schur polynomials, plactic monoid and representation theory the interested reader is referred for instance to \cite{Fulton, FultonHarris, crystal}. Although, representation theory associated with ${\mathfrak U}_q(\mathfrak{gl}_n)$ will be inevitably discussed as will be transparent in Section 7.

Throughout this paper we denote $\lambda \vdash N$ a partition $\lambda= (\lambda_1, \lambda_2, \ldots,\lambda_k)$ of the positive integer $N$,
where $\lambda_i$ are weakly decreasing positive integers and ${|\lambda|:}=\underset{1\leq i\leq k}{\sum} \lambda_i,$  is the size of $\lambda$ and in general $|\lambda| = N.$ 

\begin{defn} \label{Young} Suppose $\lambda= (\lambda_1, \lambda_2, \ldots,\lambda_k)$ is a partition of $N$ where $k\geq 1$. The Young (or Ferrers) diagram of
shape $\lambda$ is an array of $N$ squares having $k$ rows with row $i$ containing $\lambda_i$
squares.
\end{defn}
\begin{exa} The partition $\lambda = (3, 2,  1)$ has a Young diagram as follows,
    $ $
 {\footnotesize   \begin{center}
\ytableausetup{centertableaux}
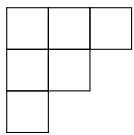
\begin{ytableau}
\ &  &  \\
\ &   \\
\
\end{ytableau}    
\end{center}}
\end{exa}

\begin{defn}
    A filling (or weight) of a Young diagram is any way of putting a positive integer
in each box of the diagram. Let $\mu = (\mu_1, \mu_2, . . . , \mu_l)$ be a filling of a Young diagram.
Each $\mu_i$ is the number of times the integer $i$ appears in the diagram.
\end{defn}
In order for any diagram to be completely filled, it is necessary for $|\lambda| = |\mu|,$ where {$|\mu| :=  \underset{1\leq j\leq l}\sum \mu_j $.}
\begin{exa}
    A possible tableau for $\lambda = (3, 2,  1)$ with filling $\mu = (3, 1, 2)$ is
    
 $ $
 
 {\footnotesize   \begin{center}
\ytableausetup{centertableaux}
\begin{ytableau}
1 & 2 & 3 \\
3 & 1 \\
1
\end{ytableau}    
\end{center}}
\end{exa}

It is possible to fill {diagrams} arbitrarily in this manner, however 
we impose certain restrictions on the filling $\mu$. These restrictions lead to
the definition of a Young tableau.

\begin{defn} \label{Young2}
Suppose $\lambda \vdash N.$ A Young tableau $T$ is obtained by filling in the boxes of the Young 
diagram with symbols taken from some alphabet, which is usually required to be a totally ordered set. 
A Young tableau of shape $\lambda$ is also called a $\lambda$-tableau. A Young tableau is standard if the rows and columns of $T$ are
increasing sequences. That is, $T$ is filled with the numbers $1, 2, \ldots ,N$ bijectively. A Young
tableau is semi-standard if the filling is weakly increasing across each row and
strictly increasing down each column.
\end{defn}
Henceforth, we use the shorthand notation SSYT and SYT for {semi-standard} and {standard} Young-tableau respectively. The dimension of any $\lambda$-SSYT with $N$ boxes is given by all possible fillings of $n$ distinct
integers for the given SSYT. The dimension of any $\lambda$-SYT is given by all possible arrangements of $N$ numbers in $N$ boxes following the rules of the SYT. It is also useful to introduce some notation at this point. Any sequence of the form,
\begin{equation}
    i_1 i_2 \ldots i_N, ~~\mbox{such that} ~~1\leq i_1\leq i_2 \leq \ldots\leq i_N\leq n, \label{s1}
    \end{equation}
    is called an {\it ordered} sequence, whereas the sequence  
    $i_1 i_2 \ldots i_N,$ such that $~1\leq i_1 <i_2 < \ldots < i_N\leq n,$ 
    is called {\it strictly ordered}.  The total number of ordered sequences for generic values of $n$ and $N$ is given by, 
\begin{equation}
 \mbox{$\#$ of ordered sequences} =  \frac{1}{N!} \prod_{k=n}^{N+n-1}k.   \label{dim1} \end{equation}
Ordered sequences are represented by SSYT of shape $\lambda =(N).$

We give a couple of indicative examples below.
\begin{exa} Consider the SSYT of shape $\lambda =(N)$ and let for instance $N=3,$

\begin{enumerate}
     
    \item  For n=2, there are four ordered states: ${\it 111, 112, 122,222}.$ \\
     In general, for $n=2,$ for any $N$ there are $N+1$ ordered sequences.


    \item For $n=3,$ there are ten ordered sequences: ${\it 111,112,122, 113},$ ${\it 133, 123, 222},$ ${\it 223,233,333 }$
\end{enumerate}

Indeed, the number (\ref{dim1}) is confirmed by the examples above.
\end{exa}

\begin{exa} Consider the SSYT of shape $\lambda = (N-1,1)$  and let $n=3, N=3,$ then the number of 
possible fillings is eight, i.e. 
{\footnotesize
\begin{ytableau}
1 & 1   \\
2
\end{ytableau}, \ytableausetup{centertableaux}
\begin{ytableau}
1 & 2   \\
2 
\end{ytableau}, 
\begin{ytableau}
1 & 2   \\
3
\end{ytableau}, 
\begin{ytableau}
1 & 1   \\
3
\end{ytableau},  
\begin{ytableau}
1 & 3   \\
3
\end{ytableau}, 
\begin{ytableau}
1 & 3   \\
2
\end{ytableau}, 
\begin{ytableau}
2 & 2   \\
3
\end{ytableau} and 
\begin{ytableau}
2 & 3   \\
3
\end{ytableau}.}

$ $

The number of possible fillings for the correspond SYT is just two: 
{\footnotesize
\begin{ytableau}
1 & 2   \\
3
\end{ytableau}, \ytableausetup{centertableaux}
\begin{ytableau}
1 & 3   \\
2
\end{ytableau}. }
\end{exa}
Due to the rule of strictly ascending order of numbers on each column of a SSYT or SYT the maximum number of rows is $n.$ For any $N$ there is only one possible SYT tableau for the diagram with $N$ rows and one column corresponding to the partition $\lambda = (\underbrace{1,1, \ldots, 1}_{N}),$ that is 
{\footnotesize
\begin{ytableau}
\   \\
\ \\
\vdots \\
\
\end{ytableau}.}

In general in representation theory, SYT of size $N$ correspond to irreducible representations of the symmetric group $S_N$, while irreducible representations of $\mathfrak{gl}_n$ (and ${\mathfrak U}_q(\mathfrak{gl}_n),$ $q$ not root of unity) are parametrized by SSYT of a fixed shape.
The number of SSYT of shape $\lambda$ and filling $\mu,$ is called {\it Kostka number} and is denoted $K_{\lambda,\mu}.$ For $\mathfrak{gl}_n$ each irreducible representation is uniquely determined by its highest weight, which is a SSYT of shape $\lambda$. 
The sum $\underset{\mu}{\sum}K_{\lambda, \mu}$, over all fillings, counts all possible SSYT of shape $\lambda$ and is also equal to the dimension of the irreducible representation of $\mathfrak{gl}_n$ with highest weight $\lambda$. 
\begin{rem} \label{r5}
Before we give the definition of basic automata associated with the braid group we introduce some notation. 
Let $X= \big \{x_1, x_2, \ldots, x_n\big \}$ and ${\mathbb C}^n$ be the $n$-dimensional vector space with basis 
${\mathbb B}_n = \big \{\hat e_{x_j}\big \},$ $j \in [n].$ 
Then $({\mathbb C}^n)^{\otimes N}$ is the $N$ tensor product vector space of dimension $n^N$ with basis ${\mathbb B}_n^{\otimes N}= \big \{\hat e_{x_{i_1}} \otimes \hat e_{x_{i_2}} \ldots \otimes \hat e_{x_{i_N}}\big \},$ $\hat e_{x_{i_k}} \in {\mathbb B}_n,$ $i_k \in [n]$ and $k \in [N].$ For the rest of the manuscript we mainly use the simplified notation: $ x_{i_1} \otimes x_{i_2} \ldots \otimes x_{i_N}:=\hat e_{x_{i_1}} \otimes \hat e_{x_{i_2}} \ldots \otimes \hat e_{x_{i_N}}.$
The standard ordering $x_1<x_2<\ldots < x_n$ is considered.
\end{rem}
 Recall also that a state $\Psi\in ({\mathbb C}^n)^{\otimes N}$ is called {\it factorized} if it can be expressed as $\Psi = a_1 \otimes a_2 \otimes \ldots \otimes a_N,$ $a_i \in {\mathbb C}^n,$ $i\in [N].$ Also a state of the form
$x_{i_1}\otimes x_{i_2}\otimes  \ldots \otimes x_{i_N},$ such that $1\leq i_1\leq i_2 \leq \ldots\leq i_N\leq n,$
    is called an {\it ordered} factorized state.  These states in fact form a crystal basis \cite{Kashi}, 
and are represented by the SSYT of shape $\lambda = (N)$ (see also for instance \cite{crystal, Kuniba1, Kuniba2} on a pedagogical exposition and later in this manuscript in the subsequent section).  There is obviously a one to one correspondence of ordered states to ordered sequences.

In order to construct tensor representations of $\mathfrak{gl}_n$ we basically use the rules of SSYT. 
For any $n$ and $N=1$ there is only one diagram of one box
{\tiny \begin{ytableau}
\
\end{ytableau}}
of dimension $n,$ i.e. {there are} $n$ possible fillings (the fundamental representation
{\footnotesize
\begin{ytableau}
1
\end{ytableau},
\begin{ytableau}
2
\end{ytableau}, \ldots, 
\begin{ytableau}
n
\end{ytableau}}).
This represents the $n$-dimensional vectors space $V_n$ (recall throughout this manuscript $V_n$ is either ${\mathbb C}^n$ or ${\mathbb R}^n$). To build higher tensor representations of $\mathfrak{gl}_n$ we basically add extra boxes to already existing SSYT. Extra boxes can be added to the right or the bottom of an existing tableaux so that a new $\lambda$-shaped Young diagram is created.\\ For instance,

{\tiny
\begin{ytableau} 
\
\end{ytableau}  
$\otimes$
\begin{ytableau}
\
\end{ytableau} = 
\begin{ytableau}
\ &
\end{ytableau} $\oplus$ \begin{ytableau}
\ \\
\ 
\end{ytableau}
 $\ $ or $\ $  \begin{ytableau}
\ &
\end{ytableau}  $\otimes$  \begin{ytableau}
\
\end{ytableau}   =  \begin{ytableau}
\ & &
\end{ytableau}  $\oplus$   \begin{ytableau}
\ & \\
\ 
\end{ytableau}.}\\
This way all possible Young tableaux can be generated as shown in the diagram below,
{\tiny 
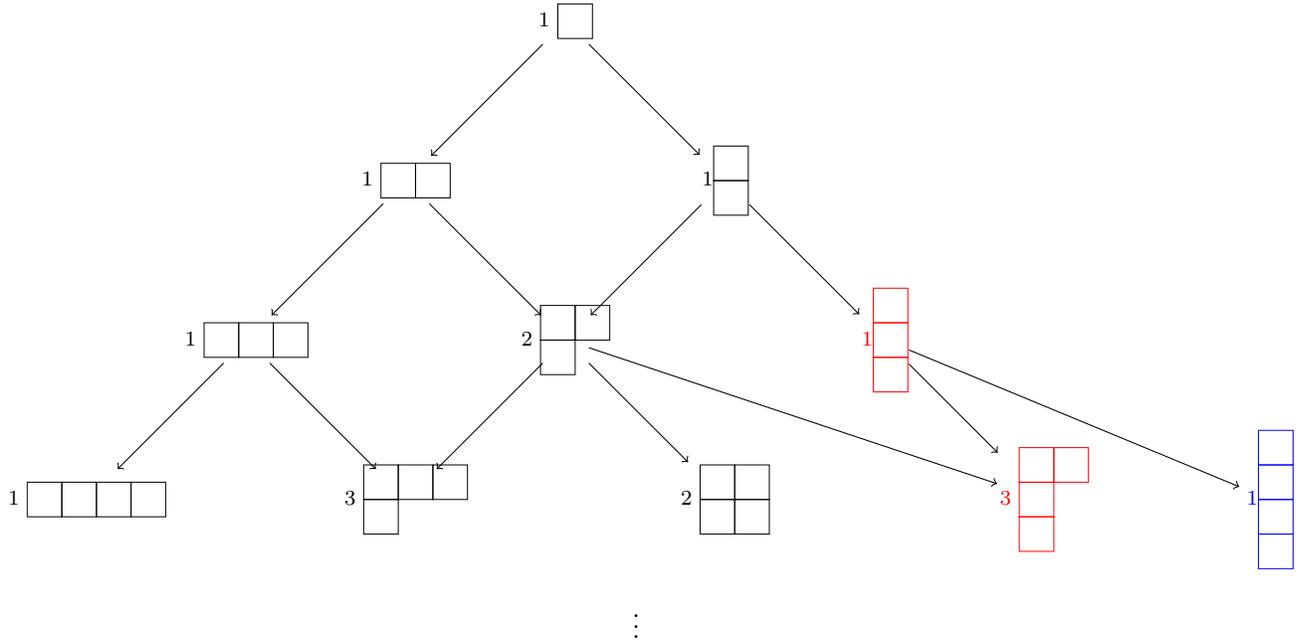
\begin{center}
\begin{tikzpicture}[shorten >=1pt,node distance=3.0cm,on grid,auto] 
   \node[inner sep=0, minimum size=2.0em,] (q_1)   {1 \begin{ytableau}
\ 
\end{ytableau}}; 
   \node[inner sep=0, minimum size=2.0em] (q_2) [below left=of q_1] {1 \begin{ytableau}
\ &
\end{ytableau}}; 
   \node[inner sep=0, minimum size=2.0em] (q_3) [below right=of q_1] {1\begin{ytableau}
\ \\
\
\end{ytableau}}; 
   \node[inner sep=0, minimum size=2.0em](q_4) [below left=of q_2] {1 \begin{ytableau}
\ & &
\end{ytableau}};
    \node[inner sep=0, minimum size=2.5em](q_5) [below right=of q_2] {2 \begin{ytableau}
\ &\\
\
\end{ytableau}};
     \node[inner sep=0, minimum size=2.0em](q_5) [below left=of q_3] {$ $};
      \node[ inner sep=0, minimum size=2.0em](q_6) [below right=of q_3] {\red{1}\red{\begin{ytableau}
\ \\
\ \\
\
\end{ytableau}}}; 
 \node[inner sep=0, minimum size=2.5em](q_7) [below left=of q_4] {1 \begin{ytableau}
\ & & &
\end{ytableau}};
 \node[inner sep=0, minimum size=2.5em](q_8) [below right=of q_4] {3 \begin{ytableau}
\ & &\\
\
\end{ytableau}};
 \node[inner sep=0, minimum size=2.5em](q_8) [below left=of q_5] { };    
  \node[inner sep=0, minimum size=2.5em](q_9) [below right=of q_5] {{2} {\begin{ytableau}
\ &\\
\ &\\
\end{ytableau}}}; 
 \node[inner sep=0, minimum size=2.5em](q_10) [below right=of q_6] {\red{3}  \red{\begin{ytableau}
\ & \\
\ \\
\ \\
\end{ytableau}}};
 \node[inner sep=0, minimum size=2.5em](q_11) [right =of q_10] {\blue{1}\blue{\begin{ytableau}
\ \\
\ \\
\ \\
\ \\
\end{ytableau}}};
  \path[->]
  (q_1) edge[left, below] node{$ $} (q_2)
   (q_1) edge[left, above] node{$ $ } (q_3)
   (q_2) edge[right, below] node{$ $} (q_4)
   (q_2) edge[right, above] node{$ $} (q_5)   
   (q_3) edge[left, below] node{$ $} (q_5)
   (q_3) edge[left, left] node{$ $} (q_6) 
 (q_4) edge[left, left] node{$ $} (q_7)
 (q_4) edge[left, left] node{$ $} (q_8)
 (q_5) edge[left, left] node{$ $} (q_8) 
 (q_5) edge[left, left] node{$ $} (q_9) 
 (q_5) edge[left, left] node{$ $} (q_10) 
 (q_6) edge[left, left] node{$ $} (q_10) 
 (q_6) edge[left, left] node{$ $} (q_11) ;
    \end{tikzpicture}
    \end{center}}
    \begin{center}
     $\vdots$
     \end{center}
Each horizontal level represents all the possible SSYT Young tableaux of $N$ boxes with the corresponding multiplicities given by the factors in front of each tableau. The multiplicities in front of each $\lambda$-shaped diagram in the figure above are equal to the dimensions of the irreducible representation of the symmetric group and are given by the number of all possible $\lambda$-SYT  
provided by the Hook length formula $m_{\lambda}= \frac{N!}{\prod_{i,j} h_{i,j}}.$ 
The dimension of each $\lambda$-SSYT is the dimension of an irreducible representation of $\mathfrak{gl}_n$.
Thus, the decomposition of tensor representations for $\mathfrak{gl}_n$ reads as
\[V_n^{\otimes N} = \bigoplus_{\lambda \vdash N}  V_{\lambda,n}^{\oplus m_{\lambda}},\]
where $\mbox{dim}V_{\lambda,n} = d_{\lambda,n}$ and is given by the dimension of the $\lambda$-shaped SSYT.

We should further note that the quantum group $\mathfrak{U}_q(\mathfrak{gl}_n)$ ($q$ not root of unity) has finite-dimensional irreducible representations in bijection with those of $\mathfrak{gl}_n$. Thus, they are also indexed by highest weights, which are here again identified with partitions $\lambda$, and are denoted $V_{\lambda,n}$ as in the classical situation. The vector representation
$V_n$ corresponds to $\lambda = (1).$ The Hecke algebra ${\cal H}_N(q)$ has irreducible representations in bijection with those of the symmetric group
$S_N$ and they are indexed by partitions $\lambda$ of size $N,$ denoted $S_{\lambda}$ ($\mbox{dim}S_\lambda = m_\lambda$). 
In general, the actions of ${\mathfrak{gl}_n}$ and  the symmetric group $S_N$ on $V_n^{\otimes N}$ commute (mutual centralizers), then the Schur–Weyl duality states that under the joint action of the symmetric group $S_N$ and $\mathfrak{gl}_n$ (or the actions of Hecke algebra ${\cal H}_N(q)$ and ${\mathfrak U}_q(\mathfrak{gl}_n),$ under the $q$-Schur-Weyl duality), the tensor space decomposes into a direct sum of tensor products of irreducible modules for $ \mathfrak{gl}_n$ and $S_N,$ i.e. $V_n^{\otimes N} = \underset{\lambda \vdash N}\sum V_{\lambda,n} \otimes S_{\lambda}$ (see also for instance, \cite{FultonHarris, Fulton, Dandecy}).

In this manuscript we are primarily interested in finite irreducible representations of $\mathfrak{U}_q(\mathfrak{gl}_n),$ so we will look at the decomposition of the tensor space from the point of view of the quantum algebra. We will come back to this point in Section 7 when studying finite irreducible representations of ${\mathfrak U}_q(\mathfrak{gl}_n)$ as eigenstates of a ${\mathfrak U}_q(\mathfrak{gl}_n)$ symmetric spin chain Hamiltonian. 

\section{q-permutation and quantum algebra automata: orthogonal combinatorial bases}
\subsection{The q-permutation automaton}
\noindent In this subsection, we focus on the fundamental representation (\ref{braidq}) of the Hecke algebra ${\cal H}_N(q)$ and we construct the associated braided finite automaton. {We also recall the ``shuffle" operator as an element of \(\text{End}(V_n^{\otimes N})\) and prove that it generates special sums of all possible permutations of any state \(\hat{e}_{i_1} \otimes \hat{e}_{i_2} \otimes \ldots \otimes \hat{e}_{i_N}\), where \(i_1 \leq i_2 \leq \ldots \leq i_N\) and \(\{\hat{e}_j\}_{j=1}^n\) is the standard basis of \(V_{n}\) (Theorem \ref{shuffle1}). The action of two special instances of the shuffle operator—namely, the \(q\)-symmetrizer and the \(q\)-anti-symmetrizer on specific tensor product states yields all \(q\)-(anti)symmetric states. These \(q\)-(anti)symmetric states are elements of a \(q\)-deformed Fock space, see Definition \ref{qFock} (cf. \cite{MisMiw, Lechner}.}

Let $X = \{x_1, x_2, \ldots, x_n\},$ such that $x_1 <x_2 \ldots <x_n,$  we define the rescaled braid operator, 
\begin{equation}
    r: = q^{-1}{\mathrm g} = \sum_{a\in X} e_{a,a} \otimes e_{a,a} + q^{-1}\sum_{a\neq b \in X} e_{a,b} \otimes 
    e_{b,a} + (1-q^{-2}) \sum_{a>b\in X} e_{a,a} \otimes e_{b,b}, \label{rr}
\end{equation}
where ${\mathrm g}$ is defined in (\ref{braidq}).
Moreover, we define 
\begin{equation}
r_j : = 1_n \otimes \ldots 1_n \otimes \underbrace{r}_{j, j+1} \otimes 1_n \ldots \otimes 1_n, ~~~j\in [N-1]. \label{rr2}
\end{equation}
The elements $r_j$  satisfy the braid relations and the quadratic relation, $r_i^2 = (1-q^{-2})r_i + q^{-2}1_{n^2}.$ We also note that $r^T = r,$ where $^T$ denotes transposition.
The action of $r$ on the tensor product of the canonical basis $\big \{\hat e_x\big \},$ $x\in X$ is given as (recall the shorthand notation, $x \otimes y : = \hat e_x \otimes \hat e_y$)
 \begin{eqnarray}
 r\ x \otimes y =\Bigg  \{ \begin{matrix} & x\otimes y, ~~~~~~x=y,\\ 
 & q^{-1}\ y \otimes x,~~~x<y
 \\ & q^{-1}\ y\otimes x + (1-q^{-2})\ x\otimes y, ~~~ x>y. \end{matrix}, \label{matrix0}
\end{eqnarray}
Recall that throughout the manuscript we consider $q = e^{\mu},$ $\mu \in {\mathbb R}.$
In general, for $x_{i_j} \in X,$ $j\in [N-1],$ 
\begin{equation} r_j\ x_{i_1} \ldots \otimes\underbrace {x_{i_j}\otimes x_{i_{j+1}}}_{j, j+1\ positions} \ldots \otimes x_{i_N} =
\Bigg  \{ \begin{matrix} & \ldots \otimes x_{i_{j}} \otimes x_{i_{j+1}} \ldots, ~~~~~~~~~~~x_{i_j}= x_{i_{j+1}},\\ 
&  q^{-1} \ldots\otimes x_{i_{j+1}}\otimes x_{i_j}\ldots , ~~~~~~x_{i_j}< x_{i_{j+1}}, 
\\ 
&  q^{-1} \ldots \otimes x_{i_{j+1}}\otimes x_{i_j} \ldots + (1-q^{-2}) \ldots\otimes x_{i_j}\otimes x_{i_{j+1}} \ldots, x_{i_j}> x_{i_{j+1}} \end{matrix}  \\ \label{braidg}\end{equation}
where $\ldots \otimes x_{i_j}\otimes x_{i_{j+1}}\otimes \ldots, \  \ldots \otimes  x_{i_{j+1}} \otimes
x_{i_j} \ldots\ \in {\mathbb B}_n^{\otimes N}$ and ${\mathbb B}_n = \big \{\hat e_j \big \},$ $j \in [n].$ 
\begin{defn} \label{auto11} {\bf (The $q$-permutation or $q$-flip automaton).}   
Let ${\mathbb B}_{n} = \big \{\hat e_{x_i}\big \},$ $x_i \in X.$

\begin{enumerate}
    

\item Let the set of states be $Q = {\mathbb B}^{\otimes N}_n$ (see Remark \ref{r5}, i.e. $Q$ consist of $n^N$ states). 

\item Let the alphabet be $\Sigma : =\big \{s_1, s_2, \ldots, s_{N-1}\big \}.$ The respective transition matrices are $r_i,$ $ i \in [N-1]$ given by the tensorial representation of the Hecke algebra, and their action on the states is given in (\ref{braidg}).
\end{enumerate}
This is called the $q$-permutation (or $q$-flip) automaton.
\end{defn}
In general, if the transition matrices of a finite automaton satisfy the relations of the braid group then we call the automaton a {\it braided automaton}.
The $q$-flip automaton is a braided automaton and is also non-deterministic due to the action of the transition matrices on the states (\ref{braidg}).
If $q=1,$ or $q=0$ then the $q$-flip automaton becomes combinatorial.

In the $q$-flip  automaton we usually consider an initial state to be an ordered state.
Also, in all automaton diagrams that follow we omit for brevity the symbol $\otimes $ in any state $a \otimes b \otimes c \ldots $ and simply write $abc \ldots$

\begin{exa}
We consider the first non-trivial examples of the $q$-flip braid automaton $N=2, 3.$ 
\begin{enumerate}
\item $N=2,$ and generic $n:$ $Q = \big \{x\otimes y\big \},$ $x, y \in X$ ($n^2$ states), $\Sigma = \big \{s_1\big \}$ and the transition matrix is $r_1.$ We consider as an initial state any state $ a\otimes b,$ $a\leq b \in X $, and set $\hat c:= 1 -q^{-2}$:
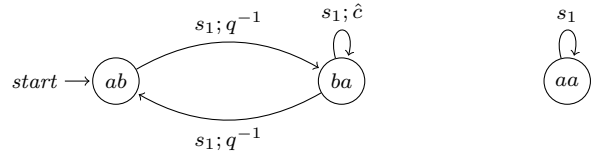
\begin{center}
\begin{figure}[ht]
\footnotesize
\begin{tikzpicture}[shorten >=1pt,node distance=3.0cm,on grid,auto] 
\node[state, inner sep=0, minimum size=2.0em, initial] (q1) {${ab}$};
\node[state,inner sep=0, minimum size=2.0em, right of=q1] (q2) {${ba}$};
\node[state,inner sep=0, minimum size=2.0em, right of=q2] (q3) {${ aa}$};
\path[->] 
(q1) edge[bend left, above] node{$s_1;q^{-1}$} (q2)
(q2) edge[bend left, below] node{$s_1;q^{-1}$} (q1)
(q2) edge[loop above] node{$s_1; \hat c$} (q2)
(q3) edge[loop above] node{$s_1$} (q3);
\end{tikzpicture}
\end{figure}
\end{center}

\item $N=3,$ generic $n:$ the set of states is $ Q =\big \{x\otimes y\otimes w\big \},$ $x,y,w \in X$ ($n^3$ states), 
$\Sigma = \big \{s_1,s_2\big \}$ and the transition matrices are $r_i,$ $i \in [2].$ 
We consider here as an initial state any state $a\otimes b\otimes c,$ $a \leq b \leq c \in X$ ($n\geq 3$)
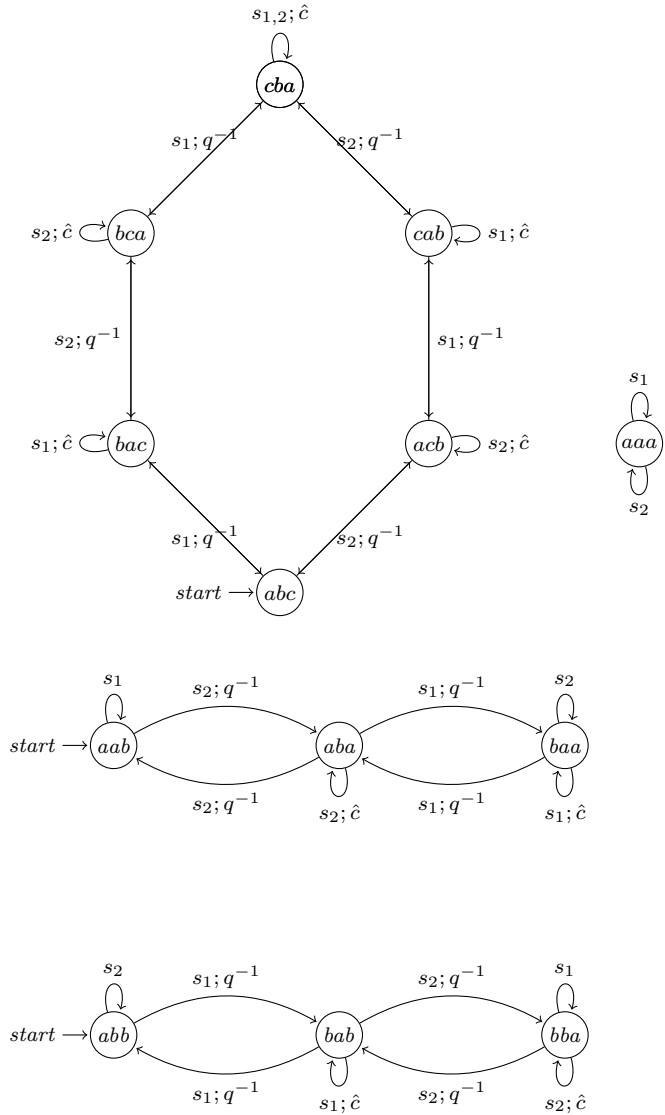
\begin{center}
\footnotesize
\begin{tikzpicture}[shorten >=1pt,node distance=2.8cm,on grid,auto] 
   \node[circle,draw, inner sep=0, minimum size=2.0em,initial] (q_1)   {${ abc}$}; 
   \node[circle,draw, inner sep=0, minimum size=2.0em] (q_3) [above right=of q_1] {${ acb}$}; 
   \node[circle,draw, inner sep=0, minimum size=2.0em] (q_2) [above left=of q_1] {${ bac}$}; 
   \node[circle,draw, inner sep=0, minimum size=2.0em](q_4) [above=of q_2] {${ bca}$};
 \node[circle,draw, inner sep=0, minimum size=2.0em](q_5) [above=of q_3] {${cab}$}; 
  \node[circle,draw, inner sep=0, minimum size=2.0em](q_6) [above right=of q_4] {${cba}$};
  \node[circle,draw, inner sep=0, minimum size=2.0em](q_7) [above left=of q_5] {${cba}$};
  \node[circle,draw, inner sep=0, minimum size=2.0em](q_0) [right=of q_3] {${aaa}$};
  \path[->]
  (q_1) edge[left, below] node{$s_1;q^{-1}$} (q_2)
  (q_2) edge[loop left] node{$s_1; \hat c$} (q_2)   
  (q_2) edge[left, above] node{$ $ } (q_1)
   (q_1) edge[right, below] node{$~~~~s_2;q^{-1}$} (q_3)
   (q_3) edge[right, above] node{$ $} (q_1)   
   (q_3) edge[loop right] node{$s_2; \hat c $} (q_3)  
   (q_4) edge[left, below] node{$ $} (q_2)
   (q_4) edge[loop left] node{$s_2;\hat c$} (q_4) 
   (q_2) edge[left, left] node{$s_2;q^{-1}$} (q_4)   
   (q_5) edge[right, below] node{$ $} (q_3)
   (q_5) edge[loop right] node{$s_{1};\hat c$} (q_5)   
   (q_3) edge[right, right] node{$s_1;q^{-1}$} (q_5)   
     (q_4) edge[left, above] node{$s_1;q^{-1}$} (q_6)
    (q_6) edge[left, below] node{$ $} (q_4)
    (q_5) edge[right, above] node{$~~~~s_2;q^{-1}$} (q_7)
    (q_7) edge[loop above] node{$s_{1,2}; \hat c$} (q_7)
    (q_0) edge[loop above] node{$s_1$} (q_0)    
     (q_0) edge[loop below] node{$s_2$} (q_0)   
     (q_7) edge[right, below] node{$ $} (q_5);    
    \end{tikzpicture}
    \end{center}
    \begin{center}
\begin{figure}[ht]
\footnotesize
\begin{tikzpicture}[shorten >=1pt,node distance=3.0cm,on grid,auto] 
\node[state, inner sep=0, minimum size=2.0em, initial] (q1) {${\it aab}$};
\node[state,inner sep=0, minimum size=2.0em, right of=q1] (q2) {${\it aba}$};
\node[state,inner sep=0, minimum size=2.0em, right of=q2] (q3) {${\it baa}$};
\path[->] 
(q1) edge[loop above] node{$s_1$} (q1)
(q2) edge[loop below] node{$s_2; \hat c$} (q2)
(q1) edge[bend left, above] node{$s_2;q^{-1}$} (q2)
(q2) edge[bend left, below] node{$s_2;q^{-1}$} (q1)
(q2) edge[bend left, above] node{$s_1;q^{-1}$} (q3)
(q3) edge[bend left, below] node{$s_1;q^{-1}$} (q2)
(q3) edge[loop below] node{$s_1;\hat c$} (q3)
(q3) edge[loop above] node{$s_2$} (q3);
\end{tikzpicture}
\end{figure}
\end{center}
 \begin{center}
\begin{figure}[ht]
\footnotesize
\begin{tikzpicture}[shorten >=1pt,node distance=3.0cm,on grid,auto] 
\node[state, inner sep=0, minimum size=2.0em, initial] (q1) {${\it abb}$};
\node[state,inner sep=0, minimum size=2.0em, right of=q1] (q2) {${\it bab}$};
\node[state,inner sep=0, minimum size=2.0em, right of=q2] (q3) {${\it bba}$};
\path[->] 
(q1) edge[loop above] node{$s_2$} (q1)
(q1) edge[bend left, above] node{$s_1;q^{-1}$} (q2)
(q2) edge[bend left, below] node{$s_1;q^{-1}$} (q1)
(q2) edge[loop below] node{$s_1; \hat c$} (q2)
(q2) edge[bend left, above] node{$s_2;q^{-1}$} (q3)
(q3) edge[bend left, below] node{$s_2;q^{-1}$} (q2)
(q3) edge[loop below] node{$s_2;\hat c$} (q3)
(q3) edge[loop above] node{$s_1$} (q3);
\end{tikzpicture}
\end{figure}
\end{center}
\end{enumerate}
Notice the correspondence between the first of the diagrams in both automata above and the weak Bruhat order 
for the symmetric groups $S_2$ and $S_3$ respectively (see Example \ref{sym2}). For $q=1$ the deterministic flip automaton is recovered.  
\end{exa}
The $q$-flip automaton provides in fact clusters of eigenstates of certain open spin chain Hamiltonians. Specifically, each disconnected graph in the $q$-flip automaton indicates the general structure of eigenstates of such spin chain Hamiltonians (see Sections 7 and 8). The eigenstates are linear combinations of all the factorized states contained in each disconnected graph in the $q$-flip automaton. 
We will define next the shuffle operator (see Theorem \ref{shuffle1}) which is a representation of the shuffle element and encodes part of the information provided by the $q$-flip automaton, i.e. given an ordered factorized state $x_{i_1} \otimes x_{i_2} \otimes \ldots \otimes x_{i_N},$ the shuffle operator produces a state that is the linear combination of all possible permutations of $i_1...i_N.$  
 
Before we formulate the following theorem we introduce some notation, we define:
\begin{equation}
[[m]]_{\xi} = \frac{\xi^{2m} - 1}{\xi^2 -1}, ~~~~ [[m]]_\xi! : = \prod_{k=1}^m [[k]]_{\xi}, ~~~~[[0]]_\xi! =1.
\end{equation}
Also, henceforth we usually write for brevity, $\underbrace{x_{i_1} \ldots x_{i_1}}_{k_1} \otimes\underbrace{x_{i_2}\ldots x_{i_2}}_{k_2} \ldots \otimes \underbrace{x_{i_m} \ldots {x_{i_m}}}_{k_m}$ instead of\\ $\underbrace{x_{i_1} \otimes  \ldots \otimes  x_{i_1}}_{k_1} \otimes  \ldots \otimes \underbrace{x_{i_m} \otimes  \ldots \otimes  {x_{i_m}}}_{k_m}.$
The set of all possible permutations in $\underbrace{x_{i_1} \ldots x_{i_1}}_{k_1} \otimes\underbrace{x_{i_2}\ldots x_{i_2}}_{k_2} \ldots \otimes \underbrace{x_{i_m} \ldots {x_{i_m}}}_{k_m} $ for fixed values $i_1, i_2,\ldots, i_m \in [n]$ and a fixed arrangement $k_1, k_2, \ldots, k_m$ 
is denoted\\ ${\mathfrak S}_{\hat i}(N, k_1, k_2, \ldots, k_m),$ where $\hat i : = \big \{i_1, i_2, \ldots, i_m\big \}$ and ${\mathfrak S}_{\hat i}(N,\underbrace{1,1, \ldots 1}_{N}) = :{\mathfrak S}_{\hat i}(N),$ (i.e. $m=N$; if $n=N,$ ${\mathfrak S}_{\hat i}(N) =: {\mathfrak S}_N$). The cardinality of ${\mathfrak S}_{\hat i}(N, k_1, k_2, \ldots, k_m)$ is given by $\frac{N!}{k_1! k_2! \ldots k_m!},$ which is the number of all possible arrangements of $k_j$ objects (balls) of type $j\in [m],$ in $N$ distinct spots (boxes),  $\underset{1\leq j\leq m}{\sum}k_j = N.$

\begin{thm} \label{shuffle1} ({\bf The shuffle operator.}) Recall the quantities ${\cal S}_k$ and ${\mathfrak y}_N$ given in Definition \ref{shuf} and 
let $\rho: {\cal H}_N(q) \to \End(({\mathbb C}^n)^{\otimes N}),$ such that $t_i \mapsto q r_i,$ (\ref{rr}), (\ref{rr2}), then ${\cal S}_k \mapsto {\mathrm S}_k$ and ${\mathfrak y}_N \mapsto {\mathrm Y}_N$:
\begin{equation}
{\mathrm S}_k(z) = 1 + z r_k + z^2 r_{k-1} r_k +  \ldots + z^k r_1 r_2\ldots r_k, ~~\mbox{and}~~~{\mathrm Y}_N(z) = {\mathrm S}_{N-1}(z){\mathrm S}_{N-2}(z)\ldots {\mathrm S}_1(z), ~~z\in {\mathbb C}. \nonumber
\end{equation}
Let also, $x_{i_1}<x_{i_2} <\ldots < x_{i_m},$ $i_1, i_2, \ldots, i_m \in [n]$ and $1\leq k_j \leq N,$ $j\in[m],$ then
\begin{eqnarray}
    &&{\mathrm Y}_N(z)\ \underbrace{x_{i_1} \ldots x_{i_1}}_{k_1}\otimes \underbrace{x_{i_2}\ldots x_{i_2}}_{k_2} \ldots \otimes  \underbrace{x_{i_m} \ldots  {x_{i_m}}}_{k_m}  =\nonumber\\ &&  [[k_1]]_{\zeta}! [[k_2]]_{\zeta}! \ldots [[k_m]]_{\zeta}!\sum_{l=0}^{l_{max}}\sum_{{\mathrm P}^{(l)} \in {\mathfrak S}_{\hat i}(N,k_1,k_2,\ldots, k_m)} z^lq^{-l}{\mathrm P}^{(l)}[ \underbrace{x_{i_1} \ldots x_{i_1}}_{k_1} \otimes \underbrace{x_{i_2}\ldots x_{i_2}}_{k_2} \ldots \otimes \underbrace{x_{i_m} \ldots {x_{i_m}}}_{k_m}] \label{bstate} \nonumber\\
\end{eqnarray}
where $\zeta = z^{\frac{1}{2}}$  and $k_1+k_2 + \ldots+k_m =N,$ $l$ is the length of the corresponding word (i.e. the corresponding permutation) and $l_{max} =k_1k_2 +(k_1 +k_2)k_3 + \ldots +(k_1 +k_2 +\ldots +k_{m-1})k_m$.\\ That is, the action  of ${\mathrm Y}_N(z)$ on 
$\underbrace{x_{i_1} \ldots x_{i_1}}_{k_1}\otimes \underbrace{x_{i_2}\ldots x_{i_2}}_{k_2} \ldots \otimes \underbrace{x_{i_m} \ldots x_{i_m}}_{k_m} $ generates specific sums of all $\frac{N!}{k_1!k_2! \ldots k_m!}$ possible permutations of $\underbrace{x_{i_1} \ldots x_{i_1}}_{k_1} \otimes \underbrace{x_{i_2}\ldots x_{i_2}}_{k_2} \ldots \otimes\underbrace{x_{i_m} \ldots x_{i_m}}_{k_m}.$
\end{thm}
\begin{proof}
The proof of the statement is given by induction.  
\begin{itemize}
\item We will first prove the case $m=N,$ $k_1=k_2 =\ldots = k_m =1$
\begin{enumerate}
   \item We first prove that the statement holds for $N=2$ ($N=1$ is trivial as it only contains the unit element), 
   let $x_i<x_j \in X,$ then
   $${\mathrm Y}_{2}(z)\  x_i\otimes x_j =  x_i\otimes x_j+ zq^{-1}\  x_j\otimes x_i.$$ 
   In fact, we can also easily show the statement for $N=3,$ (see also the first graph in the permutation automaton 
   above for $N=3$), let $x_i<x_j<x_k \in [n],$ then
\begin{eqnarray} 
 {\mathrm Y}_3(z)\ x_i \otimes x_j \otimes x_k
&=& x_i\otimes x_j \otimes x_k +zq^{-1}(x_j\otimes x_i\otimes x_k+ x_i\otimes x_k\otimes x_j) \nonumber\\ & +& z^2q^{-2} (x_j \otimes x_k \otimes x_i + x_k\otimes x_i \otimes x_j) +z^3q^{-3} x_k\otimes x_j \otimes x_i.\nonumber
\end{eqnarray}
  
\item Assume that  (\ref{bstate}) is true for $N-1$ with $l_{max} = \frac{(N-1)(N-2)}{2}$, we will then show that (\ref{bstate}) holds also for $N,$ where $l_{max} = \frac{N(N-1)}{2}.$

 We first observe that ${\mathrm Y}_N = S_{N-1} {\mathrm Y}_{N-1}$, 
   then
   \begin{eqnarray}
       && {\mathrm Y}_N(z)\ x_{i_1} \otimes x_{i_2} \ldots \otimes x_{i_N} = {\mathrm S}_{N-1}(z) {\mathrm Y}_{N-1}(z)\  x_{i_1} \otimes  x_{i_2} \ldots \otimes  x_{i_N} =\nonumber\\
       && {\mathrm S}_{N-1}(z) \sum_{l=0}^{\frac{(N-1)(N-2)}{2}}\sum_{f^l_{i_1}, \ldots f^l_{i_{N-1}} \in {\mathfrak S}_{
       \hat i}(N-1)} z^lq^{-l}\ x_{f^l_{i_1}} \otimes x_{f^l_{i_2}}\ldots  \otimes x_{f^l_{i_{N-1}} }\otimes  x_{{i_N} }.  \label{final2} 
   \end{eqnarray}
   Each term of order $z^k$ in ${\mathrm S}_{N-1}$ moves the last element $x_{i_N},$ in every factorized state\\ 
   $x_{f^l_{i_1}} \otimes x_{f^l_{i_2}}\ldots  \otimes  x_{i_N},$ $k$ positions to the left ($N-1$ being the maximum move to the left), hence any length-$l$ word now becomes a length-$l+k$ word. Moreover, the maximum length of a word becomes $\frac{(N-1)(N-2)}{2} + N-1 = \frac{N(N-1)}{2}$ and the total number of terms in (\ref{final2}) is now $(N-1)!N = N!.$ Thus, we conclude 
   \begin{eqnarray}
       {\mathrm Y}_N(z)\ x_{i_1} \otimes x_{i_2} \otimes \ldots  \otimes x_{i_N} 
       &=&\sum_{l=1}^{\frac{N(N-1)}{2}}\sum_{f^l_{i_1}, \ldots f^l_{i_N} \in {\mathfrak S}_{\hat i}(N)} z^lq^{-l}\ x_{f^l_{i_1}} \otimes \ldots \otimes x_{f^l_{i_{N-1}}} \otimes  x_{{f^{l}_{i_N}} }. \label{final2b} 
   \end{eqnarray}
\end{enumerate}

\item We also show that (\ref{bstate}) holds for the Grassmannian case $m=2.$ 
\begin{enumerate}
\item We first show that (\ref{bstate}) holds for $k_1 = N-1,\ k_2 =1$. 
\begin{eqnarray}
&&{\mathrm Y}_N(z) \ \underbrace{x_{i_1}\ldots x_{i_1}}_{N-1} \otimes x_{i_2} = [[N-1]]_{\zeta}! {\mathrm S}_{N-1}(z) \ \underbrace{x_{i_1}\ldots x_{i_1}}_{N-1} \otimes  x_{i_2} = \nonumber\\ 
&& [[N-1]]_{\zeta}! \big (\underbrace{x_{i_1}\ldots x_{i_1}}_{N-1} \otimes x_{i_2} + zq^{-1}\ \underbrace{x_{i_1}\ldots x_{i_1}}_{N-2} \otimes x_{i_2} \otimes x_{i_1} + \ldots + 
(zq^{-1})^{N-1}x_{i_2} \otimes\underbrace{x_{i_1}\ldots x_{i_1}}_{N-1}\big ).\end{eqnarray}

\item We assume that (\ref{bstate}) holds for any $N-1$ and $k_1, k_2-1,$ with $l_{max} = k_1(k_2 -1),$ we show that it also holds for $N$ and $k_1,k_2$ and $l_{max} = k_1 k_2.$
\begin{eqnarray}
&&  {\mathrm Y}_N(z)\ \underbrace{x_{i_1} \ldots x_{i_1}}_{k_1} \otimes\underbrace{x_{i_2}\ldots x_{i_2}}_{k_2}  = 
 {\mathrm S}_{N-1}(z) {\mathrm Y}_{N-1}(z)\ \underbrace{x_{i_1} \ldots x_{i_1}}_{k_1}\otimes  \underbrace{x_{i_2}\ldots x_{i_2}}_{k_2} =\nonumber\\
&&  [[k_1]]_{\zeta}! [[k_2-1]]_{\zeta}! {\mathrm S}_{N-1}(z) \sum_{l=0}^{k_1(k_2-1)}\sum_{{\mathrm P}^{(l)} \in {\mathfrak S}_{\hat i}(N-1,k_1,k_2-1)} z^lq^{-l}{\mathrm P}^{(l)}[\underbrace{x_{i_1} \ldots x_{{i_1}}}_{k_1} \otimes  \underbrace{x_{i_2} \ldots x_{{i_2}}}_{k_2-1}] \otimes x_{i_2} = \nonumber\\
&& [[k_1]]_{\zeta}! [[k_2]]_{\zeta}! \sum_{l=0}^{k_1k_2}\sum_{{\mathrm P}^{(l)} \in {\mathfrak S}_{\hat i}(N,k_1,k_2)} z^lq^{-l}{\mathrm P}^{(l)}[\underbrace{x_{i_1} \ldots x_{{i_1}}}_{k_1} \otimes \underbrace{x_{i_2} \ldots x_{{i_2}}}_{k_2}], \label{2b} \end{eqnarray}
\end{enumerate}
where we have used the action (\ref{braidg}).  Also,
for $\underbrace{x_{i_1} \ldots x_{{i_1}}}_{k_1} \otimes  \underbrace{x_{i_2} \ldots x_{{i_2}}}_{k_2-1} \otimes x_{i_2}$ the maximum length  (permutation) is $\underbrace{x_{i_2} \ldots  x_{i_2}}_{k_2-1} \otimes\underbrace{x_{i_1} \ldots x_{i_1}}_{k_1} \otimes x_{i_2}$ and the respective maximum length is $k_1(k_2-1)$. The maximum length of the move for the far right $x_{i_2}$  in $\underbrace{x_{i_2} \ldots  x_{i_2}}_{k_2-1} \otimes\underbrace{x_{i_1} \ldots x_{i_1}}_{k_1} \otimes x_{i_2}$ to reach the last $x_{i_2}$ on the left is $k_1,$ so the maximum length permutation for $\underbrace{x_{i_1} \ldots x_{{i_1}}}_{k_1} \otimes  \underbrace{x_{i_2} \ldots x_{{i_2}}}_{k_2}$ is $\underbrace{x_{i_2} \ldots  x_{i_2}}_{k_2}  \otimes \underbrace{x_{i_1} \ldots x_{i_1}}_{k_1}$ and its length is $k_1(k_2 -1) + k_1 = k_1 k_2.$ There is also an overall $[[k_2]]_{\zeta}$ factor as we go from line two to line three in (\ref{2b}). To see how this overall factor emerges we focus on the first term $\underbrace{x_{i_1} \ldots x_{i_1}}_{k_1} \otimes \underbrace{x_{i_2} \dots x_{i_2}}_{k_2}$:
\begin{eqnarray}
&& {\mathrm S}_{N-1}(z)\big ( \underbrace{x_{i_1} \ldots x_{i_1}}_{k_1} \otimes \underbrace{x_{i_2} \dots x_{i_2}}_{k_2} + q\ \underbrace{x_{i_1} \ldots x_{i_1}}_{k_1-1}\otimes x_{i_2} \otimes x_{i_1} \otimes \underbrace{x_{i_2} \dots x_{i_2}}_{k_2-1} +\ldots \big ) =\nonumber\\
&& (1 +z +z^2 + \ldots + z^{k_2-1})\  \underbrace{x_{i_1} \ldots x_{i_1}}_{k_1} \otimes \underbrace{x_{i_2} \dots x_{i_2}}_{k_2}  + \ldots = [[k_2]]_{\zeta}\ \underbrace{x_{i_1} \ldots x_{i_1}}_{k_1} \otimes \underbrace{x_{i_2} \dots x_{i_2}}_{k_2} + \ldots  \nonumber
\end{eqnarray}
To identify the overall factor $[[k_2]]_{\zeta}$ we could have focused on any permutation of $\underbrace{x_1 \ldots x_1}_{k_1} \otimes \underbrace{x_2 \dots x_2}_{k_2}$ and use similar arguments.
\end{itemize}
The proof can be generalized in an analogous way for the general case $k_1, k_2, \ldots, k_m$ ($k_1 + k_2 + \ldots  k_m =N$.)
We note that in the general case the maximum length word in $Y_N(z)\ \underbrace{x_{i_1} \ldots x_{i_1}}_{k_1}\ldots \otimes  \underbrace{x_{i_m} \ldots {x_{i_m}}}_{k_m}$ is\\ $\underbrace{x_{i_m} \ldots x_{i_m}}_{k_m} \otimes \underbrace{x_{i_{m-1}}\ldots x_{i_{m-1}}}_{k_{m-1}} \ldots \otimes  \underbrace{x_{i_1} \ldots {x_{i_1}}}_{k_1},$ so combinatorially one can compute the length of the permutation starting from the ordered state $\underbrace{x_{i_1} \ldots x_{i_1}}_{k_1}  \otimes\underbrace{x_{i_2}\ldots x_{i_2}}_{k_2} \ldots \otimes \underbrace{x_{i_m} \ldots {x_{i_m}}}_{k_m},$ which is $ l_{max} = k_1k_2 + (k_1 +k_2)k_3 +\ldots +(k_1+k_2 +\ldots +k_{m-1})k_m.$
\end{proof}
The proof of Theorem \ref{shuffle1} is quite descriptive given that the arguments used are mostly combinatorial. The proof of the first part of the proposition can be presented diagrammatically via a tree graph. Indeed, the action of $Y_N(z)$ on $x_{i_1}\otimes x_{i_2} \otimes \ldots \otimes x_{i_N},$ $x_{i_1}< x_{i_2}< \ldots <x_{i_N}$ is graphically depicted as a tree diagram below.
The length of each word is determined by the power $\xi:=zq^{-1}$ after multiplying the coefficients $\xi^k$ along each path in the diagram.
In the following tree diagram instead of indicating a generic state $ x_{j_1} \otimes x_{j_2} \otimes \ldots \otimes x_{j_N},$ we simply write $j_1 j_2 \ldots j_n \in {\mathfrak S}_{\hat i}(N).$ 

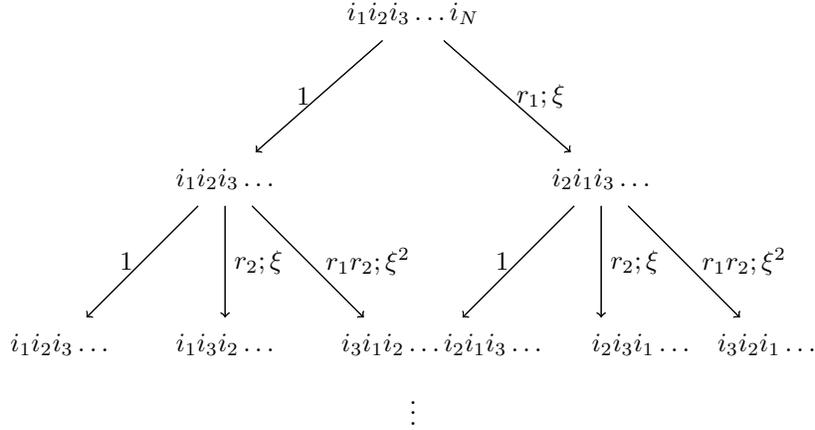
\begin{center}
\begin{tikzpicture}[grow=down,->,level/.style={sibling distance=150mm/#1}]
\begin{scope}[yshift=0cm]
  \node[inner sep=0, minimum size=2.em](e) {${{i_1}{i_2}{i_3} \ldots {i_N}}$} 
   child {node[minimum size=2.0em](e_1) {${i_1}{i_2}{i_3} \ldots $}
    child {node[minimum size=2.0em](e_2) {${ {i_1}{i_2}{i_3} \ldots}$}
      }      child {node[minimum size=2.em](q_2) {${{i_1}{i_3}{i_2}\ldots }$}
      }
      child {node[minimum size=2.em](q_3) {${ {i_3}{i_1}{i_2} \ldots }$}
      }
    }    child {node[minimum size=2.em](q_1) {${ {i_2}{ i_1}i_3\ldots} $}
    child {node[minimum size=2.em] (q_4) {$~~~~~~~~~~~~~{ {i_2}{i_1}{i_3}\ldots} $}
      }        child {node[minimum size=2.em] (q_5) {$~~~~~~~~~{ {i_2}{i_3}{i_1}\ldots} $}
      }
      child {node[minimum size=2.em](q_6) {${ {i_3}{i_2}{i_1}\ldots} $}
      }   
      }
   (e) edge[left] node{1} (e_1)
     (e) edge[right] node{$r_1;\xi$} (q_1)  
      (e_1) edge[left] node{1} (e_2)  
     (e_1) edge[right] node{$r_2;\xi$} (q_2)
      (e_1) edge[right] node{$~r_1r_2;\xi^{2}$} (q_3)  
        (q_1) edge[left] node{$1$} (q_4)  
        (q_1) edge[right] node{$r_2;\xi$} (q_5)   
    (q_1) edge[right] node{$~r_1r_2;\xi^{2}$} (q_6);
\end{scope}
\end{tikzpicture}
$$\vdots$$
\end{center}
The first level contains states produced from the action of $1+zr_1$ on $x_{i_1} \otimes x_{i_2} \otimes\ldots \otimes x_{i_N}.$
The $k^{th}$ horizontal level in the graphical representation of ${\mathrm Y}_k(z)\ x_{i_1}\otimes x_{i_2} \otimes \ldots\otimes  x_{i_k} \otimes \ldots \otimes x_{i_N}$ 
gives all possible $k!$ permutations of the first $k$ indices $ i_1 i_2\ldots i_k \ldots $ 
and the last level  contains all possible $N!$ permutations of $ i_1 i_2 \ldots i_N$

\begin{rem} \label{geny1}
Requiring,
\begin{equation}
(1+zr)\ x_i \otimes x_j \propto (1+z r)\ x_j \otimes x_i,\qquad x_i<x_j \in X, \nonumber
\end{equation}
leads to two values for $z:$ $z=q^2$ and $z=-1:$ 
\begin{eqnarray}
(1+q^2 r)\ x_j\otimes x_i = q(1+ q^2 r )\ x_i \otimes x_j, ~~~~~(1- r)\ x_j\otimes x_i =   -q^{-1}(1- r)\ x_i\otimes x_j. \label{eq1}
\end{eqnarray}
Specifically, ${\mathrm Y}(q^2)$ is called the $q$-symmetrizer and ${\mathrm Y}(-1)$ is called the $q$-antisymmetrizer and (see Proposition \ref{shuffle1}), 
\begin{equation}
    {\mathrm Y}_N(z_0)\ x_{i_1} \otimes x_{i_2}\ldots \otimes  x_{i_N} \propto {\mathrm Y}_N(z_0)\ x_{j_1} \otimes x_{j_2} \otimes \ldots  \otimes x_{j_N}, \qquad z_0 \in \big \{-1, q^2\big \}, \label{geny}
 \end{equation}    
where $x_{i_1}\leq x_{i_2 }\leq \ldots \leq x_{i_N}$ and $j_1 {j_2}\ldots {j_N} \in {\mathfrak S}_{\hat i}(N).$
The general case (\ref{geny}) can be shown by  (\ref{matrix0}), (\ref{eq1}). 
\end{rem}
The $q$-symmetrizer generates all $q$-symmetric states known also as the $q$-analogues of qudit Dicke states, which are $q$-deformed, high rank generalizations of the qubit Dicke states \cite{Dicke1} (see recent results on the construction of $q$-Dicke states in \cite{Nepo1, Nepo2} based on the action of the elements of $\mathfrak{U}_q(\mathfrak{gl}_n)$ on a reference state). The $q$-anti-symmetrizer produces fully $q$-antisymmetric states, starting from factorized states. If $q=1,$ ${\mathrm Y}_N(1),$ ${\mathrm Y}_{N}(-1)$ are the symmetrizer and anti-symmetrizer respectively; the symmetrizer produces all fully symmetric states, and the anti-symmetrizer yields the fully antisymmetric states.

\begin{lemma} \label{symmy}
Let ${\mathrm Y}_N(z),$ $z\in {\mathbb C}$ be defined in Proposition \ref{shuffle1} and 
$\pi: \mathfrak{U}_q(\mathfrak{gl}_n) \to \End({\mathbb C}^n)$ 
be the fundamental representation of $\mathfrak{U}_q(\mathfrak{gl}_n)$ (\ref{eval}), then 
\begin{equation}
\big [ {\mathrm Y}_N(z),\ \pi^{\otimes N}\Delta^{(N)}(y) \big ] =0,~~~~y \in \mathfrak{U}_q(\mathfrak{gl}_n).
\end{equation}
\end{lemma}
\begin{proof}
This statement follows from (\ref{symm1}), the definition of 
${\mathrm Y}_N(z)$ and $r_i := q^{-1}{\mathrm g}_{i},$ $i \in [N-1],$ where ${\mathrm g}_i$ is defined in (\ref{braidq1}), (\ref{braidq}). 
\end{proof}
\begin{rem} \label{remsym}
Let $X = \big \{x_1,x_2, \ldots, x_n\big \},$ $1\leq k_j\leq N,$ $j\in [m],$ such that $k_1+k_2 +\ldots + k_m =N.$ 
We present an alternative way to express any state, $\underbrace{x_{i_1} \ldots x_{i_1}}_{k_{1}} \otimes \underbrace{x_{i_2} \ldots x_{i_2}}_{k_{2}} \ldots \otimes \underbrace{x_{i_m} \ldots x_{i_m}}_{k_{m}}$ ($x_{i_1} \leq x_{i_2} \ldots \leq x_{i_m}$). Let $0\leq m_j\leq N,$ $j\in [n],$ and $m_1+m_2 +\ldots + m_n=N,$
then 
\begin{equation}
 \underbrace{x_{i_1} \ldots x_{i_1}}_{k_{1}}\otimes\underbrace{x_{i_2} \ldots x_{i_2}}_{k_{2}} \ldots\otimes \underbrace{x_{i_m} \ldots x_{i_m}}_{k_{m}} = \underbrace{x_{1} \ldots x_{1}}_{m_{1}} \otimes \underbrace{x_{2} \ldots x_{2}}_{m_{2}} \ldots \otimes \underbrace{x_{n} \ldots x_{n}}_{m_{n}},
\end{equation}
where $ m_j = k_l$ if $j= i_l$ and  $m_j =0$ if $j \neq i_l,$ $j,i_l \in [n],$ $l \in [m].$

We now define,
\begin{eqnarray}
&& b^{(0)}_{m_{1} \ldots m_{n}}: = \big ([[m_1]]_q![[m_2]]_q! \ldots [[m_n]]_q! \big)^{-1} {\mathrm Y}_N(q^2)\ \underbrace{x_{1} \ldots x_{1}}_{m_{1}} \otimes \underbrace{x_{2} \ldots x_{2}}_{m_{2}} \ldots \otimes \underbrace{x_{n} \ldots x_{n}}_{m_{n}}, \nonumber\\ 
&& {\hat b}^{(0)}_{m_{1} \ldots m_{n}} := \frac{b^{(0)}_{m_{1}\ldots m_{n}}}{\| b^{(0)}_{m_{1}\ldots m_{n}}\|}, ~~~~
\| b^{(0)}_{m_{1}\ldots m_{n}}\| = \Big(\frac{[[N]]_q!}{[[m_1]]_q! \ldots [[m_n]]_q!}\Big)^{\frac{1}{2}} \label{remsym1}
\end{eqnarray}
$x_{i}\in X,$ $i\in [n].$ 

The set of all normalized $q$-symmetric states ${\mathbb B}_s = \big \{{\hat b}^{(0)}_{m_{1}\ldots m_{n}}\big \},$ $0\leq m_j\leq N,$ $j\in [n]$ and $m_1+m_2 + \ldots +m_n=N,$ is an orthonormal basis of a vector space of dimension $d_s = \frac{n (n+1) \ldots (n+N-1)}{N!}$, this is also the dimension of the SSYT of $N$ columns and one row. For instance for $n=2$ (${\mathfrak{U}_q(\mathfrak{gl}_2)}$) $d_s =N+1,$ for $n=3$ (${\mathfrak{U}_q(\mathfrak{gl}_3)}$) $d_s= \frac{(N+1)(N+2)}{2}$ and so on. It is straightforward to show that
\begin{equation}
{\hat b}^{(0)T}_{m_{1} \ldots m_{n}} {\hat b}^{(0)}_{m'_{1} \ldots m'_{{n}}} =  \delta_{m_1,m'_1} \ldots \delta_{m_n, m'_n},
\end{equation}
where $^T$ denotes transposition.

The set of all $q$-antisymmetric states for $n\geq N$ is given by, 
\begin{equation}
\hat b^{(-)}_{i_1i_2 \ldots i_N} : = \frac{Y_N(-1) x_{i_1} \otimes x_{i_2} \otimes  \ldots\otimes x_{i_N}}{\sqrt{[[N]]_{q^{-1}}!}}
\end{equation}
$i_j \in [n],$ $j \in [N]$ and $x_{i_1}<x_{i_2} \ldots < x_{i_N}$. 
These states are the basis of the $d_a=\frac{n (n-1) (n-2) \ldots (n-N+1)}{N!}$ dimensional vector space represented by the SSYT of one column and $N$ rows (specifically for $n=N,$ $d_a =1$).
\end{rem}
The $q$-(anti)symmetric states are elements of the $q$-Fock space defined as follows (see relevant construction in \cite{MisMiw}). 
\begin{defn} \label{qFock}
Let $V_n$ be an $n$-dimensional vector space over some field $K$ (here $K$ is either ${\mathbb C}$ or ${\mathbb R}$), then the $q$-Fock space is defined as the direct sum of all $q$-(anti)symmetric tensors in $V_n^{\otimes N},$
$$ {\cal F}^{(\epsilon)}_{q}(V_n) = \bigoplus_{N\geq0} \mathfrak{s}^{(\epsilon)}_qV_n^{\otimes N},$$
where $\epsilon \in \{+, -\},$ $\mathfrak{s}^{(+)}_q$ is the $q$-symmetrizer and $\mathfrak{s}_q^{(-)}$ the $q$-antisymmetrizer.
\end{defn}
In the isotropic limit $q \to 1,$ the usual Fock space is recovered (see for instance \cite{Lechner}).

We present some explicit examples of $q$-symmetric (and anti-symmetric) states below (see also Example \ref{spectrum1} for $N=2$).
Recall, $X = \big \{x_1,x_2,\ldots, x_n\big \},$ $x_1<x_2\ldots < x_n.$ 
\begin{exa}\label{exam2}
The normalized $q$-symmetric states for $N=3$, for all $x_i<x_j<x_k \in X,$ are 
\begin{eqnarray}
&&\hat b^{(0)}_{{3}_i} = x_i \otimes x_i  \otimes x_i, \nonumber\\ 
&&\hat b^{(0)}_{{2}_i {1}_j} = \frac{1}{\sqrt{1+q^2 + q^4}} \big ( x_i\otimes x_i \otimes x_j + q\ x_i \otimes x_j \otimes x_i + q^{2} \ x_j \otimes x_i \otimes x_i\big ),\nonumber\\   
&&\hat b^{(0)}_{{1}_i  {2}_j}  = \frac{1}{\sqrt{1+q^2 + q^4}} \big ( x_i \otimes x_j\otimes x_j + q\ x_j\otimes  x_i \otimes x_j + q^{2} \ x_j \otimes x_j \otimes x_i\big ),\nonumber\\
&& {\hat b}^{(0)}_{{1}_i{1}_j  {1}_k} = \frac{1}{\sqrt{1+2q^{2}+2q^4+q^6}}\big( x_i \otimes x_j \otimes x_k +q( x_j\otimes x_i \otimes x_k+ x_i\otimes x_k\otimes x_j) \nonumber \\ && \qquad  +\ q^{2} (x_j \otimes x_k  \otimes x_i + x_k\otimes x_i\otimes x_j) + q^{3} x_k\otimes x_j \otimes x_i \big). \nonumber
\end{eqnarray}
We compare the simplified notation of the states above with the notation introduced earlier, $\hat b^{(0)}_{k_1 k_2 \ldots k_n},$ $0\leq k_j \leq N,$ $j \in [n]$ and $k_1 + k_2 + \ldots k_n =N.$ For instance for the state $b^{(0)}_{{3}_i}$ this notation means that $k_i = 3,$ and $k_j =0,$ for all $i\neq j \in [n],$ whereas for the state $\hat b^{(0)}_{{1}_i {2}_j}$ it means that $k_i =1,$ $k_j=2$ and $k_l =0$ for all $l\neq i,j \in [n],$ and so on.
We can simply write that last state as $\hat b^{(0)}_{ijk}.$

The normalized $q$-antisymmetric states for $x_i<x_j<x_k \in X,$ $N=3,$ $n>2$ are
\begin{eqnarray}
    &&\hat b^{(-)}_{ijk} = \frac{1}{\sqrt{1+2q^{-2}+2q^{-4}+q^{-6}}}\big( x_i\otimes x_j\otimes x_k - q^{-1} (x_j\otimes x_i\otimes x_k+ x_i\otimes x_k \otimes x_j) \nonumber\\ 
    && \qquad  +\ q^{-2} (x_j\otimes x_k \otimes x_i + x_k\otimes x_i\otimes x_j) - q^{-3} x_k\otimes x_j\otimes x_i \big). \nonumber
    \end{eqnarray}
\end{exa}

We present below explicitly the $q$-symmetric states for $\mathfrak{U}_q(\mathfrak{gl}_2)$ for any $N$, i.e. the qubit $q$-Dicke states \cite{Nepo2}.
\begin{exa}
For $n=2,$  and any $N\in {\mathbb Z}^+$ the normalized $q$-symmetric states are,
\begin{eqnarray}
&& \hat b^{(0)}_{N0}= x_1 \otimes x_1 \ldots \otimes x_1,~~~~\hat b _{0N} = x_2 \otimes x_2 \ldots  \otimes x_2 \nonumber\\
&& \hat b^{(0)}_{1(N-1)} = \frac{1}{\sqrt{[[N]]_q}} \big (x_1 \otimes x_2x_2 \ldots x_2 + q\ x_2\otimes x_1 \otimes x_2 \ldots  x_2 + \ldots +q^{N-1}\ x_2 x_2 \ldots x_2  \otimes x_1\big),  \nonumber\\ 
&& \hat b^{(0)}_{(N-1)1} = \frac{1}{\sqrt{[[N]]_q}}\big (x_1\ldots x_1 \otimes x_2 + q\ x_1 \ldots x_1  \otimes x_2\otimes x_1 + \ldots +q^{N-1}\ x_2 \otimes  x_1 \ldots x_1 \big ),
\nonumber\\
&& \hat{b}^{(0)}_{k_1k_2} = \sqrt{\frac{[[k_1]]_q![[k_2]]_q!}{{[[N]]_q!}}}\big (\underbrace{x_1\ldots x_1}_{k_1} \otimes  \underbrace{x_2 \ldots x_2}_{k_2} +  \ldots +q^{k_1k_2} \underbrace{x_2 \ldots x_2}_{k_2} \otimes \underbrace{x_1 \ldots x_1}_{k_1} \big ), \nonumber\\ && k_1+k_2 =N,\  k_{1,2}>1. \nonumber
\end{eqnarray} 
\end{exa}

\subsection{Orthogonal combinatorial bases and quantum algebra automata}
We established in the previous subsection that all $q$-symmetric states are obtained from the action of the $q$-symmetrizer. We also showed that for any given $n$ and $N$ such states form an orthonormal basis of a vector space of dimension $d_s = \frac{n (n+1) \ldots (n+N-1)}{N!} $, which corresponds to the SSYT of shape $\lambda =(N)$. We show in what follows that this is a combinatorial basis for ${\mathfrak U}_q(\mathfrak{gl}_n)$ that corresponds to its $d_s$ dimensional irreducible representation.
Moreover, we show how these {states} are also obtained from the action
of the elements of ${\mathfrak U}_q(\mathfrak{gl}_n)$ on a reference state {(highest weight)} (see also \cite{Lusztig, Kashi, crystal} and more recent works in connection to $q$-Dicke qudit states \cite{Nepo1, Nepo2}). The isotropic and crystal limits are also discussed.

Before we proceed with the main theorem below we introduce some handy notation. Recall the fundamental representation of $\mathfrak{U}_q(\mathfrak{gl}_n),$ $\pi:\mathfrak{U}_q(\mathfrak{gl}_n) \to \End({\mathbb C}^n),$ given explicitly in (\ref{eval}), we then define the following quantities:
\begin{eqnarray}
&& E_j := \pi^{\otimes N}\Delta(e_j), ~~~F_j:= \pi^{\otimes N}\Delta^{(N)}(f_j), ~~~q^{H_j} := \pi^{\otimes N}\Delta^{(N)}(q^{h_j}), ~~j \in [n-1], \nonumber \\ && q^{{\cal E}_j} := \pi^{\otimes N}\Delta^{(N)}(q^{\varepsilon_j}), ~~j \in [n] ~~~\mbox{and}~~~ q^{H_j} = q^{{\cal E}_j}q^{-{\cal E}_{j+1}}. \label{delta}
\end{eqnarray}

\begin{thm} \label{basicpro} Let $X = \big \{x_1, x_2, \ldots, x_n\big \},$ ${\mathrm Y}_N(q^2)$ be the $q$-symmetrizer as defined in Proposition \ref{shuffle1} and $r \in \End(({\mathbb C}^n)^{\otimes 2})$ is given in (\ref{rr}). Let also
$E_j,F_j$ $j \in [n-1]$ be defined in (\ref{delta}). Then,
\begin{enumerate}
\item \begin{equation}
E_{n-1}^{k'_{n}}\ldots E_2^{k'_3}E_1^{k'_2}\ \underbrace{x_1 x_1 \ldots x_1}_N \propto {\mathrm Y}_N(q^2)\ \underbrace{x_1 \ldots x_1}_{k_1} \otimes\underbrace{x_2\ldots x_2}_{k_2} \ldots \otimes \underbrace{x_{n}  \ldots x_n}_{k_{n}}, \label{central}
\end{equation}
such that $k_1+k_2 +\ldots +k_n =N$ and $k_j' = k_j +k_{j+1} + \ldots + k_n$ and $0\leq k_j \leq N,$ $j \in [n].$
\item \begin{equation}
F_{1}^{k'_{1}}\ldots F_{n-2}^{k'_{n-2}}F_{n-1}^{k'_{n-1}}\ \underbrace{x_n x_n \ldots x_n}_N\propto {\mathrm Y}_N(q^2) \underbrace{x_1 \ldots x_1}_{k_1} \otimes  \underbrace{x_2\ldots x_2}_{k_2} \ldots \otimes \underbrace{x_{n} \ldots x_n}_{k_{n}}, \label{central2}
\end{equation}
such that $k_1+k_2 +\ldots +k_n =N$ and $k_j' = k_j +k_{j-1} + \ldots + k_1$ and $0\leq k_j \leq N,$ $j\in [n]$.
\end{enumerate}
\end{thm}
\begin{proof} 
We first recall,
\begin{equation}
{\mathrm Y}_N(q^2) \ x_1x_1 \ldots x_1 \otimes x_2  \propto  x_1x_1 \ldots x_1 \otimes  x_2 +q \ x_1x_1 \ldots x_1 \otimes x_2 \otimes x_1 + \ldots + q^{N-1}\ x_2 \otimes x_1 \ldots x_1. \label{yy}
\end{equation}
By recalling the $N$-coproduct of $e_j, f_j$ the fundamental representation (\ref{eval}) and the notation in (\ref{delta})
we obtain (recall Remark \ref{rem1})
\begin{equation}
E_j = \sum_{k=1}^N q^{-\frac{s_j}{2}} \otimes \ldots \otimes q^{-\frac{s_j}{2}}\otimes \underbrace{e_{j+1,j}}_{k\ position} 
\otimes q^{\frac{s_j}{2}}\otimes \ldots \otimes q^{\frac{s_j}{2}}, ~~~s_j := e_{j,j} - e_{j+1, j+1}. \label{action2}
\end{equation}

We give the outline of the proof only for part (1), part (2) is shown in a analogous way. We first observe that, 
\begin{equation}
E_1\ x_1x_1\ldots x_1 = q^{-\frac{N-1}{2}} \big ( x_1\ldots x_1 \otimes x_2 + q\ x_1 \ldots x_1 \otimes  x_2 \otimes x_1 + \ldots + q^{N-1}\ x_2\otimes x_1 \ldots x_1\big ) \label{ee}
\end{equation}
From (\ref{yy})-(\ref{ee}) we conclude $~E_1\ x_1x_1 \ldots x_1 \propto {\mathrm Y}_{N}(q^2)\ x_1 \ldots x_1 \otimes x_2.$
Similarly, via Lemma \ref{symmy} 
\begin{eqnarray}
    && E_1^{2}\ x_1x_1\ldots x_1 \propto E_1 {\mathrm Y}_{N}(q^2)\ x_1\ldots x_1 \otimes x_2  =
    {\mathrm Y}_N(q^{2})E_1\ x_1\ldots x_1 \otimes x_2 \propto \nonumber\\ 
    && {\mathrm Y}_N(q^{2})\big(x_1\ldots x_1x_2 \otimes x_2 +q\ x_1\ldots x_1x_2x_1 \otimes x_2 + q^{N-2}\ x_2x_1\ldots x_1  \otimes x_2\big)\Rightarrow ~~~~~~~\mbox{(by Remark \ref{geny1})}\nonumber\\
&& E_1^{2}\ x_1x_1\ldots x_1 \propto {\mathrm Y}_{N}(q^2)\ x_1 \ldots x_1 \otimes x_2 \otimes  x_2. \nonumber 
\end{eqnarray}
Similarly, by iteration and use of Remark \ref{geny1}, Lemma \ref{symmy} and (\ref{action2}) we arrive at
\begin{equation}
E_1^{k_2}\ x_1 x_1 \ldots x_1 \propto Y_N(q^2) \underbrace{x_1 \ldots x_1}_{k_1} \otimes \underbrace{x_2 \ldots x_2}_{k_2}.
\end{equation}
We repeat the same process as above for $E_2,$ then for $E_3,$ and so on and arrive at (\ref{central}).
Expression (\ref{central2}) is also shown in a similar manner.
\end{proof}
From Theorem \ref{basicpro} we conclude that all states ${\hat b}^{(0)}_{k_{1}k_{2}\ldots k_{n}}$ defined in (\ref{remsym1}) are obtained from the action of elements of ${\mathfrak U}_q({\mathfrak{gl}_n}),$ see also Proposition \ref{basiclemma} below. 

\begin{pro} \label{basiclemma} Let $X = \big \{x_1, x_2, \ldots, x_n\big \},$
$E_j,F_j, q^{H_j}$ $j \in [n-1]$ and $ q^{{\cal E}_j},$ $j \in [n]$ be defined in (\ref{delta}) and  ${\hat b}^{(0)}_{k_{1}k_{2}\ldots k_{n}}$ be defined in Remark \ref{remsym}, equation (\ref{remsym1}),  $0\leq k_{i} \leq N,$ $i \in [n],$ and  $k_1 + k_2 + \ldots +k_n =N.$ Then,
\
\begin{eqnarray}
&& E_j {\hat b}^{(0)}_{k_1\ldots k_j k_{j+1} \ldots k_n} = c_{k_j, k_{j+1}} {\hat b}^{(0)}_{k_1 \ldots (k_{j}-1) (k_{j+1} +1)\ldots k_n}, \nonumber\\ 
&& F_j {\hat b}^{(0)}_{k_1\ldots (k_{j}-1) (k_{j+1}+1)\ldots k_n} = c_{k_j, k_{j+1}} {\hat b}^{(0)}_{k_1 \ldots k_{j} k_{j+1}\ldots k_n}, \nonumber\\ 
&& q^{{\cal E}_j}{\hat b}^{(0)}_{k_1\ldots k_n} =q^{k_j}{\hat b}^{(0)}_{k_1\ldots k_n}, ~~~q^{H_j}{\hat b}^{(0)}_{k_1\ldots k_n} = q^{k_j - k_{j+1}}{\hat b}^{(0)}_{k_1\ldots k_n}  \nonumber \end{eqnarray}
where $c_{k_j, k_{j+1}}= \sqrt{[k_{j+1}+1]_q[k_j]_q},$ $[k]_q = \frac{q^{k} - q^{-k}}{q-q^{-1}}.$

Moreover, 
\begin{eqnarray}
    && E_{j+1}E_{j}\hat b^{(0)}_{k_1 \ldots k_j k_{j+1} \ldots k_n} = c  E_{j}E_{j+1}\hat b^{(0)}_{k_1 \ldots k_j k_{j+1}\ldots k_n}\nonumber \\
    && F_{j+1}F_{j}\hat b^{(0)}_{k_1 \ldots (k_j-1) (k_{j+1} +1) \ldots k_n} = d  F_{j}F_{j+1}\hat b^{(0)}_{k_1 \ldots (k_j-1) (k_{j+1}+1)  \ldots k_n}, \label{expp}
    \end{eqnarray}
    where $c= \frac{[k_{j+1}+1]_q}{[k_{j+1}]_q},$ $d=\frac{[k_{j+1}+1]_q}{[k_{j+1}+2]_q}.$ 
\end{pro}
\begin{proof}
We focus, without loss of generality in the proof, on the action of $E_1$ on $\hat b^{(0)}_{k_1 k_2\ldots k_n},$ which leads to the final state $\hat b^{(0)}_{(k_1-1)(k_2+1) \ldots k_n}.$
We could have focused in a similar fashion to the action of $E_j$ on  $\hat b^{(0)}_{k_1 \ldots  k_j k_{j+1}\ldots k_n},$ which leads to the final state 
$\hat b^{(0)}_{k_1k_2 \ldots (k_j-1)(k_{j+1} +1) \ldots k_n}$
Indeed, via (\ref{copba}), (\ref{copbb}), (\ref{eval}) and (\ref{delta}):
\begin{eqnarray}
 E_1 \hat b^{(0)}_{k_1 k_2 \ldots k_n} &=&
\sqrt{\frac{[[k_1]]_q! [[k_2]]_q! \ldots [[k_n]]_q!}{[[N]]_q!}} E_1\big (\underbrace{x_1 \ldots x_1}_{k_1} \otimes \underbrace{x_2 \ldots x_2}_{k_2} \otimes\ldots\nonumber\\ &+ &
q\underbrace{x_1 \ldots x_1}_{k_1-1} \otimes x_2 \otimes x_1 \otimes \underbrace{x_2 \ldots x_2}_{k_2-1} \otimes\ldots + \ldots  +
q^{k_2}\underbrace{x_1 \ldots x_1}_{k_1-1} \otimes \underbrace{x_2 \ldots x_2}_{k_2}\otimes x_1 \otimes \ldots + \ldots \big) \nonumber \\
&=& \sqrt{\frac{[[k_1]]_q! [[k_2]]_q! \ldots [[k_n]]_q!}{[[N]]_q!}} q^{-\frac{k_1+k_2-1}{2}} [[k_2+1]]_q\big (\underbrace{x_1 \ldots x_1}_{k_1-1} \otimes \underbrace{x_2 \ldots x_2}_{k_2+1} \otimes\ldots\nonumber\\ &+ &
q\underbrace{x_1 \ldots x_1}_{k_1-2} \otimes x_2 \otimes x_1 \otimes \underbrace{x_2 \ldots x_2}_{k_2} \otimes\ldots + \ldots \big ) \nonumber\\
&=& \sqrt{[k_2+1]_q[k_1]_q}\ \hat b^{(0)}_{(k_1-1) (k_2+1) \ldots k_n}.
\end{eqnarray}
The first $k_2$ terms in the first line of the expression above are sufficient to provide the overall factor in front of the term $\underbrace{x_1 \ldots x_1}_{k_1-1} \otimes \underbrace{x_2 \ldots x_2}_{k_2+1} \otimes\ldots$ in the final state $\hat b^{(0)}_{(k_1-1)( k_2+1) \ldots k_n},$ and thus the overall factor $c_{k_1,k_2} = \sqrt{[k_2+1]_q[k_1]_q}.$ 
In the same way, we show that $F_1 \hat b^{(0)}_{(k_1-1) (k_2+1) \ldots k_m } = c_{k_1,k_2} \hat b^{(0)}_{k_1 k_2 \ldots k_n}.$

The action of any $E_j$ on $\hat b^{(0)}_{k_1 k_2\ldots k_n}$ can be worked out in exactly the same manner using the same arguments by just focusing on the segment $k_j k_{j+1}$ in $\hat b^{(0)}_{k_1\ldots k_j k_{j+1} \ldots k_n}$. Moreover, we have focused here on the term $\underbrace{x_1 \ldots x_1}_{k_1-1} \otimes \underbrace{x_2 \ldots x_2}_{k_2+1} \otimes\ldots$ in the final state $\hat b^{(0)}_{(k_1-1) (k_2+1) \ldots k_n}$ in order to extract the overall factor $c_{k_1,k_{2}}.$ Equivalently, we could have focused on any other term on the final state $\hat b^{(0)}_{(k_1-1) (k_2+1) \ldots k_n}$ using again similar arguments.

Also it is straightforward to see by (\ref{copba}), (\ref{copbb}), (\ref{eval}) and (\ref{delta}) 
that $ q^{{\cal E}_j}{\hat b}^{(0)}_{k_1\ldots k_n} =q^{k_j}{\hat b}^{(0)}_{k_1\ldots k_n}$ and consequently, $q^{H_j}{\hat b}^{(0)}_{k_1\ldots k_n} = q^{k_j - k_{j+1}}{\hat b}^{(0)}_{k_1\ldots k_n}.$

Finally, expressions (\ref{expp}) follow from the action of $E_j, F_j$ on the $q$-symmetric states.
\end{proof}
Theorem \ref{basicpro} and Proposition \ref{basiclemma} state that the set of $q$-symmetric states form an orthogonal basis for the $d_s =\frac{1}{N!}\overset{N+n-1}{\underset{k =n}{\prod}k}$ dimensional irreducible representation of $\mathfrak{U}_q(\mathfrak{gl}_n),$ represented by the SSYT of shape $\lambda =(N).$

\begin{defn} \label{auto11b} {\bf (The symmetric ${\mathfrak U}_q(\mathfrak{gl}_n)$ automaton).}   
\begin{enumerate}
\item Let the set of states be ${\mathbb B}_s = \big \{\hat b^{(0)}_{k_1 k_2\ldots k_n}\big \},$ that is the set of $q$-symmetric states defined in Remark \ref{remsym}. 

\item Let the alphabet be $\Sigma =\big \{e_j, f_j, q^{h_j}\big \},$ $j \in [n-1].$ The respective transition matrices are $E_j, F_j, q^{H_j}: {\mathbb B}_s \to {\mathbb B}_s \cup \hat 0$ and their action on the states are given in Proposition \ref{basiclemma}.
\end{enumerate}
This automaton is called the symmetric ${\mathfrak U}_q(\mathfrak{gl}_n).$ 
\end{defn}
This is a semi-combinatorial automaton, due to the action of transition matrices on the automaton states.
\begin{exa}
(The symmetric $\mathfrak{U}_q(\mathfrak{gl}_2)$ automaton.) In this case the set of states is ${\mathbb B}_s = \big \{\hat b^{(0)}_{k_1k_2}\big \},$  ($0\leq k_{1,2} \leq N$ and $k_1+k_2 =N$),\ $\Sigma = \big \{e,f, q^{h}\big \}$ and the respective transition matrices are given by their action,
\begin{eqnarray}
&& E {\hat b}^{(0)}_{k_1k_2} = c_{k_2} {\hat b}^{(0)}_{(k_1-1) (k_2+1)}, \nonumber\\ 
&& F {\hat b}^{(0)}_{(k_1-1)(k_2+1)} = c_{k_2} {\hat b}^{(0)}_{k_1 k_2}, \nonumber\\ && q^{H}{\hat b}^{(0)}_{k_1 k_2} = a_{k_2}{\hat b}^{(0)}_{k_1k_2},  \nonumber 
\end{eqnarray}
where $c_{k_2} = \sqrt{[k_2+1]_q[N - k_2]_q},$ $0\leq k_2 \leq N-1$ and  $a_{k_2}= q^{N -2k_2},$ $0\leq k_2\leq N.$  This is the $N+1$ dimensional irreducible representation of $\mathfrak{U}_q(\mathfrak{gl}_2).$ We graphically depict the automaton, with $\hat b^{(0)}_{N0}$ being the start state in Figure 6 (see also example \ref{exafirst}, and recall that zero transitions are omitted from the diagram):
\begin{center}
\begin{figure}[ht]
\tiny
\begin{tikzpicture}[shorten >=1pt,node distance=2.5cm, on grid,auto] 
\node[state, inner sep=0, minimum size=2.3em, initial] (q1) {${\hat b^{(0)}_{N0}}$};
\node[state,inner sep=0, minimum size=2.3em, right of=q1] (q2) {${\hat b^{(0)}_{(N-1) 1}}$};
\node[inner sep=0, minimum size=2.3em, right of=q2] (q3) {$\ldots$};
\node[state,inner sep=0, minimum size=2.3em, right of=q3] (q3b) {$\hat b^{(0)}_{0N}$};
\path[->] 
(q1) edge[bend left,  above] node{$e;c_0$} (q2)
(q2) edge[bend left, below] node{$f;c_0$} (q1)
(q2) edge[bend left, above] node{$e;c_1$} (q3)
(q3) edge[bend left, below] node{$f;c_{1}$} (q2)
(q3) edge[bend left, above] node{$e;c_{N-1}$} (q3b)
(q3b) edge[bend left, below] node{$f; c_{N-1}$} (q3)
(q3b) edge[loop above] node{$q^{h}; a_N$} (q3b)
(q2) edge[loop above] node{$q^{h}; a_1$} (q2)
(q1) edge[loop above] node{$q^{h}; a_0$} (q1);
\end{tikzpicture}
\caption{The symmetric $\mathfrak{U}_q(\mathfrak{gl}_2)$ automaton} \label{symm1b}
\end{figure}
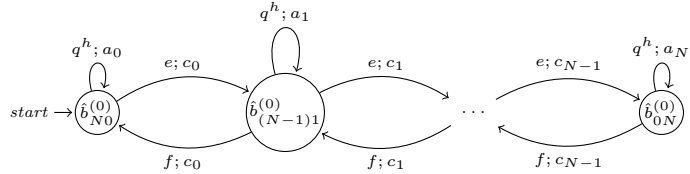
\end{center}
\end{exa}
{\bf Two important limits.}
\begin{enumerate}
\item Isotropic limit ($q=1$). In this case ${\mathfrak U}_q(\mathfrak{gl}_n) \to{\mathfrak U}(\mathfrak{gl}_n),$ and the coproducts become
\begin{equation}
\Delta(x) = 1 \otimes x +x \otimes 1, ~~~~x \in \mathfrak{gl}_n.
\end{equation}
Moreover, relations of Proposition \ref{basiclemma} still hold, but $[k]_q \to k$ and $[[k]]_q
\to k.$ Also, all states $b^{(0)}_{k_1 k_2 \ldots k_n}$ are fully symmetric states, i.e. are given by the sum of all possible permutations of\\ $\underbrace{x_1\ldots x_1}_{k_1}\underbrace{x_2\ldots x_2}_{k_2}\ldots \underbrace{x_n\ldots x_n}_{k_n}.$ The normalized symmetric states are $\hat b^{(0)}_{k_1 \ldots k_n}=\sqrt{\frac{k_1! \ldots k_n!}{N!}} b^{(0)}_{k_1 \ldots k_n}$ (see also Remark \ref{remsym}). 

\item Crystal limit ($q=0$): In this case every $q$-symmetric state reduces to an ordered factorized state, i.e. $\hat b^{(0)}_{k_1\ldots k_n}\to \underbrace{x_1\ldots x_1}_{k_1} \otimes x_2 \ldots x_2\ldots \otimes \underbrace{x_n\ldots x_n}_{k_n} =: \hat b^{(c)}_{k_1\ldots k_n}$.
\end{enumerate}
We also define $N_j := k_j + k_{j+1}$ and 
\begin{eqnarray}
    \underset{q\to 0}{\mbox{lim}}\ 
q^{\frac{N_j-1}{2}}E_{j} \hat b^{(0)}_{k_1\ldots k_j k_{j+1}\ldots k_n} = : \tilde  e_j \hat b^{(c)}_{k_1\ldots k_j k_{j+1}\ldots k_n}, \nonumber \\
 \underset{q\to 0}{\mbox{lim}}\ 
q^{\frac{N_j-1}{2}}F_{j} \hat b^{(0)}_{k_1\ldots k_j k_{j+1}\ldots k_n} = : \tilde f_j \hat b^{(c)}_{k_1\ldots k_j k_{j+1}\ldots k_n} \nonumber
\end{eqnarray}
Then,
\begin{equation}
\tilde  e_j\ \hat b^{(c)}_{k_1\ldots k_j k_{j+1}\ldots k_n} = \hat b^{(c)}_{k_1\ldots (k_j-1) (k_{j+1}+1)\ldots k_n}, \qquad 
\tilde f_j\ \hat b^{(c)}_{k_1\ldots (k_j-1) (k_{j+1}+1)\ldots k_n} = \hat b^{(c)}_{k_1\ldots k_j k_{j+1}\ldots k_n}. \label{comb2}
\end{equation}
The automaton with set of states ${\mathbb B}_c = \big \{\hat b^{(c)}_{k_1 \ldots k_n}\big \},$  $0\leq k_1, \ldots, k_n \leq N,$ $k_1 +k_2 + \ldots k_n =N,$ alphabet 
$\Sigma= \big \{{\mathfrak e}_j,  {\mathfrak f}_j\big \},$ $j\in [n-1],$ and transition matrices $\tilde e_j, \tilde f_j: {\mathbb B}_c \to {\mathbb B}_c \cup \hat 0$ (\ref{comb2}) is a combinatorial automaton, which we call the $A_{n-1}$ symmetric crystal automaton.
We provide a couple of examples after this remark. 

Recall also the completely antisymmetric states ($n\geq N$), $ \hat b^{(-)}_{i_1i_2 \ldots i_N} : = \frac{Y_N(-1) x_{i_1} \otimes x_{i_2} \otimes  \ldots\otimes x_{i_N}}{\sqrt{[[N]]_{q^{-1}}!}}$ ($i_j \in [n],$ $j \in [N]$ and $x_{i_1}<x_{i_2} \ldots < x_{i_N}$), then $\underset{q\to 0}{\mbox{lim}} = x_{i_N} \otimes x_{i_{N-1}} \otimes \ldots \otimes x_{i_1}.$ These states form the basis of a $d_{a}=\frac{n (n-1) (n-2) \ldots (n-N+1)}{N!}$ dimensional vector represented by the SSYT of shape $\lambda = (1,1, \ldots ,1).$ \hfill $\square$

\begin{exa}
We give below a couple of concrete examples of symmetric crystal automata. The zero transitions are omitted in the automaton diagrams.
\begin{enumerate}
    \item The $A_1$ symmetric crystal automaton for generic $N$ is depicted below, (we set 
    ${\mathfrak e}_1 =: {\mathfrak e},$ $\mathfrak{f}_1 =: {\mathfrak f}$ and $\tilde e_1 =:  \tilde e,$ $\tilde f_1 =:\tilde f$).
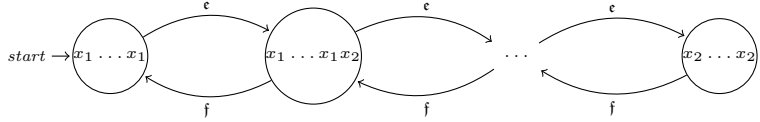
\begin{center}
\begin{figure}[ht]
\tiny
\begin{tikzpicture}[shorten >=1pt,node distance=2.7cm, on grid,auto] 
\node[state, inner sep=0, minimum size=2.3em, initial] (q1) {$x_1\ldots x_1$};
\node[state, inner sep=0, minimum size=2.3em, right of=q1] (q2) {$x_1\ldots x_1x_2$};
\node[inner sep=0, minimum size=2.3em, right of=q2] (q3) {$\ldots$};
\node[state,inner sep=0, minimum size=2.3em, right of=q3] (q5) {$x_2\ldots x_2$};
\path[->] 
(q1) edge[bend left,  above] node{$\mathfrak e$} (q2)
(q2) edge[bend left, below] node{$\mathfrak f$} (q1)
(q2) edge[bend left, above] node{$\mathfrak e$} (q3)
(q3) edge[bend left, below] node{$\mathfrak f$} (q2)
(q3) edge[bend left, above] node{$\mathfrak e$} (q5)
(q5) edge[bend left, below] node{$\mathfrak f$} (q3);
\end{tikzpicture}
\end{figure}
\end{center}

\item As a second example we consider the $A_2$ symmetric crystal automaton for $N=2$ as follows.The set of states consists of six elements ${\mathbb B} = \big \{x_1\otimes x_1, x_2 \otimes x_2, x_3 \otimes x_3, x_1 \otimes x_2, x_1 \otimes x_3, x_2 \otimes x_3\big \},$ and we consider the alphabet $\Sigma =\big \{{\mathfrak e}_j, {\mathfrak f}_j, \big \},$ $j\in [2]$ and transition matrices $\tilde e_j, \tilde f_j$ (\ref{comb2}). The automaton is graphically depicted below:
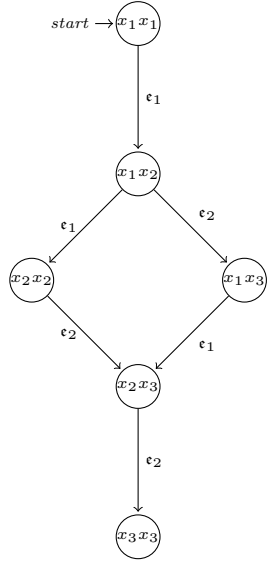
\begin{center}
\tiny
\begin{tikzpicture}[shorten >=1pt,node distance=2.cm,on grid,auto] 
   \node[state,inner sep=0, minimum size=2.0em, initial] (q_0)   {$x_1x_1$}; 
   \node[state, inner sep=0, minimum size=2.0em] (q_1) [below=of q_0] {$x_1x_2$};    
   \node[state, inner sep=0, minimum size=2.0em] (q_2) [below left=of q_1] {$x_2x_2$}; 
   \node[state, inner sep=0, minimum size=2.0em] (q_3) [below right=of q_1] {$x_1 x_3$}; 
    \node[state,inner sep=0, minimum size=2.0em](q_4) [below right=of q_2] {$x_2x_3$};
     \node[state,inner sep=0, minimum size=2.0em](q_5) [below=of q_4] {$x_3x_3$};    \path[->] 
 (q_0) edge  node {$\mathfrak e_1$} (q_1)
 (q_1) edge[left]  node {$\mathfrak e_1$} (q_2)
 (q_1) edge[] node {$\mathfrak e_2$} (q_3) 
(q_2) edge[left]  node {$\mathfrak e_2$} (q_4)
(q_3) edge[]  node {$\mathfrak e_1$} (q_4) 
(q_4) edge[]  node {$\mathfrak e_2$} (q_5) ;
\end{tikzpicture}
\end{center}
We omitted the ${\mathfrak f}_{j},$ $j \in [2]$ transitions in the automaton diagram for brevity. If we had included these transitions
 in the automaton diagram we would have drawn an opposite arrow 
next to each arrow that represents ${\mathfrak e}_{j}.$
\end{enumerate}
\end{exa}

\section{Eigenstates of integrable Hamiltonians and orthogonal bases} 
\noindent We have already shown in the previous section that the $d_s$ dimensional irreducible representation of $\mathfrak{U}_q(\mathfrak{gl}_n)$ consisting of all $q$-symmetric states corresponds to the SSYT of shape $\lambda=(N)$ (see Remark \ref{remsym} and Proposition \ref{basiclemma}). We also pointed out that the fully $q$-antisymmetric states ($n\geq N$) corresponds to a $d_a$ dimensional vector space represented by the SSYT of shape $\lambda = (1,1, \ldots,1)$ see Remark \ref{remsym}. The main question raised now is how can we systematically construct the bases for $\mathfrak{U}_q(\mathfrak{gl}_n)$ associated with any $\lambda$-SSYT? We argue in this section that this can be achieved by deriving the eigenstates of a finite spin chain Hamiltonian that is $\mathfrak{U}_q(\mathfrak{gl}_n)$ invariant. We prove this claim explicitly for the $\mathfrak{U}_q(\mathfrak{gl}_2)$ case. The open Hamiltonian  we are considering is nothing but the sum of all length-one words of the Hecke algebra ${\cal H}_N(q)$ given by,
\begin{equation}
{\cal H} = \sum_{j=1}^{N-1} r_j \in \EEnd(V_n^{\otimes N}), \label{ham0} 
\end{equation}
where $r_j$ are defined in (\ref{rr}), (\ref{rr2}), $r$ is the $\mathfrak{U}_q(\mathfrak{gl}_n)$ invariant braid operator. We recall that throughout this manuscript we consider $q := e^{\mu} \in {\mathbb R}$ and $V_n$ is either ${\mathbb C}^n$ or ${\mathbb R}^n.$

The open Hamiltonian (\ref{ham0}) describes a well known integrable system (see also \cite{Sklyanin, Pasquier, Kulish, MeNe, DoiNep, Doikous}) that has been extensively studied and the spectrum and eigenstates are known and are expressed in terms of Bethe roots \cite{MeNe, DoiNep}. We also {recall the main conjecture} \ref{conj1}, which states that sums of words of different lengths commute with each other and also all the sums of words of length $1\leq l \leq \frac{N(N-1)}{2}$ are central to ${\mathfrak U}_q(\mathfrak{gl}_n)$ (see Lemma \ref{symmy}). We focus on the spectrum and eigenstates of the Hamiltonian (\ref{ham0}) and hence obtain finite irreducible representations of  ${\mathfrak U}_q(\mathfrak{gl}_n)$ that correspond to $\lambda$-SSYT. We do not use in the present investigation the Bethe ansatz formulation, however a comparison with the results known from Bethe ansatz techniques \cite{MeNe, DoiNep} and a study of the combinatorial nature of the Bethe roots is an important task, which would provide significant information on the connection between Bethe ansatz equations, representation theory and combinatorics. To date the majority of relevant studies are restricted to periodic spin chain Hamiltonians, basically in the thermodynamic limit, where the periodic Hamiltonians recover the quantum group symmetry (see for instance \cite{KirRes1, KirRes2, Kuniba0, Kuniba1, Kuniba2}). However, it is more reasonable to study the eigenvalue problem of the special open spin chain Hamiltonian (\ref{ham0}), due to its exact ${\mathfrak U}_q(\mathfrak{gl}_n)$ symmetry for any size $N.$

{The primary results of this section are established and summarized in Propositions \ref{sypro} and \ref{systema}, and Theorem \ref{thetheorem}. These formulations address the spectral decomposition of the Hamiltonian (\ref{ham0}) and outline the systematic construction of combinatorial bases for irreducible representations of \(\mathfrak{U}_q(\mathfrak{gl}_2)\) utilizing degenerate eigenstates of (\ref{ham0}). This construction is directly enabled by the \(\mathfrak{U}_q(\mathfrak{gl}_2)\)-invariance of the Hamiltonian. While the core conclusions are summarized in these main statements, their proofs, especially the ones of Proposition \ref{systema} and Theorem \ref{thetheorem} —though highly technical and algebraic—are provided in full for completeness, and are available for the interested reader. Crucially, these proofs establish a novel explicit connection between pure representation theory of quantum algebras and the spectral decomposition of real symmetric matrices, such as the Hamiltonian (\ref{ham0}). Concrete examples are also presented in the end of this section.}

Before we proceed with our main statements regarding the eigenvalue problem of the Hamiltonian (\ref{ham0}) we note that the $r$-matrix (\ref{rr}) is a real symmetric matrix,  i.e. $r^T = r$ ($^T$ denotes total transposition; $r$ is real given that $q$ is real). It then follows that the Hamiltonian is also a real symmetric matrix. The eigenvalues of any real symmetric matrix are all real ones and the eigenstates are real vectors.
\begin{pro} \label{sypro}
Let ${\hat b}^{(0)}_{k_{1}k_{2}\ldots k_{n}}$ be the normalized $q$-symmetric states defined in Remark \ref{remsym}, equation (\ref{remsym1}), $0\leq k_{j}\leq N,$ $j \in [n]$ and $k_1 + k_2 + \ldots +k_n =N,$ and 
${\cal H} = \underset{1\leq j\leq N-1}{\sum} r_j,$
where $r\in \End({\mathbb R}^{n}\otimes {\mathbb R}^n)$ 
is the ${\mathfrak U}_q(\mathfrak{gl}_n)$ invariant solution of the braid equation (\ref{rr}).
Then the states ${\hat b}^{(0)}_{k_{1}k_{2}\ldots k_{n}}$ are eigenstates of ${\cal H}$ with eigenvalue $\Lambda_0 = N-1.$
\end{pro}
\begin{proof}
First consider the reference state $\hat b^{(0)}_{N0 \ldots 0} = x_1 \otimes x_1 \ldots \otimes  x_1,$ then using the fact that $r a\otimes a = a\otimes a,$ for all $a\in X,$ we obtain,
${\cal H} \hat b^{(0)}_{N0\ldots 0} = (N-1)\hat b^{(0)}_{N0 \ldots 0}.$ 
Recall that $\big [r_i, \Delta^{(N)}(x)\big ] = 0,$  for all $x\in  \mathfrak{U}(\mathfrak{gl}_n)$ and $i \in [N-1],$ then by Theorem \ref{basicpro} (recall also (\ref{delta}) and Proposition \ref{basiclemma}, in particular $E_{j+1} E_j \hat b^{(0)}_{k_1 \ldots k_n} \propto E_jE_{j+i}\hat b^{(0)}_{k_1 \ldots k_n}$) it follows
\begin{eqnarray}
&&(E_{n-1}^{k'_{n}}\ldots E_2^{k'_3}E_1^{k'_2})\ {\cal H}\ \underbrace{x_1 x_1 \ldots x_1}_N = (N-1 )(E_{n-1}^{k'_{n}}\ldots E_2^{k'_3}E_1^{k'_2})\ \underbrace{x_1 x_1 \ldots x_1}_N \nonumber\\
&& {\cal H}\ (E_{n-1}^{k'_{n}}\ldots E_2^{k'_3}E_1^{k'_2})\ \underbrace{x_1 x_1 \ldots x_1}_N =(N-1 )(E_{n-1}^{k'_{n}}\ldots E_2^{k'_3}E_1^{k'_2})\ \underbrace{x_1 x_1 \ldots x_1}_N \nonumber\\
&& {\cal H}\ \hat b^{(0)}_{k_1k_2\ldots k_n} = (N-1) \hat b^{(0)}_{k_1k_2 \ldots k_n}, \label{central2b}
\end{eqnarray}
such that $k_1+k_2 +\ldots +k_n =N,$ $k_j' = k_j +k_{j+1} + \ldots + k_n$ and $0\leq k_j \leq N,$ $j \in [n],$ (recall also Remark \ref{remsym} and Proposition \ref{basiclemma}).
\end{proof}

Before we proceed with our analysis we introduce some useful notation.

\noindent {\bf Notation.}
Henceforth, we adopt the following notation: consider a Young diagram of shape $\lambda = (l_1, l_2, \ldots, l_p),$ (recall $N= l_1 +l_2 + \ldots +l_p$ and $1\leq l_p\ \leq \ldots l_2\leq l_1 \leq N$), then
\begin{enumerate}
   \item $m_{l_2, \ldots, l_p}$ is the dimension of the $\lambda$-SYT ($m_{\lambda}$).
   \item $d_{l_2,\ldots,l_p,n}$ is the dimension of the $\lambda$-SSYT ($d_{\lambda,n}$). 
   \item If $\lambda =(N),$ the dimension of the SYT is denoted $m_0$ and the dimension of the SSYT is denoted $d_{0,n.}$
   \hfill $\square$ 
   \end{enumerate}
We have already seen in Proposition \ref{sypro} that all $q$-symmetric states that correspond to the SSYT of shape $\lambda =(N)$ are eigenstates of the Hamiltonian  (\ref{ham0}). In general, for the $\mathfrak{U}_q(\mathfrak{gl}_n)$ invariant Hamiltonian (\ref{ham0}) we claim that the decomposition of the space $V_n^{\otimes N},$ on which the Hamiltonian acts, in terms of eigenspaces is given as follows (we will prove this explicitly for the algebra $\mathfrak{U}_q(\mathfrak{gl}_2)$):
\begin{equation}
    V_n^{\otimes N} = \bigoplus_{\lambda \vdash N}  m_{\lambda} V_n^{(\Lambda_{\lambda})},
\end{equation}
where $\Lambda_{\lambda}$ are the Hamiltonian's (\ref{ham0}) eigenvalues associated with a $\lambda$-shaped Young diagram and $V_{n}^{(\Lambda_{\lambda})}$ 
are the corresponding eigenspaces of dimension $\mbox{dim}V_n^{(\Lambda_{\lambda})} = d_{\lambda,n}$. Also, $m_\lambda$ is the dimension of the $\lambda$-SYT and $d_{\lambda,n}$ is the dimension of the $\lambda$-SSYT. Or equivalently graphically,
\begin{center}
\begin{figure}[ht]
\footnotesize
\begin{ytableau}
\ 
\end{ytableau}$^{\otimes N} $   = \blue{$\underbrace{\begin{ytableau}
\ & & \ldots &
\end{ytableau} }_N$
$  \ \oplus\  \underset{N\geq l_1\geq l_2}{\sum} m_{l_2} $ 
$\underbrace{\begin{ytableau}
\ &  &\ldots &  & \\
\  & \dots &
\end{ytableau}}_{N-l_2=l_1}$} $\ \oplus$\\  
$\underset{N\geq l_1 \geq l_2\geq l_3\geq 1}{\sum} m_{l_2, l_3}$ $\underbrace{\begin{ytableau}
\ & & & \ldots &\\
\ & & \ldots & \\ 
\ &\ldots &
\end{ytableau}}_{N-(l_2+l_3) =l_1}$ $ \oplus \ldots \oplus\  \underset{N\geq l_1\geq  
\ldots l_{N-1} \geq 1}{\sum} m_{l_2,\ldots, l_{N-1}}$ $\underbrace{\begin{ytableau}
\ & & & \ldots &\\
\  & & \ldots & \\ 
\ \vdots &\ldots &\\
\ &
\end{ytableau}}_{N-(l_2+l_3 + \ldots l_{N-1}) = l_1}$ $\oplus$ 
$\begin{rcases*} \begin{ytableau}
\\\
\ \\ 
\ \vdots \\
\
\end{ytableau}\end{rcases*} N$.\\
\caption{Tensor representations: Young tableaux} \label{yyb}
\end{figure}
\end{center}
The constants $m_{l_2\ldots, l_P}$ as already mentioned are equal to the dimension of the corresponding SYT given by the hook length rule, i.e. $m_{\lambda}= \frac{N!}{\prod_{i,j} h_{i,j}}$. For instance, $m_{0}=1,$ $m_{1} =N-1,$ $m_{1, \ldots, 1} =1,$ etc. 
In the ${\mathfrak U}_q(\mathfrak{gl}_2)$ case in particular only the first two types of tableaux are considered (indicated in blue color in the first line in Figure 7). In Proposition \ref{sypro} we identified the eigenstates that correspond to the SSYT of shape $\lambda = (N),$ which form the basis of the $d_{0,n}$ dimensional vector space as was shown in the previous subsection. In what follows we will identify eigenstates of the Hamiltonian (\ref{ham0}) associated with any partition $\lambda = (N-k, k).$  In this case the dimension of the corresponding SYT is denoted $m_k,$ whereas the dimension of the corresponding SSYT is denoted $d_{k,n}.$   
\begin{rem} \label{important}
 Consider an $m\times m$ real symmetric matrix and assume there exist $l$ identical eigenvalues. Due to the fact that every symmetric matrix is diagonalizable, i.e. the algebraic multiplicities coincide with the geometric ones, there exist $l$ independent eigenstates (which can be made orthogonal using the Gram-Schmidt process) associated with each one of the repeated eigenvalues. Also, for any symmetric matrix any two distinct eigenvalues have orthogonal eigenstates. 
 \end{rem}
 It is practical for what follows to introduce the state ${\mathfrak b}_{(N-k)k0\ldots0}$, in line with the notation we have been using so far, 
 \begin{equation}
    {\mathfrak b}_{(N-k) k0\ldots 0} :=  \sum_{1\leq i_1<i_2< \ldots < i_k \leq N} a_{i_1i_2\ldots i_k}\ x_1 \ldots x_1 \otimes \underbrace{x_2}_{i_1} \otimes x_1 \ldots x_1\otimes  \underbrace{x_2}_{i_2}\otimes  x_1   \ldots x_1 \otimes \underbrace{x_{2}}_{i_k} \otimes x_1 \ldots x_1, \label{thestate}
    \end{equation} 
     $a_{i_1i_2\ldots i_k} \in {\mathbb R},$ although henceforth, we omit the zeroes in (\ref{thestate}) for brevity and simply write ${\mathfrak b}_{(N-k)k}.$  We will first study the eigenvalues and eigenstates of the Hamiltonian (\ref{ham0}) for a state of the type ${\mathfrak b}_{(N-1)1},$ {in this case the constants $a_{i_1 i_2 \ldots i_k}$ reduce to $a_j,$ $j \in [N]$.} 
\begin{pro} \label{systema}
Let $ {\cal H} = \underset{1\leq j \leq N-1}{\sum} r_j,$ where $r\in \End({\mathbb R}^{n}\otimes {\mathbb R}^n)$ is the $\mathfrak{U}_q(\mathfrak{gl}_n)$ invariant solution of the braid equation (\ref{rr}). Let also 
 \begin{equation}
 {\cal H}\ {\mathfrak b}_{(N-1) 1} = \Lambda \  {\mathfrak b}_{(N-1) 1}, \label{eigen11}
 \end{equation}
where $\Lambda\in {\mathbb R},$
and  ${\mathfrak b}_{(N-1)1}$  is given in (\ref{thestate}). Then: 
\begin{enumerate}
    \item 
The eigenvalue problem (\ref{eigen11}) yields $N$ distinct eigenvalues $\Lambda.$\\ 
{Also, the elements $\Lambda = \Lambda_0,$ $a_j =  q^{-(j-1)},$ $j \in [N],$ $\Lambda_0 = N-1$} satisfy equation (\ref{eigen11}). 

\item Let $\Lambda_1 \neq \Lambda_0,$ $a^{(1)}_j,$ $j\in [N]$ satisfy (\ref{eigen11}), 
 i.e. $b^{(1)}_{(N-1)1}$ is the eigenstate with eigenvalue denoted $\Lambda_1.$ 
Then: \\
(i) $\underset{1\leq j \leq N}{\sum} a^{(1)}_j q^{-j} =0.$ \\
(ii) $b^{(1)}_{(N-1) 1 } \perp V_n^{(\Lambda_{0})},$ where $V_n^{(\Lambda_{0})}$ is the $d_{0,n}$ dimensional vector space with basis ${\mathbb B} =\big \{\hat b^{(0)}_{k_1k_2 \ldots k_n}\big \},$ $0\leq k_j \leq N$ and $\underset{1 \leq j\leq n}{\sum} k_j =N,$ i.e. the set of all $q$-symmetric states (Propositions \ref{basiclemma} and \ref{sypro}).

\item Recall $E_j, F_j, q^{H_j},$ defined in (\ref{delta}) and let, $b^{(1)}_{(N-m-1) (m+1)} := E_1^{m} b^{(1)}_{(N-1) 1}.$ Then, 
\begin{eqnarray}
&& F_1 b^{(1)}_{(N-1) 1} =E_1 b^{(1)}_{1 (N-1) } =0,  \quad F_1 b^{(1)}_{(N-m-1)(m+1)}=\kappa_m^{(1)}b^{(1)}_{(N-m)m},\quad 
q^{{H_1}}  = a_m^{(1)} b^{(1)}_{(N-m)m}\nonumber
\end{eqnarray}
where $\kappa^{(1)}_{m} =[N-m-1]_q[m]_q,$ $m \in [N-2]$ and $a_m^{(1)} =q^{N-2m},$ $m\in [N-1].$ 

\item The set of eigenstates $b^{(1)}_{(N-m)m},$ $m \in [N-1]$
with an eigenvalue $\Lambda_1 \neq N-1$ is an orthogonal basis of an $N-1$ 
dimensional vector space and a basis for the $N-1$ irreducible representation 
of ${\mathfrak U}_{q}(\mathfrak{gl}_2).$
\end{enumerate}

 \end{pro}
\begin{proof}

$ $

\begin{enumerate}
    \item 
We first show that the following linear homogeneous system of $N$ equations holds:
\begin{eqnarray}
&& (N-1 -q^{-2})a_1 + q^{-1} a_2 = \Lambda a_1 \nonumber\\
&& q^{-1}a_{j-1} + (N-2 - q^{-2})a_j + q^{-1} a_{j+1} = \Lambda a_j\nonumber, ~~~j \in \big \{2, 3, \ldots, N-1\big \}\\
&& q^{-1}a_{N-1} + (N-2)a_N = \Lambda a_N. \label{system}
\end{eqnarray}    
    The proof is straightforward, it follows {from (\ref{eigen11})} and is based on the action of $r_j,$ for all $j \in [N-1],$ on the tensor product (\ref{braidg}) and the definition of the state ${\mathfrak b}_{(N-1) 1}.$ We also observe that the system of equations leads to $N$-distinct eigenvalues as it describes the eigenvalue problem of a real, symmetric tridiagonal $N\times N$ matrix. 
    Recall that a real, symmetric tridiagonal matrix has distinct eigenvalues if all its off-diagonal elements are non-zero.
    Indeed, say that two solutions of the system coincide i.e. $\Lambda =\Lambda'$  (choose for brevity $a_1=1$) then it follows from the system (\ref{systema}) that the coefficients $a_j,$ $2\leq j \leq N$ also coincide, i.e. the corresponding eigenvectors coincide. However, ${\cal H}$ is a symmetric matrix which means it is diagonalizable, hence all eigenvalues extracted from the system (\ref{system}) should be distinct (see also Remark \ref{important}).
    
    By setting $\Lambda = N-1$ (and choosing for simplicity $a_1=1$) it follows from the first equation of the system that $a_2 = q^{-1}$ from the second equation we obtain that $a_3 = q^{-2},$ and continuing in this manner we find that the $j^{th}$ equation leads to $a_j= q^{-(j-1)},$ $j \in [N].$ 

   \item  We recall from part (1) that there are $N-1$ distinct eigenvalues denoted $\Lambda_1 \neq \Lambda_0$ and the corresponding constants $a_j$ of equations (\ref{system}) are denoted $a_j^{(1)},$ $j\in [N].$ \\
   (i) We multiply the $j^{th}$ equation of the system (\ref{system}) by $q^{-j}$ and then by adding all the equations
we conclude (recall we have set $a^{(1)}_1 =1$)
\begin{eqnarray}
      &&  \big (\Lambda_1 - (N-1)\big )\sum_{j=1}^{N}   a^{(1)}_jq^{-j} =  0 ~\Rightarrow ~\sum_{j=1}^{N} a^{(1)}_jq^{-j} =0, \nonumber
   \end{eqnarray}
where $\Lambda_1 \neq N-1.$
\\
(ii) It suffices to show that $b^{(1)}_{(N-1) 1} \perp \hat b^{(0)}_{(N-1) 1}$ (recall Remark \ref{remsym}, again we omit for brevity the zeros in $\hat b^{(0)}_{(N-1)10\ldots 0}$ and write $\hat b^{(0)}_{(N-1)1}$ ), given that $\hat b^{(1)}_{(N-1) 1}$ is obviously orthogonal to all other $q$-symmetric states $\hat b^{(0)}_{k_1 k_2 \ldots k_n}$ by construction.  Indeed, take the inner product and use part (i), then
   $\hat b^{(0)T}_{(N-1)1} \cdot b^{(1)}_{(N-1)1} \propto \underset{j\in [N]}{\sum}a^{(1)}_jq^{-j} =0.$ 
   Although this is also expected due to the fact that ${\cal H}$ is a symmetric matrix, hence eigenstates with distinct eigenvalues are orthogonal.

   \item Recall the definitions of $E_j, F_j, q^{H_j}$ (\ref{delta}) and the eigenstate $b^{(1)}_{(N-1)1}$ (\ref{thestate}) with an eigenvalue $\Lambda_1,$ then $F_1 b^{(1)}_{(N-1)1} \propto x_1 \otimes x_1 \otimes \ldots x_1,$ but the latter is an eigenstate of the Hamiltonian with eigenvalue $\Lambda_0,$ Proposition \ref{sypro}. On the other hand,
   \begin{eqnarray}
   {\cal H} F_1 b^{(1)}_{(N-1)1} = F_1 {\cal H} b^{(1)}_{(N-1)1} = \Lambda_1 F_1 b^{(1)}_{(N-1)1}, \nonumber\end{eqnarray}
  which leads to $(\Lambda_0 -\Lambda_1) F_1 b^{(1)}_{(N-1)1} =0.$ Recall, $\Lambda_0 \neq \Lambda_1,$ hence $F_1 b^{(1)}_{(N-1)1} =0.$ 
   
Similarly, recall $b^{(1)}_{(N-m)m} = E_1^{m-1} b^{(1)}_{(N-1)1}.$ Then $E_1^{N-1}b^{(1)}_{(N-1)1}= E_1b^{(1)}_{1(N-1)} \propto x_2 \otimes x_2 \otimes \ldots x_2,$ which is an eigenstate with eigenvalue $\Lambda_0,$ Proposition \ref{sypro}. However, due to the exact symmetry of the open chain, $E_1^{N-1} b^{(1)}_{(N-1)1}$ is also an eigenstate with eigenvalue $\Lambda_1,$ which also leads to $E_1^{N-1}b^{(1)}_{(N-1)1} =0.$ We also show, given the structure of a state $b^{(1)}_{(N-m)m}$ and the definitions  in (\ref{delta}), that $q^{{\cal E}_1}b^{(1)}_{(N-m)m} = q^{N-m}b^{(1)}_{(N-m)m},$ $q^{{\cal E}_2}b^{(1)}_{(N-m)m} = q^{m}b^{(1)}_{(N-m)m}$ and hence
$q^{H_1}b^{(1)}_{(N-m)m} = q^{N-2m}b^{(1)}_{(N-m)m},$ $~1\leq m \leq  N-1.$ 
   
Moreover, we prove by induction that, $ F_1 b^{(1)}_{(N-m-1)(m+1)} = \kappa^{(1)}_{m} b^{(1)}_{(N-m)m},$ $ 1\leq m \leq N-2, $
where $\kappa^{(1)}_m= [N-m-1]_q[m]_q:$ 
\begin{itemize}
\item We first show this for $m=1:$  
recall $F_1 b^{(1)}_{(N-1)1} =0$ and $F_1 E_1 - E_1F_1 = \frac{q^{H_1} - q^{-H_1}}{q -q^{-1}},$ then \\
$\ F_1b^{(1)}_{(N-2)2} =  F_1 E_1 b^{(1)}_{(N-1)1} = \kappa_1^{(1)}b^{(1)}_{(N-1)1},$ where $~\kappa_1^{(1)} := [N-2]_q$

  \item We next assume for some $m -1 \geq 2$ that $F_1 b^{(1)}_{(N-m)m} = \kappa^{(1)}_{m-1} b^{(k)}_{(N-m+1)(m-1)},$
where $\kappa^{(1)}_{m-1} := [N-m]_q[m-1]_q.$

\item We finally show, using the ${\mathfrak U}_q(\mathfrak{gl}_n)$ relations that the equation above holds for $m$:
\begin{eqnarray}
 &&  F_1 b^{(1)}_{(N-m-1)(m+1)} = F_1 E_1 b^{(1)}_{(N-m)m} = \big ( \kappa^{(1)}_{m-1} + [N-2m]_q \big )b^{(1)}_{(N-m)m}, \nonumber
  \end{eqnarray}
which leads to $ F_1 b^{(1)}_{(N-m-1)(m+1)} = \kappa^{(1)}_{m}  b^{(1)}_{(N-m)m},$  where $\kappa^{(1)}_{m} := [N-m-1]_q[m]_q.$
 \end{itemize}
Notice that even if we hadn't shown that $E_1 b^{(1)}_{1(N-1)} =0$ using the symmetry arguments for the eigenvalues problem for ${\cal H},$ we would have ended up to this conclusion anyway, due to the fact that $N-1$ is an integer, hence the sequence of eigenstates $b^{(1)}_{(N-m)m}$ should terminate at $m=N-1$ (due to $\kappa^{(1)}_{N-1}=0$ and the $\mathfrak{U}_q(\mathfrak{gl}_2)$ relations) leading to $E_1b^{(1)}_{1(N-1)} =0.$ This is a standard argument for finite irreducible representations of $\mathfrak{U}_q(\mathfrak{gl}_2).$
 
\item The set ${\mathbb B}_1 = \big  \{b^{(1)}_{(N-m)m} \big  \},$ $m \in [N-1]$ is then an orthogonal basis of a $d_{1,2} = N-1$ dimensional vector space denoted $V_ 2^{(\Lambda_1)}.$ There are $N-1$ distinct eigenvalues $\Lambda_1\neq N-1$ and consequently $N-1,$ $d_{1,2}$ dimensional vector spaces orthogonal to each-other. The orthogonality between the vector spaces is guaranteed by the fact that the Hamiltonian is a real symmetric matrix. 

We now gather what we have shown in part (3): $F_1 b^{(1)}_{(N-1)1} = E b^{(1)}_{1 (N-1)} =0,$
\begin{eqnarray}
&& E_1 b^{(1)}_{(N-m)m} =   b^{(1)}_{(N-m-1) (m+1)}, \nonumber\\ 
&& F_1 b^{(1)}_{(N-m-1)(m+1)} = \kappa^{(1)}_{m} b^{(1)}_{(N-m) m}, \nonumber \\ 
&& q^{H_1}b^{(1)}_{(N-m) m} = a^{(1)}_m  b^{(1)}_{(N-m)m},   \nonumber 
\end{eqnarray}
 where $\kappa_m^{(1)}= [N-m-1]_q[m]_q,$ $a_m^{(1)} = q^{N-2m},$ $1\leq m \leq N-1$ The relations above indeed provide the $N-1$ dimensional irreducible representation of ${\mathfrak U}_q(\mathfrak{gl}_2),$ where the set ${\mathbb B}_{1} =\big  \{b^{(1)}_{(N-m)m}\big  \},$ $m\in[N-1]$ is a basis of this representation (compare with expressions (\ref{a1b}), (\ref{a2b}) in Example \ref{exafirst}).
    \qedhere
    \end{enumerate}
\end{proof}

\begin{rem}  \label{norm1} For the normalized eigenstates $\hat b^{(1)}_{(N-m)m}: =  \frac{b^{(1)}_{(N-m)m}}{||b^{(1)}_{(N-m)m}||},$ $m \in [N-1],$ we deduce
       \begin{eqnarray}
E_1 {\hat b}^{(1)}_{(N-m)m} = c^{(1)}_{m} {\hat b}^{(1)}_{(N-m-1) (m+1)}, \quad F_1 {\hat b}^{(1)}_{(N-m-1)(m+1)} = c^{(1)}_{m} {\hat b}^{(1)}_{(N-m) m}, \quad q^{H_1}{\hat b}^{(1)}_{(N-m) m} = a^{(1)}_m {\hat b}^{(1)}_{(N-m)m},   \nonumber 
\end{eqnarray}
 where the constants $c_m^{(1)} =\sqrt{[N-m-1]_q[m]_q},$ $a_m^{(1)} = q^{N-2m}$ are identified by recalling that $E_1, F_1, q^{H_1}$ are given in (\ref{delta}). The relations above as well as the last three relations proven in Proposition \ref{systema} give two distinct $N-1$ dimensional irreducible representations $\mathfrak{U}_{q}(\mathfrak{gl}_2)$ related however to each other by an algebra homomorphism (see Example \ref{exafirst} and relations (\ref{a1}), (\ref{a2}) and (\ref{a1b}), (\ref{a2b})). The constants $c^{(1)}_m$ are also in accordance with the fact that $||\hat b^{(1)}_{(N-m)m}|| =1.$
 \end{rem}
 
\begin{rem} \label{remark22} ({\bf Notation $\&$ preliminaries}.)
Before we generalize the results of Proposition \ref{systema} we introduce a convenient notation. Let ${\cal H}$ be the $\mathfrak{U}_q(\mathfrak{gl}_n) $ invariant Hamiltonian (\ref{ham0}) and 
 \begin{equation}
 {\cal H}\ {\mathfrak b}_{(N-k) k} = \Lambda \  {\mathfrak b}_{(N-k) k}, \label{eigen0}
 \end{equation}
 where $\Lambda\in {\mathbb R},$ $k\in {\mathbb Z}^+,$ such that $0\leq k \leq \frac{N}{2}$ and ${\mathfrak b}_{(N-k)k}$ is given in (\ref{thestate}).

    If there exists a state $b^{(k)}_{(N-k)k},$ such that $F_1b^{(k)}_{(N-k)k} =0$ (highest weight state), recall $F_1$ is given in (\ref{delta}), and ${\cal H} b^{(k)}_{(N-k)k} = \Lambda_k b^{(k)}_{(N-k)k},$ then we say that the eigenvalues $\Lambda_k$ and eigenstates $b^{(k)}_{(N-k) k}$ belong to the $k$-eigen-sector or simply $k$-sector. Notice, that by the definitions of $F_j$ (\ref{delta}) and 
    ${\mathfrak b}_{(N-k)k}$ (\ref{thestate}), $F_j {\mathfrak b}_{(N-k)k} =0,$ $j>1.$ Moreover, define $b^{(k)}_{(N-k-l)(k+l)}:= E_1^l b^{(k)}_{(N-k)k},$ then due to the $\mathfrak{U}_q(\mathfrak{gl}_n)$ invariance of the Hamiltonian (\ref{ham0}), we deduce that the states $b^{(k)}_{(N-k-l)(k+l)}$ are eigenstates of the Hamiltonian with eigenvalue $\Lambda_k.$ We say that the states $b^{(k)}_{(N-k-l)(k+l)}$ also belong to the $k$-sector. 

   Equation (\ref{eigen0}) describes the eigenvalue problem of an $ \binom{N}{k} \times \binom{N}{k},$ ($\binom{N}{k} :=\frac {N!}{k!(N-k)!}$) real symmetric matrix and leads to a polynomial of $\Lambda$ of degree $\binom{N}{k},$ i.e. there are $\binom{N}{k}$ roots for $\Lambda.$ Our main conjecture here is that these roots are distinct, that is (\ref{eigen0}) yields $\binom{N}{k}$ distinct eigenvalues. 
   
     Equations (\ref{eigen0}) for $k$ and $N-k$ are called complementary, they have the same number of eigenvalues $\binom{N}{k}$ and give exactly the same eigenvalues due to the $\mathfrak{U}_q({\mathfrak{gl}_2})$ symmetry of the Hamiltonian. Indeed, if ${\mathfrak b}_{(N-k)k}$ is an eigenstate for (\ref{eigen0}) with some eigenvalue $\Lambda,$ then the state ${\mathfrak b}_{k(N-k)}:= E_1^{N-2k}{\mathfrak b}_{(N-k)k}$ is an eigenstate for the problem (\ref{eigen0}) for $k\to N-k$, with the same eigenvalue $\Lambda.$
   \end{rem}
Our main conjecture (see Remark \ref{remark22}) states that equation (\ref{eigen0}) yields $\binom{n}{k}$ distinct eigenvalues. We showed for instance in Proposition \ref{systema} that for $k=1$ all eigenvalues are distinct. However, in order to prove the next main theorem (Theorem \ref{thetheorem}) it suffices to assume that eigenvalues in different sectors are distinct. Also, according Remark \ref{remark22} it suffices to study equation (\ref{eigen0}) for $0 \leq k \leq \frac{N}{2}$ and the solution for the rest of the equations can be then automatically obtained using symmetry arguments. A detailed description of the process is given in Theorem \ref{thetheorem}. 

We have thus far been able to work out explicitly the first two sectors (sectors $0$ and $1$) in Propositions \ref{sypro} and \ref{systema}. Specifically, we recall our results for the first two sectors focusing on $V_2^{\otimes N},$ $V_2\leq V_n,$ with the standard basis $\big \{\hat e_{x_1}, \hat e_{x_2}\big \}:$ 
\begin{enumerate}
\item 0-sector: there is $m_0=1$ eigenvalue $\Lambda_0$ with eigenstates denoted $b^{(0)}_{(N-m) m},$ $0\leq m\leq N,$ that form a canonical basis for the $N+1$ dimensional irreducible representation of ${\mathfrak U}_q(\mathfrak{gl}_2),$ Proposition \ref{sypro}.
    \item 1-sector: there are $m_1=N-1$ distinct eigenvalues denoted $\Lambda_1.$ Each one of these eigenvalues has eigenstates denoted $b^{(1)}_{(N-m) m},$ $1\leq m\leq N-1,$ that form a canonical basis for the $N-1$ dimensional irreducible representation of ${\mathfrak U}_q(\mathfrak{gl}_2),$ Proposition \ref{systema}.
 \end{enumerate}
We focus next on the ${\mathfrak U}_q(\mathfrak{gl}_2) \subseteq {\mathfrak U}_q(\mathfrak{gl}_n)$ part of the symmetry, generated by $\big \{E_1, F_1, q^{H_1} \big \},$ and prove that in the $k$-sector there are $m_k =\frac{N!}{k! (N- k+1)!}(N-2k +1)$ eigenvalues denoted $\Lambda_k.$ Each eigenvalue has eigenstates denoted $b^{(k)}_{(N-m)m},$ $k\leq m\leq N-k,$ which form a canonical basis for the $d_{k,2} =N-2k+1$ dimensional irreducible representation of ${\mathfrak U}_q(\mathfrak{gl}_2).$  That is, each $k$-sector is in fact a subspace of $V_2^{\otimes N}\subseteq V_n^{\otimes N},$ which is invariant under the action of $\mathfrak{U}_q(\mathfrak{gl}_2)$ as will be shown in the theorem below. 

\begin{thm} \label{thetheorem}
Let $ {\cal H} = \underset{1 \leq j\leq N-1}{\sum} r_j,$ where $r\in \End({\mathbb R}^{n}\otimes {\mathbb R}^n)$ is the $\mathfrak{U}_q(\mathfrak{gl}_n)$ invariant solution of the braid equation (\ref{rr}). Let also 
 \begin{equation}{\cal H}\ {\mathfrak b}_{(N-k) k } = \Lambda \  {\mathfrak b}_{(N-k)k}, \label{eigen1}
 \end{equation}
 where $\Lambda\in {\mathbb R},$ $k\in {\mathbb Z}^+,$ such that $0\leq k \leq \frac{N}{2}$
and ${\mathfrak b}_{(N-k) k}$ is given in (\ref{thestate}).
    Assume also that (\ref{eigen1}) yields distinct eigenvalues $\Lambda$ (see also Remark \ref{remark22}). Then:
    \begin{enumerate}
    \item 
There are $m_k =\frac{N!}{k! (N- k+1)!}(N-2k +1)$ eigenvalues, denoted $\Lambda_k$ 
with eigenstates denoted $b^{(k)}_{(N-k)k}$ being of the form (\ref{thestate}) that belong to the $k$-sector, i.e. 
$F_1 b^{(k)}_{(N-k) k} =0,$ where $F_1$ is given in (\ref{delta}).

\item Recall $E_1,\ F_1\ q^{H_1}$ (\ref{delta}) and let $b^{(k)}_{(N-m)m} := E_1^{m-k} b^{(k)}_{(N-k) k},$ $m>k$ then
\begin{eqnarray}
F_1 b^{(k)}_{(N-m-1)(m+1)} = \kappa_m^{(k)} b^{(k)}_{(N-m)m}, \quad E_1 b^{(k)}_{k (N-k)} =0, \quad q^{H_1}b^{(k)}_{(N-m)m} = a^{(k)}_{m} b^{(k)}_{(N-m)m}, \nonumber
\end{eqnarray}
where $\kappa^{(k)}_m = [N-k-m]_q[m-k+1]_q,$ $k \leq m \leq N-k-1$ and $a^{(k)}_m = q^{N-2m},$ $k\leq m \leq N-k.$

\item The eigenstates $b^{(k)}_{(N-m)m},$ $k\leq m \leq N-k,$ form an orthogonal basis of a $d_{k,2}=N-2k+1$ 
dimensional vector space and a canonical basis for the $d_{k,2}$ dimensional irreducible representation 
of ${\mathfrak U}_{q}(\mathfrak{gl}_2).$ 
\end{enumerate}  
\end{thm}
\begin{proof}
$ $
\begin{enumerate}
    \item We prove this statement by counting the eigenstates for the  eigenvalue problem (\ref{eigen1}) 
    for each $k\in {\mathbb Z}^+,$ $0\leq  k \leq \frac{N}{2}:$
    \begin{enumerate}
       \item From (\ref{eigen1}), for $k=0$:\\ 
       There is a $m_0=1$ eigenvalue denoted $\Lambda_0$ with eigenstate denoted $b^{(0)}_{N0},$ such that $F_1b^{(0)}_{N0} =0.$\\
       Also, define $b^{(0)}_{(N-m)m}:= E_1^{m}b^{(0)}_{N0},$ such that $E_1 b^{(0)}_{0N}=0$ (see also Proposition \ref{sypro}).
     
       \item From (\ref{eigen1}), for $k=1$:\\
       There is $m_0=1$ eigenvalue denoted $\Lambda _0$ with  corresponding eigenstate $b^{(0)}_{{(N-1)1}}:= E_1 b^{(0)}_{N0}.$\\
       There are $m_1=N-1$ eigenvalues denoted $\Lambda_1$ with corresponding eigenstates $b^{(1)}_{(N-1)1},$ such that $F_1 b^{(1)}_{(N-1)1} =0.$
       Also, define $b^{(1)}_{(N-m)m}:= E_1^{m-1}b^{(1)}_{(N-1)1},$ such that $E_1 b^{(1)}_{1(N-1)}=0$ (see Proposition \ref{systema}).
       
       \item We continue by iteration, for any $k$ in equation (\ref{eigen1}):\\
       There are $m_p$ eigenvalues denoted $\Lambda_p,$ $0\leq p \leq k-1,$ with corresponding eigenstates $b^{(p)}_{(N-k)k} := E_1^{k-p}b^{(p)}_{(N-p)p},$ such that $F_1 b^{(p)}_{(N-p)p} =0.$ 
    \end{enumerate}
We will first show that there exist $m_k$ eigenvalues denoted $\Lambda_k$ with corresponding eigenstates denoted $b^{(k)}_{(N-k)k}$ that satisfy $F_1 b^{(k)}_{N-k} =0.$
We count the number of all eigenvalues obtained from all previous sectors, i.e. sum up the number of eigenvalues coming from all previous sectors $0\leq p\leq k-1$
\begin{equation}
\sum_{0\leq p\leq k-1} m_p  = \sum_{0\leq p \leq k-1} \frac{N!}{p!(N-p+1)!}(N-2p+1) = \binom{N}{k-1}. \nonumber
\end{equation}
The eigenvalue problem (\ref{eigen1}) yields $\binom{N}{k}$ eigenvalues (see also Remark \ref{remark22}), hence there are indeed $m_k = \binom{N}{k} - \binom{N}{k-1}$ more eigenvalues denoted $\Lambda_k,$ different to the eigenvalues of all the previous sectors, and with corresponding eigenstates denoted $b^{(k)}_{(N-k)k}.$ 

We will now show that $F_1 b^{(k)}_{(N-k)k} =0;$ this can be shown by using symmetry arguments as in Proposition \ref{systema}.
Due to the $\mathfrak{U}_q(\mathfrak{gl}_n)$ symmetry of the Hamiltonian, we deduce that ${\cal H} F_1b^{(k)}_{(N-k)k} = \Lambda_k  F_1b^{(k)}_{(N-k)k}.$ From the definition of $F_1$ (\ref{delta}), it follows that
$F_1 b^{(k)}_{(N-k)k} \propto {\mathfrak b}_{(N-k+1)(k-1)}.$ However, a state ${\mathfrak b}_{(N-k+1)(k-1)}$ belongs to the $p<k$ sector with eigenvalue $\Lambda_p \neq \Lambda_k,$ hence $F_1b^{(k)}_{(N-k)k}=0.$

\item Recall $b^{(k)}_{(N-m)m} = E_1^{m-k} b^{(k)}_{(N-k)k},$  $m>k,$ we then show by induction that (see also Proposition \ref{systema}, the proof of part 2),
  \begin{equation}
  F_1 b^{(k)}_{(N-m-1)(m+1)} =[N-k-m]_q[m-k+1]_q b^{(k)}_{(N-m)m}. \nonumber
  \end{equation}
  Notice also from the equation above that for $m= N-k$ the sequence of states $b^{(k)}_{(N-m)m},$ $m\geq k$ terminates, hence $E_1b^{(k)}_{k(N-k)}=0$ (see a similar argument in the proof of Proposition \ref{systema}).

Also, given the structure of a state $b^{(k)}_{(N-m)m}$ and the definitions in (\ref{delta}) we show that  $q^{{\cal E}_1}b^{(k)}_{(N-m)m} = q^{N-m}b^{(k)}_{(N-m)m},$ $q^{{\cal E}_2}b^{(k)}_{(N-m)m} = q^{m}b^{(k)}_{(N-m)m}$ and $q^{H_1}b^{(k)}_{(N-m)m} = q^{N-2m}b^{(k)}_{(N-m)m},$ $k\leq m \leq  N-k.$ 
\item  The set of eigenstates ${\mathbb B}_k = \big  \{\hat b^{(k)}_{(N-m)m} \big  \},$ $k \leq m \leq N-k$ is an orthonormal basis of a $d_{k,2}=N-2k+1$ dimensional vector space denoted $V_2^{(\Lambda_k)}.$  We also collect the results of part (2) and conclude that ${\mathbb B}_k$ is also a canonical basis of the $d_{k,2}$ dimensional irreducible representation of $ \mathfrak{U}_q(\mathfrak{gl}_2):$ 
$F_1 b^{(k)}_{(N-k)k} =E_1b^{(k)}_{k(N-k)} =0$  and
       \begin{eqnarray}
&& E_1 b^{(k)}_{(N-m)m} = b^{(k)}_{(N-m-1) (m+1)}, \nonumber\\ 
&& F_1  b^{(k)}_{(N-m-1)(m+1)} = \kappa^{(k)}_{m} b^{(k)}_{(N-m) m}, \nonumber \\ 
&& q^{H_1}b^{(k)}_{(N-m) m} = a^{(k)}_m b^{(k)}_{(N-m)m},   \nonumber 
\end{eqnarray}
 where $\kappa^{(k)}_m = [N-k-m]_q[m-k+1]_q,$ $k \leq m \leq N-k-1$ and $a^{(k)}_m = q^{N-2m},$ $k\leq m \leq N-k.$
    \qedhere
    \end{enumerate}
\end{proof}
The normalized eigenstates $\hat b^{(k)}_{(N-m)m} =\frac{b^{(k)}_{(N-m)m}}{||b^{(k)}_{(N-m)m}||},$ $k\leq m \leq N-k$ constitute a canonical basis of the  $N-2k+1$ irreducible representation of $\mathfrak{U}_q(\mathfrak{gl}_2)$ (see also Remark \ref{norm1}):
\begin{eqnarray}
E_1 \hat b^{(k)}_{(N-m)m} =c^{(k)}_{m} \hat  b^{(k)}_{(N-m-1) (m+1)}, \quad F_1  \hat b^{(k)}_{(N-m-1)(m+1)} = c^{(k)}_{m} \hat b^{(k)}_{(N-m) m}, \quad q^{H_1} \hat b^{(k)}_{(N-m) m} = a^{(k)}_m {\hat b}^{(k)}_{(N-m)m},   \nonumber 
\end{eqnarray}
 where $c^{(k)}_m = \sqrt{[N-k-m]_q[m-k+1]_q},$ $k \leq m \leq N-k-1$ and $a^{(k)}_m = q^{N-2m},$ $k\leq m \leq N-k.$

We summarize in the following table (and the comments under the table) the main results from Theorem (\ref{thetheorem}) 
restricted to the $\mathfrak{U}_q(\mathfrak{gl}_2)$ section of the symmetry for $N$ even. 
\begin{center}
\begin{tabular}{ |c|c|c|c|c|c| } 
\hline
$(\Lambda_0, m_0)$ & $(\Lambda_1, m_1)$ &  $\ldots \ldots$  & $(\Lambda_k, m_k)$ & \ldots \ldots & $(\Lambda_{\frac{N}{2}}, m_{\frac{N}{2}})$\\
\hline
\multirow{3}{2em}{\ $b^{(0)}_{N0}$   \\ \ \\ \ \\ \  \\ \ } 
  &  &  & &  &\\ 
$b^{(0)}_{(N-1)1}$ & $b^{(1)}_{(N-1)1}$ &  &  & & \\ 
$\vdots$ & $\vdots$ & & & &\\
 $b^{(0)}_{(N-k)k}$ & $b^{(1)}_{(N-k)k}$ &  &  $b^{(k)}_{(N-k)k}$ & & \\
$\vdots$& $\vdots$ &   & \vdots  & & \\
$b^{(0)}_{\frac{N}{2}\frac{N}{2}}$ & $b^{(1)}_{\frac{N}{2}\frac{N}{2}}$ & $\vdots$ & $b^{(k)}_{\frac{N}{2}\frac{N}{2}}$ & \vdots & $b^{(\frac{N}{2})}_{\frac{N}{2}\frac{N}{2}}$ \\
$\vdots$ & $\vdots$ & & $\vdots$ &  &  \\
$b^{(0)}_{k(N-k)}$& $b^{(1)}_{k(N-k)}$ & & $b^{(k)}_{k(N-k)}$ & & \\
$\vdots$& $\vdots$ & &  &  &\\
$b^{(0)}_{1(N-1)}$& $b^{(1)}_{1(N-1)}$ &  &  && \\
$b^{(0)}_{0N}$&  & & & & \\\hline
\end{tabular}
\end{center}
\begin{center}
{Table 0}
\end{center}
The table above shows the canonical bases associated with each $N-2k+1$ dimensional irreducible representation of ${\mathfrak U}_q(\mathfrak{gl}_2).$
Specifically, the $k^{th}$ column in the table corresponds to the $k$-sector with $m_k$ eigenvalues denoted $\Lambda_k,$ as indicated on the top of each column. More precisely, we define the set of eigenvalues in the $k$-sector as $Eigen_{k}:= \big \{\Lambda_k^{(1)}, \Lambda_k^{(2)}, \ldots, \Lambda_k^{(m_k)}\big \} .$ Each eigenvalue $\Lambda^{(i)}_k,$ $1\leq i \leq m_k,$ has $d_{k,2} = N-2k+1$ eigenstates denoted $b^{(k,i)}_{(N-m),m},$ $k\leq m \leq N-k,$ such that $F_1 b^{(k,i)}_{(N-k),k} =0$ and $E_1b^{(k,i)}_{k(N-k)}=0,$ that form an orthogonal basis for a $d_{k,2}$ dimensional vector space. There are then $m_k$ orthogonal vector spaces of dimension $N-2k+1$ in the $k$-sector.
The well known fact, $$\underset{1\leq k \leq \frac{N}{2}}{\sum} m_k\ d_{k,2}  =  \underset{1\leq k \leq \frac{N}{2}}{\sum}\frac{N!}{k! (N-k+1)!} (N-2k+1)^2= 2^N$$ is also confirmed and indeed, if $n=2$ Table 0 provides the complete sets of eigenvalue and eigenstates. In general, for $n>2$ there are more eigenstates to each $k$-sector, beyond the $\mathfrak{U}_q(\mathfrak{gl}_2) \subseteq  \mathfrak{U}_q(\mathfrak{gl}_n)$ section generated by $\big \{E_1, F_1, q^{H_1} \big\},$ but this analysis will be undertaken in a future work. We note that if $N$ is odd then Table 0 is exactly the same except the last column, which has $m_{\frac{N-1}{2}}$ eigenvalues denoted $\Lambda_{\frac{N-1}{2}},$ with two  corresponding eigenstates,  $b^{(\frac{N-1}{2})}_{\frac{N+1}{2}\frac{N-1}{2}}$ and $b^{(\frac{N-1}{2})}_{\frac{N-1}{2}\frac{N+1}{2}}.$

We work out a few typical examples to illustrate the logic of the construction described in Theorem \ref{thetheorem}.
\begin{exa} \label{spectrum1}
Consider $N=2$ and any $n\in {\mathbb Z},$ $n>1$ (i.e. $X =\big \{x_1, x_2, \ldots, x_n\big \},$ $x_1<x_2 \ldots <x_n$). This is the simplest scenario 
${\cal H} = r_1.$ Then the eigenstates are (we only write the non-zero components $1\leq k_j \leq N,$ in $b^{(k)}_{k_1k_2\ldots k_n},$ see also Example \ref{exam2}):
\begin{itemize}
    \item For eigenvalue $\Lambda_0 =1$ ($0$-sector) the normalized eigenstates are ($x_i, x_j \in X$)
    \begin{eqnarray}
    &&\hat b^{(0)}_{{1}_i{1}_j} = \frac{1}{\sqrt{1+q^2}}( x_i \otimes x_j + q\ x_j \otimes x_i), ~~x_i < x_j, ~~~\mbox{and} ~~~~
    \hat b^{(0)}_{{2}_j}= x_j \otimes x_j. \label{symm3}
    \end{eqnarray}    
    The states (\ref{symm3}) are the $q$-symmetric states for $N=2$ and they form an orthonormal basis of an $\frac{n(n+1)}{2}$ dimensional vector space, Proposition \ref{sypro}.
\item For eigenvalue {$\Lambda_1 =-q^{-2}$} ($1$-sector),
\begin{equation}
\hat b^{(1)}_{{1}_i{1}_j} = \frac{1}{\sqrt{1+q^{-2}}} (x_i\otimes x_j - q^{-1} x_j \otimes x_i), ~~~x_i<x_j. \label{anti3}
\end{equation}
\end{itemize}
The states (\ref{anti3}) form an orthonormal basis for an $\frac{n(n-1)}{2}$ dimensional vector space.
I.e. \[V_n^{\otimes 2} = V_{\frac{n(n+1)}{2}}^{(\Lambda_0)} \oplus V_{\frac{n(n-1)}{2}}^{(\Lambda_1)} ~~\mbox{or schematically} ~~~  \begin{ytableau}
 \ 
\end{ytableau}\ \otimes\ \begin{ytableau}
 \ 
\end{ytableau} = \begin{ytableau}
 \ &   
\end{ytableau}\  \oplus \ \begin{ytableau}
 \    \\
 \
\end{ytableau}. \]
Also, the states (\ref{symm3}) and (\ref{anti3})form canonical bases corresponding to an $\frac{n(n+1)}{2}$ and an $\frac{n(n-1)}{2}$ representations of $\mathfrak{U}_q(\mathfrak{gl}_n)$ respectively (see also Propositions \ref{sypro} and \ref{systema}). 


{In the crystal limit $q \to 0:$ }
\begin{itemize}
\item  $\hat b^{(0)}_{{1}_i{1}_j} \to x_i \otimes x_j,$ $x_i< x_j$ 

\item $\hat b^{(1)}_{{1}_i{1}_j} \to x_j \otimes x_i$ (up to an overall minus sign) $x_i <x_j.$  

\end{itemize}
    \end{exa}

    \begin{exa} \label{exam3}
Consider $N=3,$ the three eigenvalues from the solution of the system (\ref{system}) are, $\Lambda_0= 2,$ $\Lambda^{(1)}_{1} = 1 - q^{-1} -q^{-2}$ and $\Lambda^{(2)}_{1} = 1 + q^{-1} - q^{-2}.$ $\Lambda_0$ belongs to the 0-sector, whereas $\Lambda^{(1)}_{1}$ and $\Lambda^{(2)}_{1}$ belong to the 1-sector.
\begin{itemize}
\item For $\Lambda_0 =2,$ one obtains all the symmetric states for the $\mathfrak{U}_q(\mathfrak{gl}_2)$ section, see Proposition \ref{sypro}:
\begin{eqnarray}
&& \hat b^{(0)}_{30} = x_1 \otimes x_1 \otimes x_1, \qquad\hat b^{(0)}_{03} = x_2 \otimes x_2 \otimes x_2, \nonumber \\
&& \hat b^{(0)}_{21} = \frac{1}{\sqrt{1 +q^2 + q^4}} \big ( x_1 \otimes x_1 \otimes x_2 + q x_1 \otimes x_2 \otimes x_1 +q^2 x_2 \otimes x_1 \otimes x_1\big ) \nonumber \\
&& \hat b^{(0)}_{12} = \frac{1}{\sqrt{1 +q^2 + q^4}} \big ( x_1 \otimes x_2 \otimes x_2 + q x_2 \otimes x_1 \otimes x_2 +q^2 x_2 \otimes x_2 \otimes x_1\big ). \nonumber
\end{eqnarray}
The eigenstates $\big \{\hat b^{(0)}_{30},\hat b^{(0)}_{21}, \hat b^{(0)}_{12}, \hat b^{(0)}_{03}\big \}$ form a  basis for the 4-dimensional representation $\mathfrak{U}_q(\mathfrak{gl}_2):$ $F_1 \hat b^{(0)}_{30} = E_1 \hat b^{(0)}_{03}=0,$ and
\begin{eqnarray}
    && E_1 \hat b^{(0)}_{30} = \sqrt{[3]_q}\hat b^{(0)}_{21}, \quad E_1 \hat b^{(0)}_{21} = [2]_q\hat b^{(0)}_{12}, \quad E_1 \hat b^{(0)}_{12} =  \sqrt{[3]_q}\hat b^{(0)}_{03} \nonumber\\
    &&  F_1\hat b^{(0)}_{21} =\sqrt{[3]_q}\hat b^{(0)}_{30}, \quad F_1\hat b^{(0)}_{12} =[2]_q\hat b^{(0)}_{21}, \quad  F_1\hat b^{(0)}_{03} =\sqrt{[3]_q}\hat b^{(0)}_{12} \nonumber\\
    && q^{H_1}\hat b^{(0)}_{30}= q^{3}\hat b^{(0)}_{30}, \quad q^{H_1}\hat b^{(0)}_{21}= q\hat b^{(0)}_{21},\quad  q^{H_1}\hat b^{(0)}_{12}= q^{-1}\hat b^{(0)}_{12}, \quad q^{H_1}\hat b^{(0)}_{03}= q^{-3}\hat b^{(0)}_{03}. \nonumber
    \end{eqnarray}
All symmetric states $\big \{ \hat b^{(0)}_{k_1k_2\ldots k_n}\big \},$ $k_j \in \{0, 1, 2, 3\}$ as derived in Theorem \ref{basicpro} and Proposition \ref{basiclemma} provide a basis of the $\frac{n(n+1)(n+2)}{6}$ dimensional irreducible representation of $\mathfrak{U}_q(\mathfrak{gl}_n).$

\item For $\Lambda_{1}^{(1)}  = 1 -q^{-1} -q^{-2},$ we find the eigenstates, restricted in the $\mathfrak{U}_q(\mathfrak{gl}_2)$ section:
\begin{eqnarray}
&& \hat b^{(1,1)}_{21} =  
\frac{1}{\sqrt{1 +q^2 + (1+q)^2}}\big(x_2 \otimes x_1 \otimes x_1 -(1+q) x_1 \otimes x_2 \otimes x_1 + q x_1 \otimes x_1 \otimes x_2\big),\nonumber\\ 
&& \hat b^{(1,1)}_{12} =  \frac{1}{\sqrt{1 +q^2 + (1+q)^2}}\big(-x_2 \otimes x_2 \otimes x_1 +(1+q) x_2 \otimes x_1 \otimes x_2 -q x_1 \otimes x_2 \otimes x_2\big). \nonumber
\end{eqnarray}

\item For $\Lambda_{1}^{(2)} = 1+q^{-1} -q^{-2},$ the corresponding eigenstates, restricted in the $\mathfrak{U}_q(\mathfrak{gl}_2)$ section, are
\begin{eqnarray}
&& \hat b^{(1,2)}_{21} =  \frac{1}{\sqrt{1 +q^2 + (1-q)^2}}\big(x_2 \otimes x_1 \otimes x_1 +(1-q) x_1 \otimes x_2 \otimes x_1 -q x_1 \otimes x_1 \otimes x_2\big), \nonumber\\
&& \hat b^{(1,2)}_{12} =  \frac{1}{\sqrt{1 +q^2 + (1-q)^2}}\big(x_2 \otimes x_2 \otimes x_1 +(1-q) x_2 \otimes x_1 \otimes x_2 -q x_1 \otimes x_2 \otimes x_2\big). \nonumber 
\end{eqnarray}
The eigenstates $\big \{\hat b^{(1,j)}_{21},\hat b^{(1,j)}_{12}\big \},$ $j \in \big \{1, 2\big \},$ form an orthogonal basis for the two-dimensional representation of $\mathfrak{U}_q(\mathfrak{gl}_2):$ $F_1 \hat b^{(1,j)}_{21} = E_1\hat b^{(1,j)}_{12}=0$ and
\begin{eqnarray}
   &&  E_1\hat b^{(1,j)}_{21} = \hat b^{(1,j)}_{12},   \quad  F_1\hat b^{(1,j)}_{12} = \hat b^{(1,j)}_{21},\quad   q^{H_1}b^{(1,j)}_{21} =q b^{(1,j)}_{21},  \quad q^{H_1}b^{(1,j)}_{12} =q^{-1} b^{(1,j)}_{12}.  \nonumber 
   \end{eqnarray}

\end{itemize}
\end{exa}
The crystal limit together with the $\mathfrak{U}_q(\mathfrak{gl}_n)$ generalization will be studied in a forthcoming article (see also relevant results in Propositions \ref{basiclemma} and \ref{sypro}).

\section{Combinatorial representations of the braid group}
\noindent In this section, we focus on combinatorial or set-theoretic solutions of the braid equation \cite{Drin, Eti, Sol, Gat1, Ru05, Ru07, JeOk, GuaVen}. It has been shown that all involutive combinatorial solutions of the braid equation can be obtained from the permutation operator via an admissible Drinfel'd twist (essentially a similarity transformation) \cite{Sol, Doikoutw, LebVen}. Consequently, the corresponding combinatorial automata are isomorphic to the \(q\)-permutation automaton at \(q=1\). We therefore turn our focus to non-involutive solutions of the braid equation, introducing certain algebraic structures that satisfy a self-distributivity condition and yield non-involutive combinatorial solutions \cite{Deho}. Specifically, we examine the dihedral quandle—a self-distributive structure—and analyze the eigenvalue problem of its corresponding braid equation solution. To reveal the structure of the associated eigenstates, we define rack and quandle automata, focusing primarily on the dihedral quandle automaton. Finally, we present preliminary results on the finite representations of the centralizers for the dihedral quandle solution.

\subsection{Self-distributive structures and braid representations}
Self-distributive structures, such as shelves, racks $\&$ quandles \cite{Jo82, Matv, Deho} 
satisfy axioms analogous to the Reidemeister moves used to manipulate knot diagrams and are associated with link invariants (see also biracks, biquandles). For recent reviews on self-distributive structures the interested reader is referred to \cite{Rev1, Rev2, Rev3}. 

\begin{defn}
    Let $X$ be a non-empty set and\, $\triangleright$\, a binary operation on $X$. Then, the pair $\left(X,\,\triangleright\right)$ 
is said to be a \emph{left shelf} if\, $\triangleright$\, is left self-distributive, namely, the identity
    \begin{equation}\label{eq:shelf}
        a\triangleright\left(b\triangleright c\right)
        = \left(a\triangleright b\right)\triangleright\left(a\triangleright c\right) 
    \end{equation}
    is satisfied, for all $a,b,c\in X$. Moreover, a left shelf $\left(X,\,\triangleright\right)$ is called 
    \begin{enumerate}
        \item  a \emph{left spindle} if $a\triangleright a = a$, for all $a\in X$;
        \item  a \emph{left rack} if $\left(X,\,\triangleright\right)$ is a \emph{left quasigroup}, i.e., the maps $L_a:X\to X$ defined by $L_a\left(b\right):= a\triangleright b$, 
        for all $b\in X$, are bijective, for every $a\in X$.
        \item a \emph{quandle} if $\left(X,\,\triangleright\right)$ is both a left spindle and a left rack.
    \end{enumerate}
\end{defn}

We are mostly interested in racks and quandles here, given that we always require invertible solutions of
the Yang-Baxter equation. We provide below some fundamental known cases of quandles and racks (see also \cite{Rev1, Rev2, Rev3}):
\begin{enumerate}[{(a)}]
   \item {{\bf Conjugate quandle.} }
   Let {$(X, \cdot)$} be a group and define $\triangleright: X \times X \to X,$ 
such that {$a\triangleright b =a^{-1} \cdot b \cdot a.$} Then $(X, \triangleright)$ is a quandle.    
    
    \item {\bf Core quandle.} Let {$(X, \cdot)$} be a group and $\triangleright: X \times X \to X,$ 
such that {$a\triangleright b =a \cdot b^{-1} \cdot a.$} Then $(X, \triangleright)$ is a quandle.

\item {\bf Alexander (affine) quandle.} Let $Q$ be a ${\mathbb Z}[t, t^{-1}]$ ring module and 
$\triangleright: Q \times Q \to Q,$ $a\triangleright b = (1-t)a +bt,$ then $(Q, \triangleright)$ is a quandle. 

    \item {\bf Rack, but not quandle.} Let $(G,\cdot)$ be a group and define 
    $\triangleright: G \times G \to G,$ such that $a\triangleright b = b\cdot a^{-1} \cdot x \cdot a,$ 
    where $x\in G$ is fixed. Then $(G, \triangleright)$ is a rack, but not a quandle.   \end{enumerate}

We also present below some concrete examples of finite quandles:
\begin{exa}  \label{exx}

$ $
    \begin{enumerate}
    \item {\bf The dihedral quandle.}
     Let $i,j \in X = {\mathbb Z}_n$ and define $\triangleright: X \times X \to X,$ such that $i\triangleright j = 2i -j$ $mod n,$ then  $(X, \triangleright)$ is a quandle. 
     This is a core quandle with an abelian group. An explicit table of the action $\triangleright$ is presented below for $n=3$
and $X = \big \{x_1=0, x_2=1, x_3=2 \big \}:$
\begin{center}
\begin{tabular}{ |c|c|c|c| } 
\hline
$\triangleright$ & $x_1$ &  $x_2$  & $x_3$\\
\hline
\multirow{3}{2em}{\ $x_1$ \\ \ $x_2$\\ \ $x_3$ } 
& $x_1$ & $x_3$ & $x_2$ \\ 
& $x_3$ & $x_2$ & $x_1$  \\ 
& $x_2$ & $x_1$ & $x_3$ \\
\hline
\end{tabular}
\end{center}
\begin{center}
{Table 1}
\end{center}
    \item {\bf The tetrahedron quandle.} Let $X = \big \{1,2,3,4\big \}$ and define $\triangleright: X \times X \to X,$ 
    such that $1\triangleright = (234),$ $2\triangleright = (143),$ $3\triangleright = (124)$ and $4\triangleright =(132).$ 
    This is also a cyclic quandle.  We construct below the explicit table of the action  $\triangleright:$
\begin{center}
\begin{tabular}{ |c|c|c|c|c| } 
\hline
$\triangleright$ & $x_1$ &  $x_2$  & $x_3$ & $x_4$ \\
\hline
\multirow{3}{2em}{\ $x_1$ \\ \ $x_2$\\ \ $x_3$ \\ \ $x_4$ } 
& $x_1$ & $x_3$ & $x_4$ & $x_2$\\ 
& $x_4$ & $x_2$ & $x_1$ & $x_3$\\ 
& $x_2$ & $x_4$ & $x_3$  & $x_1$ \\
& $x_3$ & $x_1$ & $x_2$  & $x_4$ \\
\hline
\end{tabular}
\end{center}
\begin{center}
{Table 2}
\end{center}
\end{enumerate}
\end{exa}

We recall now a fundamental statement regarding shelves and solutions of the set-theoretic Yang-Baxter equation.
\begin{pro} \label{shelf2}
We define 
the binary operation $\triangleright:  X \times X \to X,$ $(a,b) \mapsto a \triangleright b.$ 
Then $r: X \times X \to X \times X$, such that for all $a,b \in X, $  $r(a,b) = (b, b \triangleright a)$ 
is a solution of the set-theoretic braid equation if and only if $(X, \triangleright)$ is a shelf.
\end{pro}
\begin{proof} 
The proof is straightforward by direct substitution in the Yang-Baxter equation and comparison between 
LHS and RHS.
\end{proof}
If $r: X \times X \to X \times X$, such that for all $a,b \in X, $  $r(a,b) = (b, b \triangleright a)$ 
is an invertible braid solution then $(X, \triangleright)$ is a rack (or a quandle).

The graphical representation of the shelve solution $r(a,b) = (b, b\triangleright a)$:
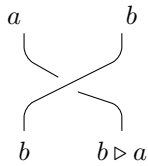
\begin{figure}[ht]
\centering
\begin{tikzpicture}[scale=1.3]
\draw [rounded corners](0,0)--(0,0.25)--(1,0.75)--(1,1);
\draw [rounded corners](1,0)--(1,0.25)--(0.55,0.4);
\draw [rounded corners](0,1)--(0,0.75)--(0.35,0.55);
\node  at (0,-0.2)  {$b$};
\node  at (1,-0.2)  {$b \triangleright a$};
\node  at (-.1,1) [above] {$a$};
\node  at (1.1,1) [above] {$b$};
\end{tikzpicture}
\end{figure}

We are interested here in invertible solutions of the braid equation, so we are focusing on rack solutions.
We note that the inverse of $r$ above is $r^{-1}: X \times X \to X \times X,$ 
$r^{-1} (a,b) = (a \triangleright^{-1} b, a),$  such that $a\triangleright(a\triangleright^{-1}b) = a \triangleright^{-1}(a\triangleright b)=b$ for all $a,b \in X.$ Notice also that a different map denoted $r': X\times X \to X\times X,$ such that
$r'(a,b) = (a \triangleright b, a)$ is also a solution of the braid equation.

\subsection{Self-distributive automata}
We consider tensor representations of the braid group\\ $\rho: B_N \to \End\big (({\mathbb C}^n)^{\otimes N}\big ),$ such that 
\begin{equation}
\rho(\sigma_i) :=  r_{i} = 1_n\otimes \ldots \otimes 1_n \otimes  \underbrace{r}_{i, i+1~~ \mbox{positions}} \otimes 1_n  \ldots \otimes 1_n \label{frep2}
\end{equation} 
where $r$ is the linearized version of the rack or quandle solution and is expressed as an $n^2 \times n^2$ matrix, such that
\begin{equation}
r\ a\otimes b = b \otimes b\triangleright a \label{rr23}
\end{equation}
where $(X, \triangleright)$ is either a rack or a quandle. Recall the simplified notation introduced earlier in the manuscript, $a\otimes b : = \hat e_a \otimes \hat e_b,$ $a,b \in X.$

Based on (\ref{rr23}), we define next the rack and quandle automaton. We have to focus on specific examples of quandles in order to examine the corresponding automata. The type of each automaton is characterized by the dimension $n$ as well as the order of the matrix $r.$ Indeed, see below the tree order diagrams for $N=2,$ in the case where the order of $r$ is $k$ ($r^k =1$):
\begin{center}
\begin{tikzpicture}[shorten >=1pt,node distance=2.0cm,on grid,auto] 
   \node[] (q_1)   {$1$}; 
   \node[] (q_2) [right=of q_1] {$r$};
    \node[] (q_3) [right=of q_2] {$r^2$};
 \node[] (q_4) [right=of q_3] {$\ldots$};
     \node[] (q_n) [right=of q_4] {$r^{k-1}$};    
     \path[->]
  (q_1) edge[left, below]  (q_2)
  (q_2) edge[left, below]  (q_3)
  (q_3) edge[left, below]  (q_4)
  (q_4) edge[left, below]  (q_n) ;    
  \end{tikzpicture}    
    \end{center}
As we have seen already in the previous sections braided automata provide clusters of eigenstates of linear combinations of all length-one words of the Hecke algebra ${\cal H}_N(q).$ Analogously, rack or quandle automata should provide again clusters of eigenstates of linear combinations of all length-one words of the braid group $B_N.$ Here we only focus on the eigenvalue problem of $r$, i.e.  we focus on the braid group for $N=2.$ 

\begin{defn} \label{auto3} {\bf (Rack $\&$ quandle automaton).}   

\begin{enumerate}
\item Let the set of states be $Q = {\mathbb B}^{\otimes N}_n,$ (see Remark \ref{r5}, i.e. $Q$ consist of $n^N$ states). 

\item Let the alphabet be  $\Sigma : =\big \{s_1, s_2, \ldots, s_{N-1}\big \}$ and the transition matrices are $r_i,$ $ i \in [N-1]$ and are given by the tensorial representation of the braid group $\rho: B_N \to \End\big (({\mathbb C}^n)^{\otimes N}\big ),$ such that
$\sigma_j \mapsto  r_j = 1_n\otimes \ldots \otimes 1_n \otimes  \underbrace{r}_{j, j+1 ~~ \mbox{positions}} \otimes 1_n  \ldots \otimes 1_n,$ where 
for all $a,b \in X,$ $r\ a \otimes b = b \otimes b \triangleright a$ and  $(X, \triangleright)$ is a rack or a quandle.
\end{enumerate}
This is a braided automaton called a rack or quandle automaton. 
\end{defn}

In the case where $(X, \triangleright)$ is a shelve, then the automaton above is called {\it shelve} automaton.
Self-distributive automata are naturally combinatorial automata given that
\begin{equation}
r_j\ x_{i_1} \ldots \otimes  \underbrace{x_{i_j} \otimes x_{i_{j+1}}}_{j, j+1\ positions}   \ldots \otimes x_{i_N} =  x_{i_1}  \ldots  \otimes \underbrace{x_{i_{j+1}} \otimes x_{i_{j+1}}\triangleright x_{i_j}}_{j, j+1\ positions} \ldots \otimes x_{i_N}. 
\end{equation}

We describe in detail the dihedral quandle automaton based on Examples \ref{exx}.
We focus on the dihedral quandle and recall that for all $x,y \in {\mathbb Z}_n$ 
\begin{equation}
r\ x\otimes y = y \otimes  y\triangleright x,~~~ x\triangleright y = 2x -y, \label{rr234}
\end{equation}
where  addition and subtraction are defined $modn.$  We focus on $N=2,$ $\Sigma = \big \{ s \big \}$ and the transition matrix $r$ is given in (\ref{rr234}). 
Then the corresponding quandle automaton is depicted below (Figure 8).
\begin{center}
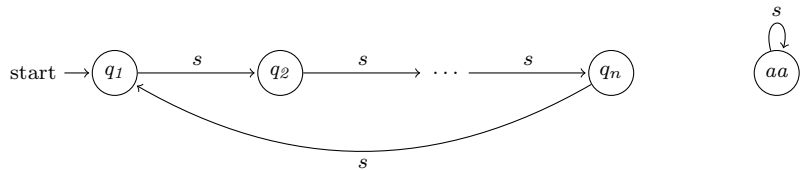
\begin{figure}[ht]
\footnotesize
\begin{tikzpicture}[shorten >=1pt,node distance=2.2cm,on grid,auto] 
\node[state, inner sep=0, minimum size=2.0em, initial] (q1) {${\it q_1}$};
\node[state,inner sep=0, minimum size=2.0em, right of=q1] (q2) {${\it q_2}$};
\node[inner sep=0, minimum size=2.0em, right of=q2] (q3) {$\ldots$};
\node[state,inner sep=0, minimum size=2.0em, right of=q3] (q3b) {$q_n$};
\node[state,inner sep=0, minimum size=2.0em, right of=q3b] (q4) {${\it aa}$};
\path[->] 
(q1) edge[left, above] node{$s$} (q2)
(q2) edge[left, above] node{$s$} (q3)
(q3) edge[left, above] node{$s$} (q3b)
(q3b) edge[bend left, below] node{$s$} (q1)
(q4) edge[loop above] node{$s$} (q4);
\end{tikzpicture}
\caption{Dihedral quandle automaton} \label{dihb}
\end{figure}
\end{center}
$q_1\in \big \{x_1\otimes x_2, x_1\otimes x_3, \ldots, x_1\otimes x_n\big \},$ that is we have $n-1$ such disconnected graphs, and $n$ diagrams of the type on the right in Figure 8, ${\it a} \in {\mathbb Z}_n$ $(x_k = k-1,\ k\in[n]),$ we consider $n$ to be odd.  
Specifically, recall that $r\ a\otimes b = b \otimes {b} \triangleright a,$ for all $a,b \in {\mathbb Z}_n,$ and
if for instance $q_1 = x_1\otimes x_2,$ then $q_2 = x_2 \otimes  x_3,\ q_3 = x_3 \otimes x_4, \ldots, q_n = x_n\otimes x_1.$ 
The $n-1$ disconnected diagrams on the left and the $n$ disconnected diagrams on the right in Figure 8 provide clusters of eigenstates of the dihedral solution $r,$ as will be shown in the proposition below.
We focus in the following proposition in the case of the dihedral quandle for $n$ odd. 
\begin{pro}  \label{dih0}
Let $ X = {\mathbb Z}_n,$ $n$ odd, and $r= \underset{a,b\in X}{\sum} e_{b,a} \otimes e_{b\triangleright a, b}$ where $(X, \triangleright)$ is the dihedral quandle, such that for all $a,b \in X$ $a\triangleright b =2a-b,$ $modn$ Then,
\begin{enumerate}
   \item The eigenvalues of the $r$-matrix are the $n$-roots of unity, $\Lambda_k = e^{\frac{2\pi ki}{n}},$ $k\in [n].$
   \item The decomposition of $V_n^{\otimes}$ ($V_n= {\mathbb C}^n$) in terms of eigenspaces of $r$ is as follows,
   \begin{equation}
   V_n^{\otimes 2} =\bigoplus_{k=1}^{n-1} V^{(\Lambda_k)}_{n-1} \oplus  V^{(\Lambda_n)}_{2n-1},    \label{dec}
   \end{equation}
   where the superscripts denote the eigenvalue of $r$ and the subscripts denote
the dimension of the corresponding eigenspace. 
\end{enumerate}
\end{pro}
\begin{proof}

$ $

\begin{enumerate}
\item The first part immediately follows from the dihedral quandle automaton in Figure 8, where we observe that $r^n=1$ it then follows that the eigenvalues of $r$ are the $n$-roots of unity, $\Lambda_k=  e^{\frac{2\pi k i}{n}},$ $k\in [n].$ To identify the corresponding eigenstates we consider the state set of states $\big \{x_1\otimes x_2, x_1 \otimes x_3, \ldots , x_1\otimes x_n\big \}$ as initial states. 
We define $G_k: = \frac{1}{\sqrt{n}}(1 + \Lambda_k^{-1} r + \Lambda_k^{-2} r^2 + \ldots + \Lambda_k^{-(n-1)}r^{n-1})$ and 
\begin{equation}
\hat b_{1m}^{(k)} := G_k\ x_1 \otimes x_m,~~2\leq m \leq n,~~ k \in [n], \quad \mbox{and} \quad  \hat b^{(n)}_{m} := x_m \otimes x_m, ~~ m\in [n]. \nonumber
\end{equation}
It then follows that, $r\ \hat b_{1m}^{(k)} =\Lambda^{(k)} \hat b^{(k)}_{1m}$ and  $r\ \hat b^{(n)}_{m} = \hat b^{(n)}_m.$ We show for instance
\begin{eqnarray}
r\ \hat b_{12}^{(k)} = \frac{1}{\sqrt{n}} \big (\Lambda_k^{-(n-1)}x_1 \otimes x_2 + \ldots +x_n\otimes x_1\big )
= \Lambda_k\hat b^{(k)}_{12}.  \nonumber
\end{eqnarray}

\item The geometric multiplicities are described below. 
\begin{itemize}
    \item For $\Lambda_k = e^{\frac{2k\pi i }{n}},$ $k\in [n-1]$    
    there are $n-1$ eigenstates $\hat b_{1m}^{(k)},$ $~k\in [n-1],$ $2\leq m \leq n,$ orthogonal to each other, i.e. they form the basis of an $n-1$ dimensional vector space.

\item For $\Lambda_n =1,$ there are $n-1$ eigenstates $\hat b^{(n)}_{1m},$ $2\leq m\leq n$ and $n$ eigenstates $\hat b^{(n)}_m$, $m\in [n],$ i.e. there are $2n-1$ eigenstates in total, which are orthogonal to each other forming the basis of an $2n-1$ dimensional vector space.
\end{itemize}
Notice also that all eigenspaces are orthogonal to each other, because $r$ is an orthogonal matrix (this is also easily explicitly checked). That is, the decomposition of the $V_n^{\otimes}$ space (\ref{dec}) holds. \qedhere
\end{enumerate}
\end{proof}

We will work out an explicit example ($n=3$) for the quandle automaton using the dihedral quandle. 
\begin{exa} \label{exaq1}
We first consider as an illustrative example the dihedral quandle for $n=3,$ i.e. $r = \underset{a,b \in {\mathbb Z}_3}{\sum} e_{b,a} \otimes e_{b\triangleright a, b},$ such that for all $a,b \in {\mathbb Z}_3,$ $a\triangleright b = 2a-b,$ where the addition and subtraction are defined $mod3.$ The corresponding combinatorial automaton, showing the action of $r$ on the states $x_i \otimes x_j,$ $x_i, x_j \in {\mathbb Z}_3,$  consists of the following disconnected graphs ($x \in {\mathbb Z}_3,$ specifically, $x_1 =0, x_2 =1, x_3=2$):
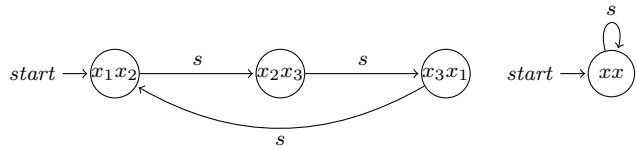
\begin{center}
\begin{figure}[ht]
\footnotesize
\begin{tikzpicture}[shorten >=1pt,node distance=2.2cm,on grid,auto] 
\node[state, inner sep=0, minimum size=2.0em, initial, right of=q_4] (q1) {${ x_1x_2}$};
\node[state,inner sep=0, minimum size=2.0em, right of=q1] (q2) {${ x_2x_3}$};
\node[state,inner sep=0, minimum size=2.0em, right of=q2] (q3) {$x_3x_1$};
\node[state,inner sep=0, minimum size=2.0em, initial, right of=q3] (q4) {${ xx}$};
\path[->] 
(q1) edge[left, above] node{$s$} (q2)
(q2) edge[left, above] node{$s$} (q3)
(q3) edge[bend left, below] node{$s$} (q1)
(q4) edge[loop above] node{$s$} (q4);
\end{tikzpicture}
\end{figure}
\end{center}
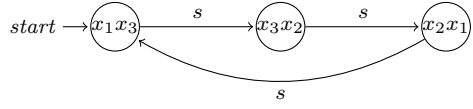
\begin{center}
\begin{figure}[ht]
\footnotesize
\begin{tikzpicture}[shorten >=1pt,node distance=2.2cm,on grid,auto] 
\node[state, inner sep=0, minimum size=2.0em, initial] (q1) {${ x_1x_3}$};
\node[state,inner sep=0, minimum size=2.0em, right of=q1] (q2) {${x_3x_2}$};
\node[state,inner sep=0, minimum size=2.0em, right of=q2] (q3) {$x_2x_1$};
\path[->] 
(q1) edge[left, above] node{$s$} (q2)
(q2) edge[left, above] node{$s$} (q3)
(q3) edge[bend left, below] node{$s$} (q1);
\end{tikzpicture}
\end{figure}
\end{center}
In this case $r^3 =1,$  the eigenvalues of $r$ are $\Lambda_k = e^{\frac{2k\pi  i}{3}},$ $k \in \big \{1,2,3\big \}$ and the
 generator of eigenstates is defined then as $G_k := \frac{1}{\sqrt{3}} (1 + \Lambda^{-1}_k r + \Lambda_k^{-2} r^2).$ We then define, $\hat b^{(k)}_{1m}:= G_k \ x_1 \otimes  x_m,$ $m \in \big \{2,3 \big \}$ and $\hat b^{(3)}_{m}:= x_m \otimes x_m,$  $m\in [3].$ 

The eigenvalues and the corresponding normalized eigenstates of $r$ are given below.
\begin{enumerate}
    \item $\Lambda_k =e^{\frac{2 k\pi  i}{3}},$ $k\in \big \{1,2\big \}$ and the corresponding eigenstates are,
    \begin{eqnarray}
 && \hat b_{12}^{(k)} = \frac{1}{\sqrt{3}}(x_1 \otimes  x_2 + \Lambda_k^{-1}  x_2 \otimes x_3 + \Lambda_k^{-2}  x_3 \otimes  x_1),  
 ~~~\hat b_{13}^{(k)} = \frac{1}{\sqrt{3}}(x_1 \otimes x_3 + \Lambda_k^{-1} x_3 \otimes x_2 + \Lambda_k^{-2}   x_2 \otimes  x_1). \nonumber
 \end{eqnarray}
Then $\big \{\hat b^{(k)}_{12}, \hat b_{13}^{(k)}\ \big \}$ is the orthonormal basis of a two
dimensional vector space denoted $V_2^{(\Lambda_k)}.$ The two vector spaces are orthogonal to each other.
\item $\Lambda_3 =1 ,$ and eigenstates,
    \begin{eqnarray}
    && \hat b_1^{(3)} =x_1 \otimes x_1, ~~~~\hat b_{2}^{(3)}  = x_2 \otimes x_2, ~~~\hat b_3^{(3)}= x_3 \otimes x_3\nonumber\\    
    && \hat b_{12}^{(3)} =\frac{1}{\sqrt{3}}(x_1 \otimes  x_2 +  x_2 \otimes x_3 +  x_3  \otimes x_1), ~~~\hat b_{13}^{(3)} =\frac{1}{\sqrt{3}}(x_1 \otimes x_3 +   x_3 \otimes x_2 + x_2 \otimes  x_1)\nonumber
    \end{eqnarray}
     $\big \{\hat b_1^{(3)},\hat b_2^{(3)},\hat b_2^{(3)}, \hat b_{12}^{(3)}, \hat b_{13}^{(3)}\big \}$ is an orthonormal basis of a five dimensional vector space denoted $V_5^{(\Lambda_3)}.$     
     \end{enumerate}
And as expected we deduce
\begin{equation}
V_3^{\otimes 2 } = V_2^{(\Lambda_1)} \oplus V^{(\Lambda_2)}_2 \oplus V_5^{(\Lambda_3)}.
\end{equation}
\end{exa}


\subsection{Centralizers}
We first introduce the rack or quandle group which is a subset of the rack or quandle algebra introduced in \cite{DoRySt}.
\begin{defn} \label{groupq}
Let $X$ be a non-empty set and $(X, \triangleright)$ be a rack. The group ${\cal G}$ generated by elements 
$q_a, q_a^{-1},$ $1_{\cal G}$ (unit element) and relations for all $a,b \in X,$ 
\begin{equation}
q_aq_a^{-1} = q_a^{-1}q_a= 1_{\cal G},~~q_aq_b = q_{b} q_{b\triangleright a.} \nonumber
    \end{equation}
is called a {\it rack group}. If $ (X, \triangleright )$ is a quandle then the group is called a quandle group.
\end{defn}
As any group, the rack group ${\cal G}$ is also a Hopf algebra $k{\cal G}$ over some field $k,$ equipped with a coproduct, 
counit (group homomorphisms) and antipode (anti-homomorphism) (see also \cite{DoRySt}), for all $a\in X$:
\begin{enumerate}
\item $\Delta: {\cal G} \to {\cal G} \otimes {\cal G},$ $\Delta(q_a) = q_a \otimes q_a.$
\item $\epsilon: {\cal G}\to k,$ $\epsilon(q_a) = 1.$
\item  $s: {\cal G} \to {\cal G},$ $s(q_a)=q_a^{-1}.$
\end{enumerate}
Coassociativity holds $(\Delta \otimes \id)\Delta = ( \id\otimes \Delta)\Delta,$ hence $\Delta^{(N)}(q_a) = q_a \otimes \ldots \otimes q_a,$ $a\in X.$

We consider the fundamental representation of a rack group (focus on finite sets of cardinality $n$),
let $\rho: {\cal G} \to \End({\mathbb C}^n),$ such that
\begin{equation}
q_a \mapsto M_a := \sum_{b\in X}e_{b,a\triangleright b}. \label{mm}
\end{equation}
We also define $\Delta^{(N)}(M_a) := \rho^{\otimes N} \Delta^{(N)}(q_a) \in \End(({\mathbb C}^n)^{\otimes N}).$
\begin{lemma} 
Let $r= \underset{b\in X}{\sum} e_{b,a} \otimes e_{b\triangleright a, b} \in \End(({\mathbb C}^n)^{\otimes 2})$ 
be an invertible solution of the braid equation. Recall also that for all $a\in X,$ $M_a \in \End({\mathbb C}^n)$ is defined in (\ref{mm}) and  $\Delta^{(N)}(M_a) : = M_a \otimes \ldots \otimes M_a.$
\begin{enumerate}
\item Then,
\begin{equation}
r \Delta(M_a) = \Delta(M_a) r, ~~\forall a\in X.\nonumber
\end{equation}
\item Let also $r_j := 1 \otimes \ldots \otimes \underbrace{r}_{j, j+1\ positions} \otimes \ldots \otimes 1,$ $j\in [N-1],$ then 
\begin{equation}
r_j \Delta^{(N)}(M_a) = \Delta^{(N)}(M_a) r_j,~~~\forall a\in X.
\end{equation}
\item Let $\psi \in {\mathbb C}^n \otimes {\mathbb C}^n,$ such that $r\ \psi = \Lambda \psi,$ $\Lambda \in {\mathbb C},$ 
and define $\psi'_a := \Delta(M_a)\psi$ for all $a \in X,$ then $r\ \psi'_a = \Lambda \psi'_a.$ 
\end{enumerate}
\end{lemma}
\begin{proof} The first part is shown by direct matrix multiplication, 
whereas the second and third part are immediate consequences of the first part.
\end{proof}
We now focus on the fundamental representation of the rack (or quandle) group, $M_a =\underset{b \in X}{\sum} e_{b, a\triangleright b},$ and try to work out tensor representations of the group. In general, we observe that $\Delta(M_a)\ a\triangleright b \otimes a\triangleright c =b \otimes c, $ for all $a, b, c\in X.$  We work out a specific example below for $(X, \triangleright)$ being the dihedral quandle for $n=3.$ A general study on the representation theory of rack and quandle groups will be presented in a separate investigation.

\begin{exa}
We consider the dihedral quandle for $n=3,$ (see also Example \ref{exaq1}). We then identify the three generators (recall $x_1 =0, x_2 =1, x_3 =2.$)
$$M_{x_1} = \begin{pmatrix} 1 &0 &0\\ 0&0 & 1\\ 0& 1 &0 \end{pmatrix},
~~ M_{x_2} = \begin{pmatrix} 0 &0 &1\\ 0&1 & 0\\ 1& 0 &0 \end{pmatrix},~~M_{x_3} = \begin{pmatrix} 0 &1 &0\\ 1&0 & 0\\ 0& 0 &1 \end{pmatrix}$$
Let ${\mathrm M}_a := \Delta(M_a),$ $a\in X$  and recall from Example \ref{exaq1}, $\Lambda_k = e^{\frac{2k\pi i}{3}},$ $k\in [3]$ are the eigenvalues of $r,$ whereas $\hat b^{(k)}_{1m},$ $k\in [3],\ m\in \big \{2, 3\big \}$ and $ \hat b^{(3)}_{m},$ $m\in [3]$ are the corresponding eigenstates. We then observe that $M_a^2 = 1_3,$ $a \in {\mathbb Z}_3$ and 
\begin{eqnarray}
&& {\mathrm M}_{x_1} \hat b^{(k)}_{12} = \hat b^{(k)}_{13},~~~{\mathrm M}_{x_2} \hat b^{(k)}_{12} = \Lambda_k\hat b^{(k)}_{13},
~~~{\mathrm M}_{x_3} \hat b^{(k)}_{12} = \Lambda_k^{-1}\hat b^{(k)}_{13}\nonumber \\ 
&&{\mathrm M}_{x_1} \hat b^{(3)}_1 = \hat b_1^{(3)}, ~~~{\mathrm M}_{x_1} \hat b^{(3)}_2 = \hat b_3^{(3)}, ~~~{\mathrm M}_{x_1} \hat b^{(3)}_3 = \hat b_2^{(3)} \nonumber \\
&&{\mathrm M}_{x_2} \hat b^{(3)}_1 = \hat b_3^{(3)}, ~~~{\mathrm M}_{x_2} \hat b^{(3)}_2 = \hat b_2^{(3)}, ~~~{\mathrm M}_{x_2} \hat b^{(3)}_3 = \hat b_1^{(3)} \nonumber \\
&&{\mathrm M}_{x_3} \hat b^{(3)}_1 = \hat b_2^{(3)}, ~~~{\mathrm M}_{x_3} \hat b^{(3)}_2 = \hat b_1^{(3)}, ~~~{\mathrm M}_{x_3} \hat b^{(3)}_3 = \hat b_3^{(3)}. \label{qua3}
\end{eqnarray}
\end{exa}
In general, we can define the dihedral quandle group automaton, which consists of an alphabet $\Sigma  = {\mathbb Z}_n,$ a set of states ${\mathbb B} = \big \{ b^{(k)}_{1m}, b^{(n)}_{m} \big \},$ $k \in [n]$ (the states are given in Example \ref{exaq1}) and transition matrices ${\mathrm M}_a:= \Delta(M_a): {\mathbb B} \to {\mathbb B},$ $a \in \Sigma,$ given in (\ref{qua3}). Figures 9 and 10 below depict the dihedral quandle group automaton for $n=3$ (see equations \ref{qua3}). Let $x_j\in {\mathbb Z}_3,$ $c_{1,k} = 1,$ $c_{2,k} =\Lambda_k$ and $c_{3,k} =\Lambda_k^{-1},$ $k \in [2].$
\begin{center}
\begin{figure}[ht]
\footnotesize
\begin{tikzpicture}[shorten >=1pt,node distance=3.0cm,on grid,auto] 
\node[state, inner sep=0, minimum size=2.0em, initial] (q1) {${ \hat b^{(k)}_{12}}$};
\node[state,inner sep=0, minimum size=2.0em, right of=q1] (q2) {$ \hat b^{(k)}_{13}$};
\path[->] 
(q1) edge[bend left, above] node{${x_j};c_{j,k}$} (q2)
(q2) edge[bend left, below] node{${x_j};c^{-1}_{j,k}$} (q1);
\end{tikzpicture}
\caption{Two dimensional representation} \label{2d}
\end{figure}
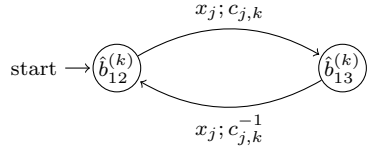
\end{center} 
\begin{center}
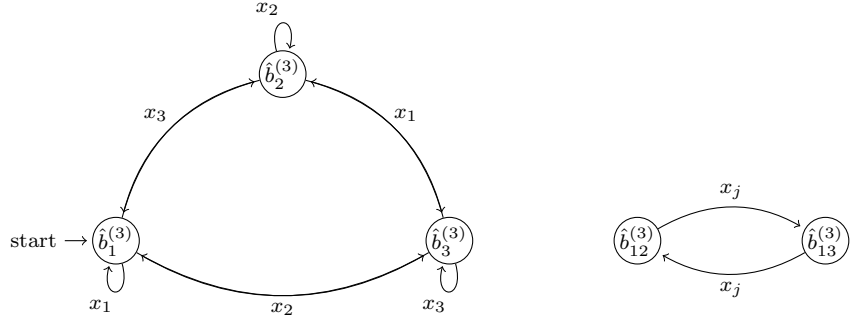
\begin{figure}[ht]
\footnotesize
\begin{tikzpicture}[shorten >=2pt,node distance=2.5cm] 
   \node[state,inner sep=0, minimum size=2.0em, initial] (q_0)   {$\hat b^{(3)}_1$}; 
   \node[state,inner sep=0, minimum size=2.0em] (q_1) [above right=of q_0] {$\hat b^{(3)}_2$}; 
   \node[state,inner sep=0, minimum size=2.0em] (q_2) [below right=of q_1] {$\hat b^{(3)}_3$}; 
   \node[state,inner sep=0, minimum size=2.0em, right of=q_2] (q_3) {$\hat b^{(3)}_{12}$};
   \node[state,inner sep=0, minimum size=2.0em, right of=q_3] (q_4)  {$\hat b^{(3)}_{13}$};   
   \path[->] 
    (q_0) edge[bend left, above]  node {$x_3~~~~$} (q_1)  
    (q_0) edge[loop below]  node {$x_1~~~~$} (q_0)
    (q_1) edge[loop above]  node {$x_2~~~~$} (q_1) 
    (q_2) edge[loop below]  node {$x_3~~~~$} (q_2)    
    (q_1) edge[bend right, below]  node {$ $} (q_0) 
    (q_1) edge[bend left, above]  node [swap]  {$~~~x_1$} (q_2)
    (q_2) edge[bend right, below]  node  [swap] {$ $} (q_1)          
    (q_0) edge[bend right, above]  node  {$ $} (q_2)
 (q_2) edge[bend left, below]  node  {$x_2$} (q_0) 
 (q_3) edge[bend left, above] node{${x_j}$} (q_4)
(q_4) edge[bend left, below] node{${x_j}$} (q_3);
\end{tikzpicture}
\caption{Three and two dimensional representations} \label{23d}
\end{figure}
\end{center}
From equations (\ref{qua3}) and the automaton graph above we read of the two and three dimensional irreducible representation of the quandle dihedral group for $n=3$ with generators $q_a,$ $a\in X$  and relations given in Definition \ref{groupq} ($(X,\triangleright)$ is the dihedral quandle).

The action of the dihedral quandle group, generated by $q_a,$ $a\in {\mathbb Z}_3$ on the five dimensional vector with basis $\big \{\hat b^{(3)}_{12}, \hat b^{(3)}_{13}, \hat b^{(3)}_{1}, \hat b^{(3)}_{2}, \hat b^{(3)}_{3} \big \}$ (see Proposition \ref{dih0} and Example \ref{exaq1}) decomposes into two irreducible representations of dimension three and two, as illustrated in Figure 10. This structure is reflected in the corresponding automaton, which contains two disconnected graphs (Figure 10). Generally, a reducible representation is characterized by such disconnected graphs, where each component corresponds directly to a distinct block in the block-diagonal decomposition. Conversely, irreducible representations yield fully connected graphs. Furthermore, the two-dimensional representation shown in Figure 9 is non-faithful, as all generators \(q_{a}\) (\(a \in {\mathbb Z}_3\)) map to the identical \(2 \times 2\) anti-diagonal matrix, \(\operatorname{antidiag}(1,1)\).

{
This process can be generalized for \(N > 2\) to find higher tensor representations of \(q_{a}\) (\(a \in X\)) by deriving the eigenvalues and eigenstates of the open spin-chain Hamiltonian, \(\mathcal{H}^{\pm }=\underset{j\in [N-1]}{\sum}\left(r_{j}\pm r_{j}^{-1}\right).\) This operator represents specific linear combinations of all length-one words of the braid group, analogous to the Hecke algebra discussion in the previous section. In the special case of the quandle solution on \(X = \{x_1, x_2, \ldots, x_n\}\), recall that the linear maps \(r, r^{-1}: {\mathbb C}^n \otimes {\mathbb C}^n \to {\mathbb C}^n \otimes {\mathbb C}^n\) act on the basis elements as: \(r\ e_{a}\otimes e_{b}=e_{b}\otimes e_{b\triangleright a}.\) and $r^{-1} = r^{T}.$ Specific examples  are detailed in Proposition \ref{dih0} and Examples \ref{exx}, \ref{exaq1}. Under these conditions, the Hamiltonian expands explicitly as:
\begin{equation}{\cal H}^{\pm} = \sum_{j\in [N-1]} \sum_{a,b \in X} \big ( e^{(j)}_{b,a} e^{(j+1)}_{b\triangleright a, b} \pm  e^{(j)}_{a,b} e^{(j+1)}_{b, b\triangleright a} \big ), \label{new}
\end{equation}
where $e^{(j)}_{a,b} = \id^{j-1} \otimes e_{a,b} \otimes \id^{N-j}.$
The spectral decomposition of this new class of Hamiltonians (\ref{new}) for arbitrary \(N > 2\) is considerably more complex and will be investigated separately. Nonetheless, we observe that \(\mathcal{H}^{\pm }\) are (skew-)symmetric matrices and are therefore unitarily diagonalizable. In general, these \(r\)-matrices satisfy a higher-order condition, \(r^m = \id\), \(m > 2\). Consequently, several critical challenges—such as Baxterisation, the derivation of associated quantum algebras, extending the findings in \cite{DoRySt}, and the explicit construction of local integrable Hamiltonians of type (\ref{new})—remain highly significant and will be addressed in future work.}

\subsection*{Acknowledgments}
\noindent 
I am indebted to M.V. Lawson for numerous illuminating discussions on finite state automata.
I am also thankful to D. Johnston and F. Tesolin for useful discussions and comments.


\begin{thebibliography}{99}

\bibitem{Andru}
N Andruskiewitsch, HJ Schneider, {\it Pointed Hopf algebras,}
New directions in Hopf algebras 43:1-68, 2002.

\bibitem{Bethe1} 
H. Bethe, {\it Eigenwerte und Eigenfunktionen Atomkette, Zeitschrift fur Physik,} 71:205–226, 1931.


\bibitem{Comb-book}
A. Bjorner and F. Brenti, {\it Combinatorics of Coxeter Groups,} Graduate Texts in Mathematics 231, Springer, 2000.

\bibitem{crystal}  
D. Bump and  A. Schilling, {\it Crystal Bases: Representations and Combinatorics,} World Scientific, 2017.

\bibitem{JeOk}
F.~Ced{\'o}, E.~Jespers, and J.~Okni{\'n}ski.
\newblock Braces and the {Yang}-{Baxter} equation.
\newblock {\em Commun. Math. Phys.}, 327(1):101-116, 2014.

\bibitem{Lechner}
R. Correa da Silva and G. Lechner, {\it Modular Structure and Inclusions of Twisted Araki-Woods
Algebras}, Commun. Math. Phys. 402:2339–2386, 2023.

\bibitem{Paths}
E. Date, M. Jimbo, A. Kuniba, T. Miwa and M. Okado, Paths, {\it Maya Diagrams and representations of $\widehat{sl}(r, C)$,} 
Advanced Studies in Pure Mathematics 19:149-191, 1993.

\bibitem{Deho}
P. Dehornoy, {\it Braids and Self-Distributivity}, Progress in Mathematics, volume 192, Birkh\"auser,
Springer Basel, 2012.

\bibitem{Dicke1} 
R.H. Dicke, {\it Coherence in Spontaneous Radiation Processes,} Phys. Rev. 93:99–110, 1954.

\bibitem{Doikous}
A. Doikou, {\it From affine Hecke algebras to boundary symmetries,} Nucl. Phys. B725:493-530, 2005.

\bibitem{Doikoutw}
A. Doikou, {\it Set-theoretic Yang-Baxter equation, braces and Drinfel'd twists}, J. Phys.  A: Math, Theor. 54, 41, 2021.

\bibitem{DoiNep}
A. Doikou and R.I. Nepomechie, {\it Duality and quantum-algebra symmetry of the $A_{N-1}^{(1)}$ open spin chain with diagonal boundary fields},
 Nucl.Phys. B530:641-664, 1998.
 
\bibitem{DoRySt} 
A. Doikou, B. Rybo{\l}owicz and P. Stefanelli, 
{\it Quandles as pre-Lie skew braces, set-theoretic Hopf algebras $\&$ universal R-matrices}, J. Phys. A: Math. Theor. 57, 405203, 2024.

\bibitem{DoiSmo2}
A. Doikou and A. Smoktunowicz, 
{\it Set-theoretic Yang-Baxter $\&$ reflection equations and quantum group symmetries}, Lett. Math. Phys. 111, 105, 2021.

\bibitem{Drinfeld}
V.G. Drinfeld, {\em  Hopf algebras and the quantum Yang–Baxter equation}, Sov. Math. Dokl. 32:254, 1985.

\bibitem{Drinfeld2}
V.G. Drinfeld, {\it Quantum Groups}, in Proc. of the I.C.M. Berkeley, 1986.

\bibitem{Drin}
    V.G. Drinfeld, {\it On some unsolved problems in quantum group theory}, in: Quantum groups 
(Leningrad, 1990), vol. 1510 of Lecture Notes in Math., Springer, Berlin, pp. 1–8, 1992.

\bibitem{Rev3} 
M. Elhamdadi, {\it A Survey of Racks and Quandles: Some Recent Developments},  
Algebra Colloquium, 27, 03:509-522, 2020.

\bibitem{Eti}
    P. Etingof, T. Schedler and A. Soloviev, 
    \emph{Set-theoretical solutions to the Quantum
    Yang-Baxter equation}, Duke Math. J. 100(2):169–209, 1999.

\bibitem{FRT}
L.D. Faddeev, N.Yu. Reshetikhin and L.A. Takhtajan, {\em Quantization of Lie groups and Lie
algebras},  Leningrad Math. J. 1:193, 1990.

\bibitem{Bethe2}
 L. D. Faddeev, E. K. Sklyanin, and L. A. Takhtajan, {\it Quantum Inverse Problem. I,} Theor.
Math. Phys. 40:688–706, 1979.

\bibitem{Fulton} 
W. Fulton, {\it Young Tableaux: with Applications to Representation Theory and Geometry,} Cambridge University Press, 1997.

\bibitem{FultonHarris}
W. Fulton and J. Harris, {\it Representation theory: a first course,} Graduate Texts in Mathematics, Vol.129, Springer, 2004.

\bibitem{Kuniba0}
G. Hatayama, A. Kuniba and T. Takagi, {\it Soliton cellular automata associated with crystal bases,} Nucl. Phys. B
577:619-645, 2000.

\bibitem{Gat1} 
T. Gateva-Ivanova, {\it Set-theoretic solutions of the Yang–Baxter equation, braces and symmetric
groups}, Adv. Math. 388(7):649–701, 2018.

\bibitem{GuaVen} 
L. Guarnieri, L. Vendramin, \emph{Skew braces and the Yang-Baxter equation}, 
Math. Comput. 86(307):2519–2534, 2017.

\bibitem{Jimbo} 
M. Jimbo, {\it A q-difference analogue of U(g) and the Yang-Baxter equation,} Lett. Math. Phys. 10:63-69, 1985.

 \bibitem{Jimbo2}
 M. Jimbo, {\it Quantum R-matrix for the generalized Toda system,} Comm. Math. Phys. 102:537–547, 1986.

\bibitem{Jo82} 
D.  Joyce, \emph{A classifying invariant of knots, the knot quandle,} {J. Pure Appl. Algebra}, 23, 1:37--65, 1982.

\bibitem{Kac}
V.G. Kac, {\it Infinite dimensional Lie algebras,} 3rd edition, Cambridge Univ. Press. Cambridge, 1990.

\bibitem{MisMiw} 
S.J. Kang , K.C. Misra and T. Miwa, {\it Fock Space Representations of the Quantized Universal Enveloping Algebras $U_q(C^{(1)}_l)$, $U_q(A^{(2)}_{2l})$, and $U_q(D^{(2)}_{l+1})$},  Journal of Algebra 155, Issue 1:238-251, 1993.

\bibitem{Kashi}
M. Kashiwara, {\it On crystal bases of the $q$-analogue of universal enveloping algebras,} Duke Math. J. 63:465–516, 1991.

\bibitem{KirRes1}
S. V. Kerov, A. N. Kirillov and N. Yu. Reshetikhin, {\it Combinatorics, the Bethe ansatz and representations 
of the symmetric group,} J. Soviet Math. 41: 916–924, 1988. 

\bibitem{KirRes2}
A. N. Kirillov and N. Yu. Reshetikhin, {\it The Bethe ansatz and the combinatorics of Young tableaux,} J. Soviet Math. 41:925–955, 1988.

\bibitem{Kir3}
Kirillov A.N., Schilling A., Shimozono M., {\it A bijection between Littlewood-Richardson tableaux and rigged
configurations,} Selecta Math. (N.S.) 8:67-135, 2002.

\bibitem{Kleene}
S.C. Kleene, {\it Representation of events in nerve nets and finite automata,}
in C.E. Shannon and J. McCarthy, editors, {\em Automata Studies}, number~34 in Annals of Mathematics Studies, pages 3-40. 
Princeton University Press, 1956.

\bibitem{quantum1} 
A. Kondacs and J. Watrous, {\it On the power of quantum finite state automata,} 
in 38th IEEE Conference on Foundations of Computer Science, 66–75, 1997.

\bibitem{Kulish}
P.P. Kulish and E.K. Sklyanin,  {\it The general $U_q(sl(2))$ invariant XXZ integrable quantum spin chain} J. Phys. A24:L435, 1991. 

\bibitem{Kuniba1}
A. Kuniba, M. Okado, and Y. Yamada, {\it Box-ball system with reflecting end,} 
Journal of Nonlinear Mathematical Physics, 12(4):475–507, 2005.

\bibitem{Kuniba2}
A. Kuniba, M. Okado, R. Sakamoto, T. Takagi, Y. Yamada, 
{\it Crystal interpretation of Kerov-Kirillov-Reshetikhin bijection,} Nucl. Phys. B740:299-327, 2006.

\bibitem{Lawson}
M. Lawson, {\it Finite Automata}, Chapman and Hall/CRC, 2004.

\bibitem{Rev1} 
V. Lebed, {\it Applications of self-distributivity to Yang–Baxter operators and their cohomology}, 
Journal of Knot Theory and Its RamificationsVol. 27, No. 11, 1843012, 2018.

\bibitem{LebVen} 
V. Lebed and L. Vendramin, 
    \emph{Reflection equation as a tool for studying solutions to the Yang-Baxter equation}, 
    J. Algebra, 607:360-380, 2022.

\bibitem{Lusztig}
G.Lusztig, {\it Canonical bases arising from quantized enveloping algebras,} J. Amer. Math.
Soc. 3:447-498, 1990.

\bibitem{Matsu}
H. Matsumoto, {\it G{\'e}n{\'e}rateurs et relations des groupes de Weyl g{\'e}n{\'e}ralis{\'e}s,} 
C.R. Acad. Sci. Paris 258:3419–3422, 1964.

\bibitem{Matv}
S. V. Matveev, {\it Distributive groupoids in knot theory}, Math. USSR Sb. 47:73-83, 1984.

\bibitem{MeNe}
L. Mezincescu and R.I. Nepomechie, {\it Analytical Bethe Ansatz for quantum-algebra-invariant
spin chains,}   Nucl. Phys. B372:597—621, 1992.

\bibitem{quantum2}
C. Moore and J. P. Crutchfield. Quantum automata and
quantum grammars. Theor. Comp. Sci., 237:275–306, 2000.

\bibitem{Nepo1} 
R.I. Nepomechie and D. Raveh, {\it Qudit Dicke state preparation,} Quant. Inf. Comput. 24, 1-2, 0037-0056, 2024.

\bibitem{Pasquier}
V. Pasquier and H. Saleur,  {\it Common structures between finite systems and conformal field theories through quantum groups,}  Nucl. Phys. B330:523-556, 1990.

\bibitem{Dandecy}
L. Poulain d'Andecy, {\it Centralisers and Hecke algebras in Representation Theory, with applications to Knots and Physics,} arXiv:2304.00850 [math.RT].

\bibitem{RabinScott}
M.O. Rabin and D.Scott,
{\it Finite automata and their decision problems,}
IBM Journal of Research and Development, 3:114--125, 1959.
\newblock Reprinted in E.~F.~Moore, editor, {\em Sequential Machines: Selected
  Papers}, Addison-Wesley, 1964.

\bibitem{Rabin}
M.O. Rabin. {\it Probabilistic automata}, Information and
Control, 6:230–245, 1963.

\bibitem{Nepo2}
D. Raveh and R.I. Nepomechie, {\it $q$-analog qudit Dicke states,} J. Phys. A: Math. Theor. 57 065302, 2024.

\bibitem{Ru05}
 W. Rump, \emph{A decomposition theorem for square-free unitary 
solutions of the quantum Yang–Baxter equation}, Adv. Math. 193:40–55, 2005.
    
\bibitem{Ru07}
W. Rump, \emph{Braces, radical rings, and the quantum Yang-Baxter equation}, J. Algebra 307 (1):153–170, 2007.

\bibitem{Rev2} 
J. Scott Carter, {\it A Survey of Quandle Ideas}, Introductory Lectures on Knot Theory, pp. 22-53, 2011. 

\bibitem{Sklyanin}
 E.K. Sklyanin, {\it Boundary conditions for integrable quantum systems,} J. Phys. A21:2375, 2375, 1988.
 
 \bibitem{Sol}
A. Soloviev, {\it Non-unitary set-theoretical solutions to the quantum Yang-Baxter equation,}
Math. Res. Lett. 7:577-596, 2000.

\bibitem{Tzeng}
W. Tzen,. {\it A polynomial-time algorithm for the equivalence of probabilistic automata,}
SIAM Journal on Computing, 21(2):216–227, 1992.

\bibitem{review}
K. Wiesner, J.P. Crutchfield
Physica D: Nonlinear Phenomena
Volume 237, Issue 9:1173-1195, 2008.



\end{thebibliography}
\end{document}